\documentclass[11pt,letterpaper]{article}

\usepackage{MathStyle}
\usepackage{microtype}
\usepackage{needspace}
\usepackage{placeins}

\DeclareMathOperator{\SAT}{SAT}
\DeclareMathOperator{\ALG}{ALG}
\newtheorem{proposition}{Proposition}
\crefname{appendix}{Appendix}{Appendices}
\Crefname{appendix}{Appendix}{Appendices}
\hypersetup{pdftitle={Online Discrepancy for Sub-Exponential Product Distributions via Regularization},pdfauthor={Nicola Wengiel}}

\title{Online Discrepancy for Sub-Exponential Product Distributions\\
via Regularization}
\author{Nicola Wengiel\thanks{Email: \href{mailto:nicola.wengiel@fau.de}{nicola.wengiel@fau.de}.}}
\affil{Friedrich-Alexander-Universit\"at Erlangen-N\"urnberg\\
Department of Mathematics, Erlangen, Germany}
\date{}

\begin{document}
\maketitle

\begin{abstract}
We study online discrepancy minimization for independent random vectors in $\R^n$ with independent coordinates.
Our deterministic algorithms require no knowledge of the horizon, and their guarantees hold at every prescribed deterministic terminal time $T$.
For unit-variance sub-exponential coordinates with $\psi_1$-norms at most $\sigma$, we obtain terminal discrepancy $O(\sigma^4\sqrt n)$ with failure probability $\exp(-\Omega(\sigma^3\sqrt n))$.
Independently masking centered unit-variance marks from the same class with retention probability $k/n$ gives $O(\sigma^4\sqrt k)$ discrepancy with failure probability $\exp(-\Omega(\sigma^3\sqrt k))$ when $k\gtrsim\log^2(2n)$.
The fourth-power dependence on $\sigma$ is necessary for these sub-exponential classes at the stated confidence levels.

\smallskip
For standard Gaussian marks, we obtain $O(\sqrt k+\sqrt{\log(2n)})$ discrepancy with failure probability $\exp(-\Omega(\sqrt k+\sqrt{\log(2n)}))$.
Both terms in this discrepancy bound are necessary up to absolute constants for $T\ge n$.
In particular, $k\gtrsim\log n$ is the optimal sparsity threshold for $O(\sqrt k)$ discrepancy.
For general centered unit-variance sub-Gaussian marks, complementary bounds improve the joint dependence on sparsity and tails.
For every fixed $\psi_2$-norm bound $\sigma$, logarithmic sparsity suffices for $O(\sigma\sqrt{k\log(e+\sigma)})$ discrepancy with probability tending to one as $n\to\infty$.
For a dimension-dependent family of unit-variance laws with rare activations, matching bounds determine the fixed-confidence online discrepancy order $\Theta(\sqrt n\log\log n)$ at sufficiently long prescribed horizons.

\smallskip
Our dense result generalizes the Rademacher-input theorem of Bansal and Spencer (2020) and establishes the horizon-independent $O(\sqrt n)$ Gaussian terminal discrepancy scale anticipated by Gamarnik et al.\ (2022).
The proofs combine regularization, global Fenchel estimates, coupled Lyapunov functions, and terminal moment recursions.
\end{abstract}

\clearpage
\tableofcontents
\clearpage

\section{Introduction}
\label{sec:Introduction}

Given vectors $v_1,\ldots,v_T\in\R^n$, discrepancy minimization seeks signs $x_1,\ldots,x_T\in\{\pm1\}$ that minimize the $\ell_\infty$-norm of the signed sum $\sum_{t=1}^T x_tv_t$. The minimum achievable norm is the discrepancy of the input sequence. We also use discrepancy for the norm attained by a particular signing. When the $v_t$ are columns of a set-system incidence matrix, the signs give a two-coloring of the ground set, and the signed sum records the color imbalance in each set. The \emph{terminal discrepancy} is the $\ell_\infty$-norm of the signed sum at a prescribed final time $T$, while the stronger \emph{prefix discrepancy} objective controls the largest such norm over all prefixes.

In the online setting, the vectors are revealed one at a time, and each sign $x_t$ must be chosen immediately after observing $v_t$, without knowledge of future vectors. We study terminal discrepancy in this setting and use discrepancy to mean terminal discrepancy throughout, unless stated otherwise.

For offline signing, Spencer \cite{Spencer85} proved that any $n$ vectors in $[-1,1]^n$ have discrepancy $O(\sqrt n)$, a bound that is tight up to an absolute constant, as witnessed by the columns of a Hadamard matrix. His proof was non-constructive. More than two decades later, Bansal \cite{Bansal10} obtained an algorithmic version, followed by further algorithmic developments \cite{Lovett15,Levy17,Rothvoss17,Eldan18,Pesenti23}; see also Bansal's survey \cite{Bansal23}.

The Koml\'os problem asks whether vectors in $\R^n$ of Euclidean norm at most one always admit a signed sum with $\ell_\infty$-norm bounded by an absolute constant. Banaszczyk \cite{Banaszczyk98} established an $O(\sqrt{\log(2n)})$ bound, and Bansal and Jiang \cite{Bansal26} obtained the first asymptotic improvement. Guo, Fang and Lu proved a constant bound \cite{Guo26Existence} and gave a deterministic offline algorithm using polynomially many real arithmetic operations \cite{Guo26Algorithm}. Anand and Chen \cite{AnandChen26} gave a faster randomized algorithm with expected bit-complexity bounds for rational inputs. These results allow the signs to depend on the entire input sequence.

For worst-case online inputs, the situation is different. An adaptive adversary can choose each arriving vector perpendicular to the current signed sum $d_{t-1}\colonequals\sum_{s<t}x_sv_s$, so that $\norm{d_t}_2^2=\norm{d_{t-1}}_2^2+\norm{v_t}_2^2$ regardless of the chosen sign. For vectors in $[-1,1]^n$, this strategy forces terminal discrepancy $\Omega(\sqrt T)$ when $n\ge2$.\footnote{An extreme point of $[-1,1]^n\cap d_{t-1}^{\perp}$ has at least $n-1$ coordinates in $\{\pm1\}$.} In the square regime $T=n$, the adversarial online problem has discrepancy of order $\sqrt{n\log n}$; see Bansal and Spencer \cite[Section~1.1]{Bansal20a}. These lower bounds rely on adapting each vector to the preceding signs. Can independent random arrivals instead admit a discrepancy bound that depends on the dimension but not on the prescribed horizon?

\paragraph{Random inputs and the horizon-independent scale.}
Independent random arrivals can help an online algorithm counteract its accumulated imbalance. A simple calculation for Rademacher inputs illustrates this mechanism and explains why controlling the Euclidean norm alone is insufficient to obtain $O(\sqrt n)$ discrepancy.

Condition on the current state $d\colonequals d_{t-1}$. For a fresh vector $v_t$ with independent Rademacher coordinates, Khintchine's inequality gives $\E|\inp{d}{v_t}|\gtrsim\norm d_2$. Choose $x_t$ to minimize $x_t\inp{d}{v_t}$, breaking ties arbitrarily. Then
\begin{equation*}
    \E{\norm{d+x_tv_t}_2^2-\norm d_2^2}
    =n-2\E{\abs{\inp{d}{v_t}}}
    \le n-c_0\norm d_2,
\end{equation*}
where $c_0>0$ is an absolute constant. The expected change in the squared Euclidean norm is negative once $\norm d_2$ exceeds a sufficiently large constant multiple of $n$, suggesting a natural Euclidean scale of order $n$.

\Needspace{8\baselineskip}
A Euclidean norm of order $n$, however, still permits individual coordinates of order $n$. To achieve $\ell_\infty$-discrepancy of order $\sqrt n$, an algorithm must also control coordinate outliers. Bansal and Spencer \cite{Bansal20a} addressed this difficulty using potential functions and, for arbitrary horizons, a hybrid strategy.\footnote{For an extended comparison of Bansal and Spencer's signing strategies, see \cite[Section~1.2 and Table~2]{Wengiel26v1}.} They obtained the following terminal guarantee.

\begin{literaturetheorem}[{\cite[Theorem~3.4]{Bansal20a}}]
    \label{thm:BansalSpencer}
    Fix a positive integer $n$. There is a polynomial-time online algorithm depending only on $n$ with the following guarantee. Let $v_1,v_2,\ldots\in\R^n$ be independent vectors with independent Rademacher coordinates. For every prescribed deterministic terminal time $T\in\N$, the signs $x_1,x_2,\ldots\in\{\pm1\}$ chosen by the algorithm satisfy
    \begin{equation*}
        \P{\norm{\sum_{t=1}^T x_tv_t}_\infty>C\sqrt n}
        \le \exp{-c\sqrt n},
    \end{equation*}
    where $c,C>0$ are absolute constants.
\end{literaturetheorem}

At $T=n$, independent random signs typically give discrepancy of order $\sqrt{n\log n}$. Bansal and Spencer's theorem removes this logarithmic loss and achieves the $\sqrt n$ scale at every prescribed horizon. This order is optimal even offline at $T=n$ \cite{Bansal20a}; subsequent work \cite[Theorem~10]{Gamarnik23} established its optimality among online algorithms for every prescribed $T\ge n$.

Gaussian inputs highlight the distinction between online and offline signing. In the proportional regime $T=\Theta(n)$, the optimal offline discrepancy is $\Theta(\sqrt n)$ with high probability \cite{Perkins21,Abbe22a,Altschuler23Critical}. When $n\to\infty$ and $n=o(T)$, Turner, Meka and Rigollet \cite{Turner20} showed that it instead has order $\sqrt T\,2^{-T/n}$; see \cite{Costello09} for fixed dimension. Thus, the offline discrepancy becomes $o(\sqrt n)$ as $T/n\to\infty$.

Online, by contrast, the Gaussian lower bound of Gamarnik, K{\i}z{\i}lda\u{g}, Perkins and Xu \cite[Theorem~9]{Gamarnik23} forces the $\sqrt n$ scale in the proportional regime even for arbitrarily large fixed ratios $T/n$, regardless of running time. Can Bansal and Spencer's horizon-independent guarantee be extended to unbounded product inputs, in particular Gaussian inputs? Such an extension must accommodate occasional large arrivals.

\subsection{Main results}
\label{sec:MainResults}

Our first result extends the horizon-independent terminal guarantee to product inputs with uniformly bounded sub-exponential tails. Neither centering, symmetry nor identical coordinate distributions are required, and the laws may vary with time. We use the Orlicz norms $\norm{\cdot}_{\psi_1}$ and $\norm{\cdot}_{\psi_2}$ for sub-exponential and sub-Gaussian random variables, respectively, with the conventions in \eqref{eq:OrliczNorms}. All unspecified constants below are positive and absolute.

\begin{theorem}[Dense inputs]
    \label{thm:OnlineDiscrepancy}
    Fix a positive integer $n$ and a bound $\sigma>0$. There is a deterministic online algorithm depending only on $n$ and $\sigma$ with the following guarantee. Let $v_1,v_2,\ldots\in\R^n$ be independent vectors with independent coordinates of variance one and $\psi_1$-norms at most $\sigma$. For every prescribed deterministic terminal time $T\in\N$, the signs $x_t\in\{\pm1\}$ chosen by the algorithm satisfy
    \begin{equation*}
        \P{\norm{\sum_{t=1}^T x_tv_t}_\infty>C\sigma^4\sqrt n}
        \le \exp{-c\sigma^3\sqrt n}.
    \end{equation*}
    Here $c,C>0$ are absolute constants. The algorithm uses $O(n\log(2n\sigma))$ real arithmetic operations per arrival.
\end{theorem}

For fixed tail parameters, the theorem gives $O(\sqrt n)$ terminal discrepancy for standard Gaussian, Rademacher, and centered unit-variance Laplace entries, as well as for bounded centered unit-variance product distributions. The same parameter dependence applies to sub-Gaussian inputs, since a $\psi_2$-norm bound $\sigma$ implies a $\psi_1$-norm bound $O(\sigma)$. The theorem extends the Rademacher-input result of Bansal and Spencer \cite{Bansal20a}. In particular, it attains the horizon-independent $O(\sqrt n)$ Gaussian terminal discrepancy scale anticipated by Gamarnik, K{\i}z{\i}lda\u{g}, Perkins and Xu \cite{Gamarnik22}, with failure probability $\exp(-\Omega(\sqrt n))$.

\Cref{prop:SigmaSharpness} shows that the $\sigma^4$ dependence is necessary for the full sub-exponential class at this confidence level. This sharpness concerns the joint discrepancy--confidence guarantee. With error $e^{-c\sqrt n}$, centered inputs instead admit $O(\sigma\log(e\sigma)\sqrt n)$ discrepancy; see \cref{thm:DenseConfidence} for the full profiles. We prove \cref{thm:OnlineDiscrepancy} in \cref{sec:OnlineDiscrepancy}, together with the conditional extension in \cref{prop:DenseConditional}, which allows the next product law to depend on the past.

Unit variance alone does not suffice for guarantees with constants uniform over input laws and horizons. For each prescribed $T\ge2$, one can choose a symmetric, centered, unit-variance \iid entry law such that, with probability $1-\exp(-\Omega(n))$, rare large entries force discrepancy at least $\sqrt T$ for every signing. These laws depend on $T$, and their sub-exponential norms grow with the horizon; see \cref{rem:UniformTailControl} for the construction.

\paragraph{Sparse inputs.}
We next ask whether terminal discrepancy can adapt to the sparsity of the arrivals. If the mask activates $k$ coordinates on average, the natural target for fixed tail parameters is $O(\sqrt k)$, replacing the dense $O(\sqrt n)$ scale. The Beck--Fiala conjecture predicted this scale for set systems of degree at most $k$.

We model sparsity by independent Bernoulli masks. Put $p\colonequals k/n$ and write each arriving vector as $v_t\colonequals \xi_t\odot w_t$, where $\odot$ denotes coordinatewise multiplication. The mask coordinates are $\Ber(p)$ variables. The coordinates of $w_t$, which we call marks, are centered, unit-variance sub-exponential random variables. All mask and mark variables are independent.

The size of the mask support has distribution $\Bin(n,p)$ and expectation $k$; it is not deterministically bounded by $k$. Unit variance refers to the unmasked marks; each masked coordinate has variance $p$. The marks may be unbounded, and their distributions may differ across coordinates and time.

\begin{theorem}[Sparse inputs]
    \label{thm:SparseOnlineDiscrepancy}
    Fix a positive integer $n$, a parameter $K_0\log^2(2n)\le k\le n$, and a bound $\sigma>0$, where $K_0$ is a sufficiently large absolute constant. Put $p\colonequals k/n$. There is a deterministic online algorithm depending only on $n$, $k$ and $\sigma$ with the following guarantee. Let $w_t(i)$, $t\in\N$ and $1\le i\le n$, be independent centered unit-variance marks with $\psi_1$-norms at most $\sigma$. Let $\xi_t(i)$ be independent $\Ber(p)$ masks, independent of all marks. For every prescribed deterministic terminal time $T\in\N$, the signs $x_t\in\{\pm1\}$ chosen by the algorithm satisfy
    \begin{equation*}
        \P{\norm{\sum_{t=1}^T x_t(\xi_t\odot w_t)}_\infty
             >C\sigma^4\sqrt k}
        \le \exp{-c\sigma^3\sqrt k}.
    \end{equation*}
    Here $c,C>0$ are absolute constants.
\end{theorem}

The algorithms in \cref{thm:OnlineDiscrepancy,thm:SparseOnlineDiscrepancy} require no knowledge of the input laws or the horizon; both use fixed tie-breaking. \Cref{sec:SparseOnlineDiscrepancy} proves \cref{thm:SparseOnlineDiscrepancy} and the sparse conditional extension in \cref{cor:SparseAllLevels}; \cref{sec:implementation} gives implementation details and operation counts.

The dense theorem remains more general even when $k=n$: it permits noncentered coordinates and imposes no dimension restriction. The sparse theorem requires centering and $k\ge K_0\log^2(2n)$. Its fourth-power dependence on $\sigma$ is also necessary at the stated confidence level for the full symmetric sub-exponential class, by \cref{prop:SigmaSharpness}.

A direct application of the dense theorem using a generic $\psi_1$-norm estimate loses powers of $p$ when $p$ is small and does not yield the desired $\sqrt k$ scale. The sparse analysis retains the Bernoulli parameter in its gradient and moment estimates; the loss from the generic estimate is computed in \cref{sec:SparseOnlineDiscrepancy}. The same analysis applies to every $1\le k\le n$, with an additive dimension correction; the condition $k\gtrsim\log^2(2n)$ above allows this correction to be absorbed into the stated $O(\sigma^4\sqrt k)$ bound.

\paragraph{Confidence profiles and refinements.}
We next make the dependence on the target failure probability explicit.
Recall that $d_T=\sum_{t=1}^T x_tv_t$.
For the centered sub-exponential marks considered above and every $1\le k\le n$,
\cref{thm:SparseConfidence} gives a fixed deterministic rule, independent
of the horizon and confidence parameter, such that
\[
 \P\left(\norm{d_T}_\infty>
 C\sigma[(\log(e\sigma))^2\sqrt k+\log(2n)+r]\right)\le e^{-r}
\]
at every prescribed deterministic $T\ge1$ and every $r\ge0$.
The full theorem also covers sub-Gaussian marks, allows conditional
product laws depending on the past, and provides a second rule with a
different trade-off between tail dependence and confidence.
Under the stronger sub-Gaussian assumption, the following result reduces
the dimension correction to square-root logarithmic order.

\begin{theorem}[Sub-Gaussian sparse inputs]\label{thm:SparseRefinements}
Fix a positive integer $n$, a parameter $1\le k\le n$, and $\sigma>0$,
and put $p\colonequals k/n$. Consider arrivals $v_t\colonequals \xi_t\odot w_t$.
At each arrival, conditionally on the complete past $\Fc_{t-1}$,
suppose that $w_t$ has independent centered unit-variance coordinates
with $\psi_2$-norms at most $\sigma$, and that $\xi_t$ has independent
$\Ber(p)$ coordinates and is independent of $w_t$.
These assumptions hold almost surely.
There is a deterministic online algorithm depending only on $n,k,\sigma$
such that, uniformly over these input laws, for every prescribed
deterministic integer $T\ge0$ and every $r\ge0$,
\begin{equation}\label{eq:SparseSubGaussianIntro}
 \P\left(\norm{d_T}_\infty>
 C\sigma\sqrt{\log(e+\sigma)}
 \left[\sqrt k+\sqrt{\log(en/k)}+r\right]\right)\le e^{-r}.
\end{equation}
Here $C>0$ is an absolute constant. The algorithm uses fixed tie-breaking
and requires no knowledge of the input laws, the horizon, or the
confidence parameter.
\end{theorem}

For every fixed $\sigma$, taking $r=c\sqrt k$ shows that
$k\gtrsim\log n$ suffices for $O_\sigma(\sqrt k)$ discrepancy with
error $e^{-c\sqrt k}$. For standard Gaussian marks, the tail parameter
is an absolute constant, giving the following specialization.

\begin{corollary}[Gaussian marks]\label{cor:SparseGaussian}
Fix a positive integer $n$ and a parameter $1\le k\le n$, and put $p\colonequals k/n$.
Let all marks be independent standard Gaussian variables and,
independently, let all masks be independent $\Ber(p)$ variables.
There is a deterministic online algorithm depending only on $n,k$ such
that, for every prescribed deterministic integer $T\ge0$ and every $r\ge0$,
\begin{equation}\label{eq:SparseGaussianIntro}
 \P\left(\norm{d_T}_\infty>C\left[\sqrt k+\sqrt{\log(en/k)}+r\right]\right)
 \le e^{-r}.
\end{equation}
Here $C>0$ is an absolute constant. The algorithm uses fixed tie-breaking
and requires no knowledge of the horizon or confidence parameter.
\end{corollary}

The corollary follows from \cref{thm:SparseRefinements} with $\sigma=2$,
since $\E(e^{G^2/4})=\sqrt2\le2$ for $G\sim N(0,1)$.
For Gaussian marks, $\sqrt k+\sqrt{\log(en/k)}$ and
$\sqrt k+\sqrt{\log(2n)}$ are equivalent up to absolute factors.
By~\cref{cor:GaussianDimensionLower}, \cref{cor:SparseGaussian}
has the optimal dimension and sparsity order at every prescribed $T\ge n$.
With $r=c\sqrt k$ and $k\asymp\log n$, the failure bound tends to zero,
but more slowly than any fixed negative power of $n$.
We prove \cref{thm:SparseRefinements} in \cref{sec:SparseGaussianUpper};
\cref{sec:implementation} gives implementation details and operation counts.

\Cref{thm:SparseEnvelope} in \cref{sec:SparseSubGaussian} gives
complementary bounds that keep the coefficient of the leading $\sqrt k$
term independent of $\sigma$ or improve the joint parameter dependence
when $\sigma$ grows with $n$. The rare-activation examples in
\cref{prop:RareActivationLower} show that adding the separately necessary
scales in \cref{prop:JointSubGaussianLower} does not yield a uniform
upper bound.

For a concrete rare-activation law, suppose all input entries are independent
copies of $W$ with $\P(W=0)=1-1/n$ and
$\P(W=\sqrt n)=\P(W=-\sqrt n)=1/(2n)$.
Together, \cref{prop:RareActivationLower,prop:RareActivationUpper} determine
the fixed-confidence online discrepancy order $\Theta(\sqrt n\log\log n)$
at sufficiently long prescribed horizons.
The matching upper bound is attained by a deterministic priority rule using
$O(n)$ operations per arrival. This resolves the order for this law;
the uniform question over the full admissible sub-Gaussian class remains open.

For the full sub-exponential class, the $\log^2 n$ scale remains
necessary for a pure $O(\sqrt k)$ bound with fixed tail parameter,
even offline. For unit-variance Laplace marks, take $T=\floor{n/k}$.
When $k=o(\log^2 n)$, a constant fraction of coordinates have exactly
one hit with high probability. Their magnitudes are unaffected by
signing, and their maximum has order $\log n$.

\Needspace{8\baselineskip}
\paragraph{Terminal times and optimality.}
The constants in each guarantee apply at every prescribed deterministic terminal time, and none of the algorithms requires knowledge of the horizon. These guarantees do not assert that the same discrepancy bound holds simultaneously at all times. A union bound over the first $T$ arrivals introduces a factor $T$ in the failure probability. For Gaussian inputs, unbounded coordinate jumps preclude a fixed discrepancy bound valid throughout an infinite stream.

Gamarnik, K{\i}z{\i}lda\u{g}, Perkins and Xu \cite[Theorem~9]{Gamarnik23} established a Gaussian online obstruction in the proportional small-margin regime. The following theorem gives an explicit finite bound valid for every prescribed $T\ge n$, together with an extension to Bernoulli-masked Gaussian inputs. Both bounds are proved in \cref{sec:BlockLowerBounds}.

\begin{theorem}[Gaussian online lower bounds]
    \label{thm:GaussianLowerBound}
    Fix a positive integer $n$ and a prescribed deterministic terminal time $T\in\N$ with $T\ge n$. The following bounds hold for every randomized online signing rule.
    \begin{enumerate}[label=(\roman*),ref=(\roman*),leftmargin=*]
        \item \label{item:GaussianLowerBoundDense} Dense inputs. For independent vectors $v_1,\ldots,v_T\sim N(0,I_n)$, the chosen signs satisfy
        \begin{equation}
            \label{eq:GaussianLowerBound}
            \P{\norm{\sum_{t=1}^T x_tv_t}_\infty\le\frac{\sqrt n}{2}}
            <0.767^n.
        \end{equation}
        \item \label{item:GaussianLowerBoundSparse} Bernoulli-masked inputs. Let $16\le k\le n$ and put $p\colonequals k/n$. Let the marks $w_t(i)$ be independent standard Gaussian variables and, independently, let the masks $\xi_t(i)$ be independent $\Ber(p)$ variables. Then
        \begin{equation}
            \label{eq:SparseGaussianLowerBound}
            \P{\norm{\sum_{t=1}^T x_t(\xi_t\odot w_t)}_\infty\le\frac{\sqrt k}{2}}
            <0.835^n.
        \end{equation}
    \end{enumerate}
    Both bounds remain valid when the final $n$ signs are chosen jointly after observing that suffix.
\end{theorem}

For dense and Bernoulli-masked standard Gaussian inputs, \cref{thm:GaussianLowerBound} gives matching lower bounds for the dense $\sqrt n$ scale and the sparse $\sqrt k$ scale, respectively. Thus, the dimension and sparsity dependence in \cref{thm:OnlineDiscrepancy,thm:SparseOnlineDiscrepancy} is optimal up to constants at prescribed times $T\ge n$, throughout their stated ranges. The lower bounds apply without any restriction on the running time of the signing rule. For conditional computational barriers to fully offline signing under lattice-hardness assumptions, see Vafa and Vaikuntanathan \cite{Vafa25} and \cref{sec:OpenQuestions}.

\Cref{cor:GaussianAllLevels} gives a Gaussian lower bound for all $k\ge1$, and \cref{cor:ThirdMomentLower} extends the block argument to independent centered unit-variance marks with a uniform third-moment bound. For independent copies of any fixed unit-variance mark law, \cref{prop:FiniteVarianceLower} gives the same $\sqrt k$ obstruction at every prescribed $T\ge n$ once $k$ exceeds a distribution-dependent threshold. For Laplace marks, a shorter single-hit block gives a necessary threshold of order $\log n+r$, up to an additive absolute constant, at error $e^{-r}$ for $r\ge\log2$, by \eqref{eq:LaplaceNecessaryConfidence}; \cref{sec:TailSharpness} also proves the $\sigma$-sharpness statement.

The main remaining questions concern the possible reduction of the Gaussian confidence cost from $r$ to square-root order and the optimal joint dependence on dimension, sparsity and tail bounds for general sub-Gaussian marks. We discuss these questions in \cref{sec:OpenQuestions}.

\paragraph{Proof strategy.}
Global Fenchel estimates and terminal moment recursions organize the proofs. The dense rule uses the current gradient. The rules underlying the basic centered sparse bounds and one sub-Gaussian refinement minimize a coupled Lyapunov function; the drift estimates follow from comparison with an auxiliary randomized threshold rule. The other sub-Gaussian refinements use a different regularizer and control rare simultaneous mask hits. \Cref{sec:TechnicalOverview} explains these mechanisms and their shared estimates.

\paragraph{Relation to the first version.}
This paper revises and extends arXiv:2608.10040v1 \cite{Wengiel26v1},
simplifying the proofs and resolving several questions posed there.
We extend the horizon-independent terminal bounds from sub-Gaussian to
sub-exponential product inputs and remove the symmetry assumption.
The sparse result retains centering; the dense result requires neither
centering nor a dimension restriction. With the original $\psi_2$ parameter,
the main bounds improve from $\sigma^8$ to $\sigma^4$ while retaining the
earlier confidence exponents. For centered sub-Gaussian marks with any fixed
tail bound, $k\gtrsim\log n$ now suffices for $O_\sigma(\sqrt k)$ discrepancy
with probability tending to one as $n\to\infty$, improving the previously
sufficient squared-logarithmic scale. Complementary parameter choices also
improve the joint dependence on sparsity and tails. We strengthen the
finite-dimensional lower bounds, extend them to broader input distributions,
and show that adding the separate necessary scales does not yield a uniform
upper bound. The first version already established an $O(\sqrt n)$ Gaussian
terminal discrepancy bound at every prescribed horizon.

\subsection{Comparison with previous work}
\label{sec:RelatedWork}

Related results differ in their input assumptions, normalization, discrepancy objective, and dependence on the horizon. We first discuss proportional online asymptotics, followed by stochastic and oblivious prefix guarantees, sparse inputs, and offline random ensembles. \Cref{tab:MainResultsComparison} summarizes the online guarantees and their assumptions.

\paragraph{Proportional asymptotics and distributional universality.}
Fiedler, Jackson, Lacker and Niles-Weed \cite[Theorem~1.1 and equation~(1.4)]{Fiedler26} characterized the optimal expected terminal discrepancy achievable online for standard Gaussian inputs in the proportional regime. If $T/n\to r\in(0,\infty)$, this optimum, normalized by $\sqrt n$, converges to a constant $V_r^\infty$ described by a mean-field stochastic control problem. For \iid centered unit-variance entries with finite fourth moment, they established a corresponding liminf lower bound for the normalized optimum.

Jo and Lee \cite[Theorem~1.1]{JoLee26} proved universality of the proportional expected optimum under the same fourth-moment assumption. They also obtained quantitative online upper bounds for targets separated from the optimum by a fixed positive gap \cite[Theorem~1.3]{JoLee26}. For \iid centered unit-variance entries distributed as $X$ with $\E|X|^4\le K$ and each fixed $A>0$, their failure probability is bounded by\footnote{Fix $\alpha>0$ and a target $\kappa>\sqrt\alpha\,V_{1/\alpha}^\infty$ for discrepancy normalized by $\sqrt T$. The ratio $n/T$ must lie in a sufficiently small fixed neighborhood of $\alpha$. The algorithm, constant $C_A$, neighborhood, and minimum horizon depend on $A$, $\alpha$, $\kappa$, and the fourth-moment bound $K$.}
\begin{equation*}
    C_A\bigl[T^{-A}+\E{\abs X^4\I{\abs X>T^{1/3}}}\bigr]
\end{equation*}
in the proportional regime.\footnote{Their algorithms use $O(n)$ arithmetic operations and scalar feedback evaluations per arrival for fixed precomputed control data.} Finite fourth moment makes the tail term tend to zero, but its decay can be arbitrarily slow even for a fixed entry law. Consequently, the displayed bound need not decay polynomially and supplies no explicit exponential failure estimate.

Jo and Lee also extended universality to independent, non-identically distributed centered unit-variance entries whose fourth powers are uniformly integrable across all dimensions and entries \cite[Theorem~6.1]{JoLee26}. Their quantitative bound then involves the corresponding uniform fourth-moment tail. These results identify exact proportional constants under finite-moment assumptions. Our results impose stronger tail assumptions and give explicit exponential terminal guarantees at every prescribed finite horizon.

\paragraph{General stochastic arrivals and oblivious sequences.}
For \iid arrivals from arbitrary distributions on $[-1,1]^n$, Bansal et al. \cite{Bansal20b} obtained $O(n^2\log(nT))$ prefix discrepancy with high probability. Subsequent work \cite{Bansal21} improved this to $O(\sqrt n\log^4(nT))$ and developed guarantees for general norms under additional assumptions. These results allow arbitrary dependence among coordinates. Under the stated variance and tail assumptions, our results treat independent coordinates, allow unbounded inputs and laws that vary with time, and give horizon-independent terminal guarantees. The dense result also allows nonzero means, and both models have extensions with conditional independence given the past.

In the oblivious online setting, the input sequence is fixed independently of the algorithm's random choices and revealed one vector at a time. The guarantees concern the algorithm's internal randomness and require no distributional assumption on the inputs. Approaches include self-balancing and Gaussian fixed-point walks \cite{Alweiss21,Liu22}. For vectors of Euclidean norm at most one, Kulkarni, Reis and Rothvoss \cite{Kulkarni24} obtained optimal $O(\sqrt{\log T})$ prefix discrepancy, with a construction whose complexity depends exponentially on the dimension and horizon. Aden-Ali \cite{AdenAli26} attained the same prefix order in $O(nT)$ time in a real-arithmetic model with exact random sampling.\footnote{The operation count in \cite[proof of Theorem~1]{AdenAli26} treats real arithmetic, exponential evaluation, and exact scalar Gaussian and uniform sampling as constant-cost operations.}
Kintali \cite[Theorem~3.1]{Kintali26} showed that a universal constant bound for fixed-order prefix discrepancy is impossible even for vectors of Euclidean norm at most one and signs chosen with knowledge of the entire sequence.

The normalization matters: Rademacher vectors have Euclidean norm $\sqrt n$. For $T\ge n$ and failure probability $n^{-a}$, with fixed $a>0$, applying the unit-ball prefix guarantee to $v_t/\sqrt n$ and rescaling yields $O(\sqrt{n\log T})$. Aden-Ali's separate terminal guarantee \cite[Corollary~3]{AdenAli26} similarly yields $O(\sqrt{n\log n})$. This rescaling therefore retains a logarithmic factor relative to our $O(\sqrt n)$ terminal bound. For unbounded arrivals, applying a unit-ball guarantee additionally requires control of the input norms.

\paragraph{Sparse and small-coordinate inputs.}
Altschuler and Tikhomirov \cite{Altschuler26} obtained an efficient $O_\varepsilon(\sqrt k)$ prefix guarantee for every fixed $\varepsilon>0$ and every fixed sequence of vectors in $[-1,1]^n$ with at most $k$ nonzero coordinates, with failure probability at most
\begin{equation*}
    C_\varepsilon T\exp{-c_\varepsilon k/\log^{2+\varepsilon}(ek)}.
\end{equation*}
Hmadi \cite[Corollary~3.5]{Hmadi26} sharpened this to $O(\sqrt k)$ prefix discrepancy with failure probability at most\footnote{The algorithm assumes exact sampling and exact potential evaluation; no finite-bit running-time bound is claimed.}
\begin{equation*}
    CT\exp{-ck/\log^2(ek)}.
\end{equation*}
More generally, Hmadi obtained $O(1)$ prefix discrepancy for vectors satisfying $\norm{v_t}_2\le1$ and $\norm{v_t}_\infty\le k^{-1/2}$, without a support-size restriction. In this formulation, $k$ controls coordinate size rather than sparsity.

These results control all prefixes of fixed sequences with deterministic support or coordinate bounds. Our sparse theorem treats independent Bernoulli-masked arrivals with potentially unbounded centered marks and gives a terminal failure bound independent of $T$. Since the exponents in $k$ also differ, the factor $T$ alone does not determine the relative size of the failure bounds.

For a concrete comparison, consider Bernoulli-masked independent centered
unit-variance marks bounded by a fixed $B\ge1$, with $k=C_0\log^2(2n)$
for fixed $C_0\ge K_0$ and $T=\floor{n^a}$ for fixed $a>0$.
A Chernoff bound controls every mask support by $\ceil{2k}$, except with
probability $T\exp(-c'k)$.
Applying \cite[Corollary~3.5]{Hmadi26} to each such input sequence after
scaling by $B$, with a random seed independent of the inputs, gives
randomized $O(B\sqrt k)$ prefix discrepancy. For sufficiently large $n$, the total
failure probability is $\exp[-\Omega_{a,C_0}(\log^2 n/(\log\log n)^2)]$,
improving the $\exp[-\Omega_B(\log n)]$ error of
\cref{thm:SparseOnlineDiscrepancy} at the same $O_B(\sqrt k)$ discrepancy order.

\begin{table}[!t]
    \centering
    \caption{Selected online guarantees in our notation. Terminal bounds concern a prescribed deterministic time; prefix bounds hold simultaneously up to the horizon. Hard sparsity denotes a deterministic support bound, whereas Bernoulli sparsity refers to expected mask support size. In the sub-Gaussian row, $a_\sigma\colonequals \sigma\sqrt{\log(e+\sigma)}$.}
    \label{tab:MainResultsComparison}
    \begingroup
    \small
    \setlength{\tabcolsep}{4pt}
    \renewcommand{\arraystretch}{1.12}
    \begin{tabularx}{\textwidth}{@{}
        >{\raggedright\arraybackslash}p{0.23\textwidth}
        >{\raggedright\arraybackslash}p{0.29\textwidth}
        >{\raggedright\arraybackslash}X@{}}
    \toprule
    \textbf{Work} & \textbf{Input and objective} & \textbf{Guarantee and scope} \\
    \midrule
    Bansal--Spencer \cite{Bansal20a}
    & Rademacher product inputs; terminal
    & $O(\sqrt n)$; every prescribed $T$; error $\exp(-\Omega(\sqrt n))$. \\
    \addlinespace
    Fiedler et al. \cite{Fiedler26}
    & Gaussian; also centered unit-variance \iid inputs with finite fourth moment; terminal
    & Exact Gaussian asymptotic expected optimum for $T/n\to r\in(0,\infty)$; corresponding liminf lower bound for the other stated laws. \\
    \addlinespace
    Jo--Lee \cite{JoLee26}
    & Centered unit-variance \iid inputs; finite fourth moment; terminal
    & Universality of the same proportional expected optimum; upper bounds for targets a fixed positive distance above the optimum, with a fourth-moment tail term. \\
    \addlinespace
    \textbf{This work, Theorem~\ref{thm:OnlineDiscrepancy}}
    & Unit-variance product inputs; $\psi_1$-norms $\le\sigma$; terminal
    & $O(\sigma^4\sqrt n)$; every prescribed $T$; error $\exp(-\Omega(\sigma^3\sqrt n))$. \\
    \midrule
    Bansal et al. \cite{Bansal21}
    & \iid arrivals in $[-1,1]^n$; coordinate dependence allowed; prefix
    & $O(\sqrt n\log^4(nT))$ with high probability. \\
    \addlinespace
    Aden-Ali \cite{AdenAli26}
    & Fixed inputs with $\norm{v_t}_2\le1$; oblivious online; prefix
    & $O(\sqrt{\log(2T/\delta)})$ with error $\delta$; $O(nT)$ time in the stated real-arithmetic model. \\
    \midrule
    Altschuler--Tikhomirov \cite{Altschuler26}; Hmadi \cite{Hmadi26}
    & Fixed inputs in $[-1,1]^n$; hard sparsity $k$; prefix
    & Hmadi's bound: $O(\sqrt k)$; error $CTe^{-ck/\log^2(ek)}$. \\
    \addlinespace
    Altschuler--Tikhomirov (binary model) \cite{Altschuler26Binary}
    & Independent uniform vectors with exactly $k$ ones; prefix upper / terminal lower
    & $\Theta(\sqrt k+\log\log n)$ with high probability as $n\to\infty$, for $T=\Theta(n)$ and $2\le k\le n/2$; terminal lower bound extends to prescribed $T\ge n$. \\
    \addlinespace
    \textbf{This work, Theorem~\ref{thm:SparseOnlineDiscrepancy}}
    & Independent centered unit-variance marks; $\psi_1$-norms $\le\sigma$; Bernoulli masks with expected support size $k$; terminal
    & $O(\sigma^4\sqrt k)$ for $k\gtrsim\log^2(2n)$; every prescribed $T$; error $\exp(-\Omega(\sigma^3\sqrt k))$. \\
    \addlinespace
    \textbf{This work, Theorem~\ref{thm:SparseRefinements}}; Corollary~\ref{cor:SparseGaussian}
    & Centered unit-variance product marks; $\psi_2$-norms $\le\sigma$; independent $\Ber(k/n)$ masks; terminal
    & $O(a_\sigma[\sqrt k+\sqrt{\log(en/k)}+\log(1/\delta)])$ with error $\delta$; $1\le k\le n$ and every prescribed $T$; one rule for all $\delta\in(0,1)$. For Gaussian marks, $a_\sigma=O(1)$. \\
    \bottomrule
    \end{tabularx}
    \endgroup
\end{table}

For independent uniformly random vectors with exactly $k$ ones, Altschuler and Tikhomirov \cite{Altschuler26Binary}\footnote{This supersedes \cite{Altschuler25}, whose upper bound covered only the ultra-sparse range.} determined the optimal discrepancy order $\sqrt k+\log\log n$ with high probability as $n\to\infty$ when $2\le k\le n/2$ and $T=\Theta(n)$. Their efficient upper bound controls prefixes, and their terminal lower bound extends to every prescribed $T\ge n$. A matching upper bound uniform over arbitrary prescribed horizons remains open.

This binary model has fixed support size and nonzero entries equal to one. The fixed support size induces dependence among coordinates within a vector. Our model instead uses independent Bernoulli masks and centered marks, which may be unbounded. The Laplace single-hit obstruction in \cref{sec:MainResults} relies on unbounded marks and does not apply to the binary model.

For vectors in $[-1,1]^n$ with at most $k$ nonzero coordinates, the Koml\'os results above imply an offline $O(\sqrt k)$ bound by scaling \cite{Guo26Existence,Guo26Algorithm}. Earlier, Bansal and Jiang \cite{Bansal26} obtained an efficient offline bound of the same order for vectors in $\{0,\pm1\}^n$ with at most $k$ nonzero coordinates when $k\gtrsim\log^2 T$. These algorithms use the full input sequence, whereas our signs are chosen upon arrival.

\paragraph{Other random ensembles and offline discrepancy.}
Altschuler and Niles-Weed \cite{Altschuler22} proved that \iid Bernoulli matrices have offline discrepancy at most one with high probability when $T\ge Cn\log n$, allowing the Bernoulli parameter to vary with the dimensions. They obtained the same bound for \iid Poisson entries of fixed positive mean in the same regime, for even $T$. These are offline bounds for nonnegative inputs; our online model masks centered marks.

Conjecture~1 of the first arXiv version \cite{Wengiel26v1} asserted that
no randomized polynomial-time algorithm achieves $o(\sqrt n)$ offline
discrepancy with high probability for independent standard Gaussian inputs
at $T=\ceil{n^\beta}$, for any fixed $\beta>1$.
We announce a randomized polynomial-time offline algorithm achieving
$O_\beta(\sqrt{n/\log(T/n)})=o(\sqrt n)$ discrepancy with high probability.
Its proof is deferred to a companion manuscript in preparation and is not
included here; no theorem in this article depends on this announcement.
The announced bound would refute the conjecture and establish a separation
between polynomial-time offline signing and unrestricted online signing
at these horizons, by
\cref{thm:GaussianLowerBound}\ref{item:GaussianLowerBoundDense}.
Related small-margin asymptotics were announced by Huang and Sellke
in Sun's survey \cite[Table~1]{Sun26ICM}.
The fixed-aspect-ratio theorem of Huang, Sellke and Sun
\cite[Theorem~1.16]{Huang26} does not directly refute that offline conjecture:
it assumes a positive limiting ratio $n/T$, whereas the conjecture
concerns $T=\ceil{n^\beta}$ with $n/T\to0$.
\Cref{sec:OpenQuestions} explains the required uniformity and
discusses conditional computational barriers.

\FloatBarrier
\subsection{Application to the symmetric binary perceptron}
\label{sec:SBP}

For Gaussian inputs, our discrepancy bound gives an online algorithm for the \emph{symmetric binary perceptron} (SBP). Introduced by Aubin, Perkins and Zdeborov\'{a} \cite{Aubin19}, this model asks for signs $x_1,\ldots,x_T\in\{\pm1\}$ satisfying
\begin{equation*}
    \norm{\sum_{t=1}^T x_tv_t}_\infty\le\kappa\sqrt T,
\end{equation*}
where the $v_t\in\R^n$ are independent standard Gaussian vectors and $\kappa>0$ is the margin. The $T$ signs are its variables, and the $n$ coordinates are its constraints. We write $r\colonequals T/n$ for the variables-per-constraint ratio and $\alpha\colonequals n/T=1/r$ for the usual constraint density.

For fixed $\kappa>0$, the asymptotic satisfiability threshold, proved in \cite{Perkins21,Abbe22a}, is\footnote{The Gaussian threshold follows from the argument of Perkins and Xu together with the analytic estimates of Abbe, Li and Sly; see \cite[Theorem~3.6 and Lemma~3.11]{Altschuler23Critical} for an unconditional formulation.}
\begin{equation*}
    \alpha_{\SAT}(\kappa)
    \colonequals -\frac1{\log_2\P{\abs Z\le\kappa}},
    \qquad Z\sim N(0,1).
\end{equation*}
As $n,T\to\infty$ with $n/T\to\alpha\in(0,\infty)$, solutions exist with high probability when $\alpha<\alpha_{\SAT}(\kappa)$ and do not exist with high probability when $\alpha>\alpha_{\SAT}(\kappa)$. As $\kappa\downarrow0$, the corresponding variables-per-constraint threshold is $r_{\SAT}(\kappa)\colonequals 1/\alpha_{\SAT}(\kappa)=\log_2(1/\kappa)+O(1)$. We compare this logarithmic offline scale with the online bounds below.

For the Rademacher SBP at any fixed positive margin, Abbe, Li and Sly
\cite{Abbe22b} gave an efficient multiscale majority algorithm that finds
a solution with high probability at sufficiently small fixed constraint
density.

\begin{samepage}
\begin{corollary}[Gaussian online SBP]
    \label{cor:SBP}
    Fix a positive integer $n$. In the real-arithmetic model of \cref{sec:implementation}, there is a deterministic polynomial-time online algorithm depending only on $n$ that, for every $\kappa>0$ and every prescribed deterministic terminal time $T\in\N$ satisfying $T/n\ge C\kappa^{-2}$, finds a solution to the Gaussian SBP with probability at least $1-\exp(-c\sqrt n)$. Here $c,C>0$ are absolute constants.
\end{corollary}
\end{samepage}

The corollary follows directly from \cref{thm:OnlineDiscrepancy}, since the $\psi_1$-norm of a standard Gaussian variable is an absolute constant. The same algorithm works for all margins and horizons and uses $O(nT\log(2n))$ real arithmetic operations.

Conversely, every randomized online signing rule has SBP success probability below $0.767^n$ when $T\ge n$ and $r\le(4\kappa^2)^{-1}$, by \cref{thm:GaussianLowerBound}\ref{item:GaussianLowerBoundDense}. Together with \cref{cor:SBP}, this establishes the necessary and sufficient order $\kappa^{-2}$ of the variables-per-constraint ratio, up to absolute constants, in the small-margin regime with $T\ge n$. The necessary order was already established in the proportional small-margin regime by the Gaussian online lower bound of \cite{Gamarnik23}.

For fixed constraint density $\alpha>0$, the proportional results of \cite{Fiedler26,JoLee26} identify the critical online margin
\begin{equation*}
    R_\alpha\colonequals \sqrt\alpha\,V_{1/\alpha}^\infty.
\end{equation*}
As $n,T\to\infty$ with $n/T\to\alpha$, any fixed margin $\kappa>R_\alpha$ is attainable online with high probability, whereas every online signing rule has success probability tending to zero for fixed $\kappa<R_\alpha$ \cite[Corollary~1.2]{JoLee26}. \Cref{cor:SBP} provides an explicit finite-dimensional guarantee without requiring convergence of the aspect ratio.

\section{Technical overview and preliminaries}
\label{sec:TechnicalOverview}

The dense and centered sparse proofs use regularized maxima, global
Fenchel estimates, and terminal exponential moments.
The sub-Gaussian refinements give three complementary parameter profiles
$U_0,U_G,U_{\rm red}$, defined in \cref{thm:SparseEnvelope}, which describe
how the bounds depend on dimension, sparsity and tails.
The $U_0$ bound retains an absolute coefficient on the leading $\sqrt k$
term, while $U_G$ yields \cref{thm:SparseRefinements} and its Gaussian
specialization. The reduced profile $U_{\rm red}$ can improve the joint
parameter dependence when the tail bound grows with the dimension.
The basic centered
sparse analysis and the $U_0$ branch also control individual coordinates
through a coupled Lyapunov function. The $U_G$ and $U_{\rm red}$ branches
use a different regularizer to exploit rare simultaneous mask hits.
We explain the shared estimates and distinguish
the signing rules before giving their proofs.

\paragraph{Normalization and notation.}
For background on sub-Gaussian and sub-exponential random variables, see
\cite[Sections~2.5--2.7]{Vershynin18}. For a real random variable $Z$,
we use the conventions
\begin{equation}\label{eq:OrliczNorms}
 \norm Z_{\psi_1}\colonequals \inf\{s>0:\E e^{|Z|/s}\le2\},\qquad
 \norm Z_{\psi_2}\colonequals \inf\{s>0:\E e^{Z^2/s^2}\le2\}.
\end{equation}
A $\psi_2$-norm bound $\sigma>0$ for a conditional law given a
sigma-algebra $\mathcal G$ means
$\E(e^{Z^2/\sigma^2}\mid\mathcal G)\le2$ almost surely.
All unspecified constants are positive and absolute.
Write $\Fc_t$ for the complete history through arrival $t$.
In every one-step expectation, we condition on the preceding history, fixing
the current state and the next product law.
This yields the conditional extensions in
\cref{prop:DenseConditional,cor:SparseAllLevels} and
\cref{thm:SparseRefinements,thm:SparseEnvelope}.
The actual signing rules in this paper are deterministic.

For the basic dense and sparse bounds, $\sigma$ bounds the $\psi_1$-norm
of the raw coordinates or unmasked marks. With $p=k/n$ in the sparse
model, set
\begin{equation}\label{eq:normalization}
 (Y_t,X_t)\colonequals 
 \begin{cases}
 (v_t/\sigma,v_t/\sigma),&\text{dense},\\
 (w_t/\sigma,\xi_t\odot w_t/\sigma),&\text{sparse}.
 \end{cases}
\end{equation}
For the sub-Gaussian refinements, we instead use $X_t\colonequals v_t/A$, where
$\sigma_*\colonequals \max(\sigma,2)$ and
$A\colonequals 8\sigma_*\sqrt{\log(e\sigma_*)}$.
The Gaussian specialization takes $\sigma_*=2$, so its scaling is absolute.
In each proof, $d_t\colonequals \sum_{s\le t}x_sX_s$ denotes the state after the
indicated scaling. The raw discrepancy is $\sigma\norm{d_T}_\infty$
for the basic bounds and $A\norm{d_T}_\infty$ for all three
sub-Gaussian branches, including the Gaussian specialization.
The sparse marks are centered with variance one; dense coordinates
need not be centered. We never normalize sparse inputs by $\sqrt p$.

A variable of variance one has second moment at least one.
Since $e^u\ge1+u^2$ for $u\ge0$, the stated $\psi_1$ bound gives
$\sigma\ge1$. Under the basic normalization, the dense variance is $\sigma^{-2}$;
for centered marks this is also their normalized second moment.
For the centered profiles using this $\sigma$-normalization,
$\alpha\in\{1,2\}$ specifies the tail assumption
$\E e^{|Y_i|^\alpha}\le2$ and
$\ell\colonequals (\log(e\sigma))^{1/\alpha}$.
In the sub-Gaussian case,
\begin{equation}\label{eq:commonexp}
 \E e^{|Y_i|/2}\le e^{1/8}(\E e^{Y_i^2})^{1/2}<2.
\end{equation}

\paragraph{A geometric cap separate from the mask rate.}
For $1/n\le c\le1$, the common potential is
\[
 \F_c(d)\colonequals \max_{\substack{u,v\in[0,c]^n\\\sum_i(u_i+v_i)=1}}
 \left\{\inp{u-v}d+\frac{2\sqrt c}{\eta}
                    \sum_i(\sqrt{u_i}+\sqrt{v_i})\right\}.
\]
For $c=1$, it approximates $\norm d_\infty$ within
$O(\sqrt n/\eta)$ and is the square-root regularization used in the
first version \cite{Wengiel26v1}. In this case it is exactly the $q=1/2$
regularized maximum of Pesenti and Vladu \cite[Definition~3.6]{Pesenti23}
applied to $(d,-d)$. For $c<1$, the definition additionally imposes a
coordinate cap and scales the regularization coefficient by $\sqrt c$.
The KKT equations in \cref{lem:kkt} express both the optimizer and the
gradient in terms of the solution to a scalar normalization equation.

The general centered sparse confidence profiles take $c=p$.
For the $U_0$ branch, normalize by $A$ and retain
the coordinate variance $q\colonequals p/A^2$. Choosing $c=\max(q,1/n)$ makes the geometric scale
$\sqrt{nc}$ track the effective variance. Rescaling gives
$A\sqrt{nc}=\max(\sqrt k,A)$, which is how the leading
$\sqrt k$ term loses its multiplicative tail-parameter factor.
The same deterministic cap geometry applies to both choices.

The mask probability $p$, normalized coordinate variance $q$, and geometric
cap or breakpoint $c$ are distinct parameters; $q$ need not be the
probability that an input coordinate is nonzero. The local symbol $\nu$
denotes $\sigma^{-1}$ in the dense and basic sparse arguments, and $nc$
in the sub-Gaussian refinements; $\alpha$ denotes the tail class or,
in the two-breakpoint proof, a small error parameter.

\paragraph{Dense gradient signing and the global remainder.}
The dense rule uses $a_t\colonequals \nabla\F_1(d_{t-1})$ and chooses
\begin{equation}\label{eq:signingrule}
 x_t\colonequals -\sgn\inp{a_t}{X_t},
 \qquad d_t\colonequals d_{t-1}+x_tX_t,
\end{equation}
with a fixed sign at ties. Its main parameters are
\begin{equation}\label{eq:parameters}
 \eta\colonequals c_\eta\sigma^{-3},\qquad\lambda\colonequals c_\lambda,
\end{equation}
where the constants are sufficiently small and absolute.
Only $\eta$ enters the exact dense signing rule.

Let $a\colonequals \nabla\F_c(d)$. Fenchel duality bounds the increment,
for either sign, by
\[
 \F_c(d+xX)-\F_c(d)\le x\inp aX+U(X).
\]
The remainder $U$ is a sum of scalar functions of the individual
coordinates, with quadratic bounds for small updates and global linear
bounds.
It retains the cap and applies globally, avoiding an event bounding
all arrivals up to the horizon. The uncapped rational calculation has
an antecedent in the Tsallis-INF stability analysis
\cite[Appendix C.1, Lemmas~18--19]{Zimmert21}.
Conditional independence factors the exponential moment of $U$.
H\"older combines that moment with the negative linear term.
For dense coordinates, a variance-sensitive fourth-moment interpolation
and a negative Laplace estimate give a negative contribution of order
$-\lambda\sigma^{-3}\norm a_2$. Symmetrization permits nonzero means.
Above the potential threshold $6\sqrt n/\eta$, dense gradient geometry
allows the remainder to be absorbed, yielding the raw scale
$O(\sigma^4\sqrt n)$.

\paragraph{Centered sparse signing and the coupled potential.}
For masked coordinates, the head--tail functional $\mathcal K_c$ in
\cref{sec:SharedSparseGeometry} links lower bounds for first absolute
projection moments with upper bounds for the regularization remainder.
A capped potential need not dominate the largest coordinate, so its
contraction alone does not control discrepancy.
We use an auxiliary randomized rule that mixes greedy signing with
$-\sgn(\inp aX+Z_0)$ for an independent uniform threshold $Z_0$.
The greedy branch supplies a negative scalar contribution, while the
uniform-threshold branch ensures quantitative inward drift at
coordinates with large gradient weights.
A covariance identity for centered variables and monotone responses
produces inward coordinate drift without symmetry.
The general sub-exponential proof treats small and large marks
separately; the $U_0$ proof uses a simpler quadratic
remainder bound and an equal mixture of the two branches.

The potential and coordinate drifts are combined in
\[
 V(d)\colonequals e^{\lambda\F_c(d)}+
          \varepsilon\sum_i\cosh(\lambda d_i).
\]
The common coupling lemma absorbs growth of either component through
contraction of the other, with an injection proportional to the
contraction rate. The ideal deterministic rule chooses the smaller of
$V(d+X)$ and $V(d-X)$. This pointwise comparison inherits the full
conditional drift of the auxiliary randomized rule and requires no
knowledge of the input laws. Its certified implementation preserves the
terminal guarantees with adjusted absolute constants; see
\cref{sec:implementation}. The function $V$ therefore enters the
sparse sign computation; the dense rule remains a gradient rule.
For fixed tail parameters, the resulting additive $\log(2n)$ term is
absorbed by $\sqrt k$ once $k\gtrsim\log^2(2n)$.

\paragraph{Removing the hard cap.}
For Gaussian marks, a dimension term of order $\sqrt{\log n}$ is
necessary for $n\ge3$ at prescribed terminal times $T\ge n$;
see \cref{cor:GaussianDimensionLower}.
The mask-sized cap can prevent a single leading coordinate from
contracting strongly enough. The $U_G$ and $U_{\rm red}$ rules remove
the hard cap and use a convex penalty with two separately tuned branches.
The lower branch has square-root form, while the upper branch allows
larger weights on leading coordinates. The curvature breakpoint is
$c=p$ for $U_G$ and $c=\max(p/A^2,1/n)$ for $U_{\rm red}$.
Each rule minimizes successor values of the resulting regularized maximum.
Its scalar Fenchel calculation retains inward motion, and simultaneous
mask hits on two distinct coordinates contribute a factor $p^2$.
Variance-sensitive moment bounds and a product tilt relative to the
actual projection control the low coordinates. The tilted laws need not be
centered; centering is used only under the original conditional laws.
Heavy-coordinate errors
are compared with the actual inward deficit in the same coordinate.
The resulting regularized maximum $F$ satisfies
\[
 \norm d_\infty\le F(d)+O(S/\kappa),\qquad
 S\colonequals \sqrt{\log(e/c)},\quad
 \kappa\asymp[\log(ep/c)]^{-1/2}.
\]
The terminal moment argument gives the raw discrepancy threshold
\[
 C\left[\frac{A}{\kappa}
       \left(\sqrt{nc}+\sqrt{\log(e/c)}\right)+Ar\right]
\]
at failure probability $e^{-r}$. Substituting the two breakpoints yields
$U_G+Ar$ and $U_{\rm red}+Ar$, respectively, without a coordinate cosh
correction. For every fixed $\sigma$, the
$U_G$ correction is absorbed by $\sqrt k$ once $k\gtrsim\log(2n)$.
This gives the sufficient logarithmic threshold for the full
sub-Gaussian class with fixed tail bound, including Gaussian marks.

\paragraph{Terminal moments.}
For the terminal argument, we use $V=e^{\lambda\F_1}$ in the dense
case, the coupled function above for the general centered sparse
profiles and $U_0$, and $V=e^{\lambda F}$ for $U_G$ and $U_{\rm red}$.
Each proof obtains a time-independent recursion
\[
 \E{V(d_{t+1})\mid\Fc_t}\le(1-\delta)V(d_t)+J,
 \qquad V(d)\ge q_0e^{\lambda\norm d_\infty}.
\]
Here $0<\delta\le1$, $J\ge0$, and $q_0,\lambda>0$. These parameters
may depend on the model parameters but are independent of time, the
terminal horizon, and the realized history. The initial moment is
finite, so \cref{lem:terminal} gives the bound at each prescribed
terminal time by iteration and Markov's inequality.
This use of exponential moments is in the spirit of the classical drift
approach of \cite{Hajek82}.
For the two new branches, contraction and injection carry the same
small factor, so iteration introduces no inverse-variance loss.
The terminal interpretation is explained in \cref{sec:MainResults}.

\paragraph{Rare activations.}
The lower bound in \cref{prop:RareActivationLower} iterates a level-growth
argument over fresh time blocks. For activations of magnitude one,
coordinates of magnitude at least $j$ are paired when both receive their
only hit in the block at the same
arrival. If their mark signs, oriented relative to their initial states,
oppose, either signing forces one magnitude to at least $j+1$.
For the dense witness law, a priority rule gives the matching upper bound
in \cref{prop:RareActivationUpper}. Its proof in
\cref{app:RareActivationUpper} uses exponential moments and a recursion
for level occupancies.

\Cref{sec:CommonEstimates} gives the shared estimates,
\cref{sec:OnlineDiscrepancy} completes the dense proof, and
\cref{sec:SparseOnlineDiscrepancy} develops the centered sparse
arguments, including the low-coordinate tilt and heavy-coordinate error
estimates in \cref{sg:lem:low-tilt,sg:lem:heavy-deficits}.
Both estimates remain uniform over the conditional mark laws, even when
these laws differ across coordinates or depend on the past.
Implementation and comparison tolerances are given
in \cref{sec:implementation}.

\section{A common capped regularization}
\label{sec:CommonEstimates}

We establish the common estimates at a deterministic state $d\in\R^n$;
the probabilistic statements also apply conditionally on the complete past
whenever their independence and moment assumptions hold conditionally
almost surely.
For $\eta>0$ and $1/n\le c\le1$, define
\begin{equation}\label{eq:potential}
 \begin{aligned}
 D_c&\colonequals \{z=(u,v)\in[0,c]^{2n}:\textstyle\sum_i(u_i+v_i)=1\},\\
 \F_c(d)&\colonequals \max_{z\in D_c}\left\{\inp{u-v}{d}
                   +\frac{2\sqrt c}{\eta}\sum_j\sqrt{z_j}\right\}.
 \end{aligned}
\end{equation}
The geometric cap $c$ and the physical mask probability $p$ play different
roles and need not be equal.
We suppress the subscript when the chosen cap is clear. Here
$u,v$ denote the maximizing weights, and $a\colonequals u-v$ their signed
difference. For each coordinate, the \emph{favored weight} is $u_i$ when
$d_i\ge0$ and $v_i$ when $d_i<0$. The bound on the individual weights is
inactive in the dense case $c=1$; the sparse algorithms choose the cap separately from the mask probability.
The dense and basic centered sparse bounds, as well as the $U_0$ branch,
use the capped-potential tools developed here. The $U_G$ and $U_{\rm red}$
branches use their own uncapped geometry, projection estimates, and
remainder bounds; from this section they reuse the Bernoulli norm
comparison and the terminal-moment lemma.

\subsection{Geometry and computation of the potential}

The maximizer and the gradient can be recovered from a single scalar
normalization equation. The next lemma records this representation, the
geometric bounds used in the drift arguments, and the effect of
approximating the scalar root.

\begin{lemma}[KKT formulas and implementation geometry]\label{lem:kkt}
\leavevmode
\begin{enumerate}[label=(\roman*),ref=(\roman*),leftmargin=*]
\item \label{item:KKTOptimizer} Optimizer and normalization.
The maximizer is unique, has positive entries, and $\F$ is differentiable
with gradient $a=u-v$. Define
\[
 \phi_c(q)\colonequals \begin{cases}c,&q\le c^{-1/2},\\q^{-2},&q>c^{-1/2}.
 \end{cases}
\]
There is a unique root $L$ of
\begin{equation}\label{eq:root}
 \sum_i\left[\phi_c(L-\eta d_i/\sqrt c)
                    +\phi_c(L+\eta d_i/\sqrt c)\right]=1,
\end{equation}
and $u_i=\phi_c(L-\eta d_i/\sqrt c)$,
$v_i=\phi_c(L+\eta d_i/\sqrt c)$.
\item \label{item:KKTGeometry} Geometric bounds.
Writing $M\colonequals \norm{d}_\infty$, $k_c\colonequals nc$, and $\theta\colonequals \sqrt c L/\eta$, one has
\begin{gather}\label{eq:kktbounds}
 \sqrt{2n}\le L\le\eta M/\sqrt c+\sqrt{2n},\qquad
 |a_i|\le c,\quad\norm{a}_1\le1,\quad\sgn a_i=\sgn d_i,\\
 \label{eq:Fbounds}
 \F(0)=2\sqrt{2k_c}/\eta,\quad \F(d)\ge\F(0),\qquad \F(d)\le M+2\sqrt{2k_c}/\eta,
 \qquad \theta\le \F(d).
\end{gather}
$\F$ is $1$-Lipschitz in $\ell_\infty$. If $|d_i|>\F(d)$,
its favored weight is capped and $|a_i|\ge c/2$.
\item \label{item:KKTApproximation} Approximate weights.
For $\widehat L\ge L$, define approximate weights by the same formulas,
$\widehat u_i\colonequals\nobreak \phi_c(\widehat L-\eta d_i/\sqrt c)$ and
$\widehat v_i\colonequals \phi_c(\widehat L+\eta d_i/\sqrt c)$, and put
$\widehat a\colonequals \widehat u-\widehat v$. Then
\begin{equation}\label{eq:gradienterror}
 \norm{a-\widehat a}_1\le2\sqrt c(\widehat L-L).
\end{equation}
The approximate weights have total mass at most one and obey the same cap,
and $\sgn(\widehat a_i)=\sgn(d_i)$ for every $i$.
\end{enumerate}
\end{lemma}
\begin{proof}
The domain is compact and nonempty, so the maximum is attained, and strict
concavity gives uniqueness. Moving mass $\varepsilon$ from a positive
entry to a zero entry gives a gain of order $\sqrt\varepsilon$, while
the loss is $O(\varepsilon)$. Thus every entry is positive.
The cap-constrained KKT equations give the stated formula.
The normalization function is continuous and nonincreasing, equals $2nc>1$
far to the left, and tends to zero to the right. Once any weight leaves its
plateau, the normalization function is strictly decreasing; hence its root is unique.
For $L\le c^{-1/2}$ all favored weights are capped, contributing at least
$nc\ge1$, while the opposite weights are positive. Thus, $L>c^{-1/2}$.
In fact, $L\ge\sqrt{2n}$. Otherwise, a coordinate pair with an uncapped
favored weight has mass at least $2/L^2>1/n$, by convexity of $q^{-2}$.
A pair with a capped weight has mass greater than $c\ge1/n$.
Summing over the pairs contradicts total mass one.
At the displayed upper bracket every KKT argument is at
least $\sqrt{2n}$, so the total mass is at most one.
The envelope theorem gives the gradient; monotonicity of $\phi_c$ and
$L>c^{-1/2}$ give its sign. Convexity and evenness give $\F(d)\ge\F(0)$.
The cap, mass, and Lipschitz bounds follow directly from the domain, and Cauchy--Schwarz gives the potential upper bound
and its value at zero.

Writing $g\colonequals (d,-d)$, the upper-cap multipliers $\mu_j\ge0$ satisfy
\[
 g_j+\frac{\sqrt c}{\eta\sqrt{z_j}}=\theta+\mu_j,
 \qquad \mu_j(z_j-c)=0.
\]
Multiplication by $z_j$ and summation give
\begin{equation}\label{eq:euler}
 \F(d)=\theta+\frac{\sqrt c}{\eta}\sum_j\sqrt{z_j}
                        +\sum_j\mu_jz_j\ge\theta.
\end{equation}
If $d_i>\F(d)\ge\theta$, the positive weight is capped and
$v_i\le L^{-2}\le1/(2n)\le c/2$; negative $d_i$ is identical.
Finally, let $z_j(\ell)$ denote the weight obtained by replacing $L$ with
$\ell$ in its KKT formula. For $\ell\ge L$, the mass is at most one, and almost everywhere
$|z_j'(\ell)|\le2z_j(\ell)^{3/2}$. Thus,
\[
 \sum_i|a_i'(\ell)|\le2\sum_jz_j(\ell)^{3/2}
                        \le2\sqrt c\sum_jz_j(\ell)\le2\sqrt c.
\]
Integrating proves~\eqref{eq:gradienterror}.
\end{proof}

\subsection{A global bound on the Bregman remainder}

The following scalar estimate keeps the upper bound on the weights. This
provides a global linear bound for large updates as well as a quadratic
bound for small updates.

\begin{lemma}[Capped scalar Fenchel estimate]\label{lem:fenchel}
For $0<z\le c$, define
\[
 g_z(q)\colonequals \sup_{0\le r\le c}\left\{q(r-z)
             -\frac{\sqrt c}{\eta\sqrt z}(\sqrt r-\sqrt z)^2\right\},
 \qquad h_z(q)\colonequals \max\{g_z(q),g_z(-q)\}.
\]
\begin{enumerate}[label=(\roman*),ref=(\roman*),leftmargin=*]
\item Scalar bounds.
Then $0\le h_z(q)\le c|q|$, and
\begin{equation}\label{eq:scalarbound}
 h_z(q)\le\frac{2\eta z^{3/2}}{\sqrt c}q^2
 \quad\text{if }|q|\le\frac{\sqrt c}{2\eta\sqrt z}.
\end{equation}
\item Potential increment.
For every update $w$ and sign $x$,
\begin{equation}\label{eq:fenchel}
 \F(d+xw)-\F(d)-x\inp{a}{w}
 \le U(w)\colonequals \sum_i[h_{u_i}(w_i)+h_{v_i}(w_i)].
\end{equation}
\end{enumerate}
\end{lemma}
\begin{proof}
The penalty in the definition of $g_z$ equals the Bregman divergence
$D_f(r,z)\colonequals f(r)-f(z)-f'(z)(r-z)$ of $f(r)\colonequals -2\sqrt{cr}/\eta$.
It is nonnegative; choosing $r=z$ gives the lower
bound, and $|r-z|\le c$ gives the global linear cap.
Removing the upper cap and solving the scalar stationary equation gives
\[
 \frac{\eta z^{3/2}q^2}{\sqrt c-\eta\sqrt zq}
 \quad(q<\sqrt c/(\eta\sqrt z)),
\]
which bounds the capped supremum from above and proves~\eqref{eq:scalarbound}.
For completeness the capped value switches at
$q=(\sqrt c/\eta)(z^{-1/2}-c^{-1/2})$ to
$q(c-z)-\sqrt c(\sqrt c-\sqrt z)^2/(\eta\sqrt z)$.

Let $z$ be the optimizer at $d$, let $r$ be the optimizer at $d+xw$, and
write $Bw\colonequals (w,-w)$. Subtracting the objectives and using the KKT equations
and the equality of their total masses gives
\[
 \F(d+xw)-\F(d)-x\inp{a}{w}
 =\sum_j\left[x(Bw)_j(r_j-z_j)-D_f(r_j,z_j)
                         +\mu_j(r_j-z_j)\right].
\]
The multiplier contribution $\sum_j\mu_j(r_j-z_j)$ is nonpositive:
wherever $\mu_j>0$, $z_j=c$ and $r_j\le c$.
Dropping this nonpositive term and relaxing only the simplex constraint to
$[0,c]^{2n}$ give a sum of scalar conjugates. Taking the larger scalar
value for the two signs proves~\eqref{eq:fenchel}.
\end{proof}

For the next product estimate, the geometric cap equals the mask
probability: $c=p$. Later refinements allow the cap and mask rate to differ.

\begin{lemma}[Product exponential-moment bound]\label{lem:remainder}
Suppose $X_i\colonequals \xi_iY_i$ have independent coordinates, with independent
$\Ber(p)$ masks, also independent of the marks; for $p=1$, the masks
are identically one. Fix $A\in[1,2]$ and assume
$\E e^{|Y_i|/A}\le2$ for every $i$. If $\eta A\le1$, then for
$0\le t\le1/16$,
\begin{equation}\label{eq:remainder}
 \log\E e^{tU(X)}\le Ct\eta\sqrt p S_z,
 \qquad S_z\colonequals \sum_jz_j^{3/2}.
\end{equation}
\end{lemma}
\begin{proof}
For one coordinate, abbreviate $u\colonequals u_i$, $v\colonequals v_i$ and write
$U_i\colonequals h_u(X_i)+h_v(X_i)\le2p|X_i|$.
For $z=u,v$, set $B_z\colonequals \sqrt p/(2\eta\sqrt z)$, and split according to
$|Y_i|\le B_z$ and its complement.
The local estimate and the Bernoulli mask give
\[
 \E{h_z(X_i)e^{tU_i}}
 \le\frac{2p\eta z^{3/2}}{\sqrt p}
                \E{Y_i^2e^{2tp|Y_i|}}
       +p^2\E{|Y_i|e^{2tp|Y_i|}\I{|Y_i|>B_z}}.
\]
The first expectation is at most $CA^2$. The second term is at most
$Cp^2A e^{-B_z/(4A)}$, by the exponential moment assumption and
$2tp\le1/(4A)$. Set $b\colonequals \eta A\sqrt{z/p}\le1$.
Since $e^{-1/(8b)}\le Cb^3$ and $\eta A\le1$,
\[
 p^2A e^{-1/(8b)}\le Cp^2Ab^3
                       \le C\eta A^2\sqrt p z^{3/2}.
\]
Thus, $\E(U_ie^{tU_i})\le C\eta\sqrt p(u^{3/2}+v^{3/2})$.
For $q\ge0$, use $e^{tq}-1\le tq e^{tq}$ and $\log(1+r)\le r$ to obtain
\[
 \log\E e^{tU_i}\le t\E{U_i e^{tU_i}}
 \le Ct\eta\sqrt p(u_i^{3/2}+v_i^{3/2}).
\]
Independence now permits multiplication of the coordinate moments, which
proves~\eqref{eq:remainder} in both regimes.
\end{proof}

\subsection{Drift and terminal moments}

We first separate the linear signing term, the regularization remainder,
and the error from approximate weights. The resulting drift estimate
will be converted into a terminal tail bound.

\begin{lemma}[Signing errors and negative Laplace transforms]\label{lem:drift}
Fix $a,\widehat a\in\R^n$ and put $e\colonequals a-\widehat a$.
For a fresh vector $X$, let $x\colonequals -\sgn\inp{\widehat a}{X}$,
with either tie convention.
\begin{enumerate}[label=(\roman*),ref=(\roman*),leftmargin=*]
\item Signing errors.
Then $x\inp{a}{X}\le-|\inp{a}{X}|+2|\inp{e}{X}|$.
If the $X_i$ are independent and $\E e^{|X_i|/A}\le2$ for $A>0$,
then $t\ge0$ and $tA\norm{e}_\infty\le1$ imply
\begin{equation}\label{eq:errormgf}
 \log\E e^{t|\inp{e}{X}|}\le(\log2)tA\norm{e}_1.
\end{equation}
\item Negative Laplace transform.
For any $Z\ge0$ with $\E Z\ge m>0$ and $\E Z^2\le v$,
\begin{equation}\label{eq:negative}
 \log\E e^{-tZ}\le-\frac{tm^2}{m+tv}\qquad(t\ge0).
\end{equation}
\item Combined potential drift.
For the composition bound, fix a state $d$, take $c=p$ and
$a=\nabla\F_p(d)$, and assume the product and moment hypotheses of
\cref{lem:remainder}, including $\eta A\le1$.
Let $Z\colonequals |\inp{a}{X}|$ satisfy $\E Z\ge m>0$ and $\E Z^2\le v$.
Writing $\Delta\F\colonequals \F_p(d+xX)-\F_p(d)$, H\"older with three factors yields
\begin{equation}\label{eq:composition}
 \log\E e^{\lambda\Delta \F}
 \le-\frac{\lambda m^2}{m+3\lambda v}+C\lambda\eta\sqrt p S_z
                      +C\lambda A\norm{e}_1,
\end{equation}
whenever $0\le\lambda\le1/48$ and $6\lambda A\norm{e}_\infty\le1$.
For exact gradients the last term vanishes.
\end{enumerate}
In conditional applications, $d,a,\widehat a$ may be measurable with
respect to the complete past; the sign may depend on the fresh vector $X$.
\end{lemma}
\begin{proof}
The signing inequality follows from the triangle inequality.
For the error bound use $|\inp{e}{X}|\le\sum_i|e_i||X_i|$, multiply the
coordinate moments, and use the concave-power form of Jensen:
$\E e^{t|e_i||X_i|}\le2^{tA|e_i|}$.
For $t>0$, use $1-e^{-q}\ge q/(1+q)$ and Cauchy--Schwarz to obtain
\[
 1-\E e^{-tZ}\ge t\E\frac Z{1+tZ}
 \ge\frac{t(\E Z)^2}{\E Z+t\E Z^2}
 \ge\frac{tm^2}{m+tv}.
\]
The final inequality uses that $s^2/(s+tu)$ is increasing in $s$
and decreasing in $u$. Applying $\log(1-r)\le-r$ proves~\eqref{eq:negative};
$t=0$ is immediate.
Finally combine the signing inequality with~\eqref{eq:fenchel} and apply
H\"older's inequality:
\begin{align*}
 \log\E e^{\lambda\Delta\F}
 &\le\tfrac13\log\E e^{-3\lambda Z}
       +\tfrac13\log\E e^{3\lambda U(X)}
       +\tfrac13\log\E e^{6\lambda|\inp{e}{X}|}.
\end{align*}
The three preceding estimates give~\eqref{eq:composition}.
\end{proof}

The following interpolation bound supplies the first absolute projection
moment needed in the drift estimate.

\begin{lemma}[Fourth-moment interpolation]\label{lem:interpolation}
Let $Y_i$ be independent centered variables with variance $\nu^2>0$ and
$\E Y_i^4\le M_4$. For every deterministic $a$,
\[
 \E|\inp{a}{Y}|\ge
 \frac{\nu^3}{\sqrt{M_4+3\nu^4}}\norm{a}_2.
\]
\end{lemma}
\begin{proof}
Writing $g\colonequals \norm{a}_2$, independence and centering give
$\E|\inp{a}{Y}|^2=\nu^2g^2$ and
\[
 \E|\inp{a}{Y}|^4=\sum_i a_i^4\E Y_i^4
                    +6\sum_{i<j}a_i^2a_j^2\nu^4
                 \le(M_4+3\nu^4)g^4.
\]
The interpolation inequality
$\E Z^2\le(\E|Z|)^{2/3}(\E Z^4)^{1/3}$ proves the result; $g=0$ is immediate.
\end{proof}

A conditional moment recursion then gives a bound at each prescribed
terminal time, without a union bound over arrivals.

\begin{lemma}[Terminal moments]\label{lem:terminal}
If an adapted state process and a nonnegative function $V$ obey
$\E(V(d_{t+1})\mid\Fc_t)\le(1-\delta)V(d_t)+J$
with $0<\delta\le1$, $J\ge0$, and finite initial moment, then
\[
 \sup_{t\ge0}\E V(d_t)\le\E V(d_0)+J/\delta.
\]
If $V(d)\ge q_0e^{\lambda\norm{d}_\infty}$ with $q_0,\lambda>0$, then for $r\ge0$,
\[
 \P\left(\norm{d_t}_\infty>
    \lambda^{-1}\left[\log\frac{\E V(d_0)+J/\delta}{q_0}+r\right]\right)
 \le e^{-r}.
\]
All bounds concern one marginal at each prescribed time.
\end{lemma}
\begin{proof}
Writing $M_t\colonequals \E V(d_t)$, taking expectations and iterating gives
\[
 M_t\le(1-\delta)^tM_0+J\sum_{j=0}^{t-1}(1-\delta)^j
       \le M_0+J/\delta.
\]
The lower bound on $V$ and Markov's inequality then imply the asserted tail
estimate at each prescribed time.
\end{proof}

The estimates developed so far suffice for the basic dense bound.
The finer dense profiles additionally use
\cref{sec:VarianceRefinements,sec:SharedSparseGeometry}, while the capped
sparse arguments also require \cref{sec:ThresholdComparison,sec:CappedGeometry}.

\subsection{Variance-sensitive refinements}
\label{sec:VarianceRefinements}

For the centered sparse theorem and the finer
confidence profiles in \cref{thm:DenseConfidence,thm:SparseConfidence},
we also keep the normalized variance in the remainder bound.
Let $\alpha\in\{1,2\}$, $\sigma\ge1$, $\nu\colonequals \sigma^{-1}$, and
$\ell\colonequals (\log(e\sigma))^{1/\alpha}$. Suppose the independent centered marks
satisfy $\E Y_i^2=\nu^2$ and $\E e^{|Y_i|^\alpha}\le2$.

\begin{lemma}[Refined moments and remainder]\label{lem:variancerefinement}
\leavevmode
\begin{enumerate}[label=(\roman*),ref=(\roman*),leftmargin=*]
\item Mark moments.
One has
\begin{equation}\label{eq:refinedmoments}
 \E|Y_i|^4\le C\nu^2\ell^2,\qquad
 \E|Y_i|^3\le C\nu^2\ell,\qquad
 \E\bigl[Y_i^2e^{t|Y_i|}\bigr]\le C\nu^2\quad(0\le t\le c_*/\ell).
\end{equation}
\item Masked remainder.
For independent $\Ber(p)$ masks, independent of the marks, put
$X_i\colonequals \xi_iY_i$; when $p=1$ the masks are identically one.
For this remainder assertion, the geometric cap is $c=p$.
The remainder in \cref{lem:fenchel} satisfies
\begin{equation}\label{eq:varianceremainder}
 \log\E e^{tU(X)}\le Ct\eta\nu^2\sqrt p S_z
\end{equation}
whenever $0\le t\le c_*$, $\eta\ell\le c_*$, and $t\eta\ell^2\le c_*$.
\end{enumerate}
\end{lemma}
\begin{proof}
Set $H_i\colonequals |Y_i|$. The exponential-moment assumption implies, for
$j=1,\ldots,4$, $R\ge1$, and $0\le\theta\le1/4$,
\[
 \E\bigl[H_i^j e^{\theta H_i}\I{H_i>R}\bigr]
 \le C_j e^{-R^\alpha/2}.
\]
Indeed, $s^j e^{\theta s}\le C_j e^{s^\alpha/2}$ for $s\ge1$.
Fix $K=8$, choose the absolute parameter constant $c_*\le1/32$,
and split at $R_*\colonequals K\ell$.
Below this level, $H_i^4\le R_*^2H_i^2$. The tail estimate bounds the
remaining fourth moment by $Ce^{-R_*^\alpha/2}\le C\nu^2$.
Cauchy--Schwarz between the second and fourth moments gives the third
moment. The same split proves the tilted second-moment bound: the
exponential factor is bounded by an absolute constant on $H_i\le R_*$,
and for $t\le c_*/\ell$ its
remaining tail contributes at most $C\nu^2$.

For one coordinate, put $U_i(y)\colonequals h_{u_i}(y)+h_{v_i}(y)$. We bound
\[
 \E\bigl[U_i(X_i)e^{tU_i(X_i)}\bigr]
 \le C\eta\nu^2\sqrt p\,(u_i^{3/2}+v_i^{3/2})
\]
by treating the two weights separately and conditioning on the mask.
For every $z\le p$, the condition $\eta\ell\le c_*$ gives
$B_z\colonequals \sqrt p/(2\eta\sqrt z)\ge16\ell\ge R_*$.
On $H_i\le R_*$, the scalar quadratic bound gives
$U_i(Y_i)\le4\eta pR_*^2$, so $e^{tU_i(Y_i)}\le C$.
On the complement use $U_i(Y_i)\le2pH_i$ and the exponential tail.
Consequently $\E\bigl[H_i^2e^{tU_i(Y_i)}\bigr]\le C\nu^2$.
For $z=u_i,v_i$, the local contribution including its mask is at most
$C\eta\nu^2\sqrt p z^{3/2}$. The remaining contribution is bounded by
\[
 p^2\E\bigl[H_i e^{2tpH_i}\I{H_i>B_z}\bigr]
 \le Cp^2e^{-B_z^\alpha/2}.
\]
Put $b\colonequals \eta\sqrt{z/p}\le\eta$, so $B_z=1/(2b)$.
Since $B_z\ge16\ell$,
\[
 e^{-B_z^\alpha/2}
 =e^{-B_z^\alpha/4}e^{-B_z^\alpha/4}\le C\nu^2b^3.
\]
The first factor is at most $\nu^2$, and the second is at most $Cb^3$
because $B_z=1/(2b)$.
Thus this contribution is at most
$C\nu^2\eta^3\sqrt p z^{3/2}\le C\nu^2\eta\sqrt p z^{3/2}$.
Finally, $e^{tq}-1\le tq e^{tq}$ and coordinate independence give
\eqref{eq:varianceremainder}, as in \cref{lem:remainder}.
\end{proof}

The $U_0$ branch instead takes $c=\max(p/A^2,1/n)$ and uses the
variance-sized remainder bound in \cref{vg:lem:remainder}.

\subsection{Head--tail geometry and centered projections}
\label{sec:SharedSparseGeometry}

The next two lemmas separate the geometry of the cap from the moment
assumptions on the input.
For a geometric cap $1/n\le c\le1$, put $m_c\colonequals \lfloor1/c\rfloor$ and
\[
 \mathcal K_c(a)\colonequals c\sum_{j\le m_c}a_j^*
       +\sqrt c\left(\sum_{j>m_c}(a_j^*)^2\right)^{1/2},
\]
where $a_j^*$ are the decreasing absolute coordinates. This definition
uses the cap, which need not equal the physical mask rate.
The Bernoulli variables in the norm comparison are auxiliary variables
with parameter $c$; they need not have the physical mask distribution.
In the next projection lemma, $q$ denotes the common coordinate variance.

\begin{lemma}[Bernoulli norm comparison]\label{lem:BernoulliNorm}
For independent $\operatorname{Ber}(c)$ variables $\zeta_i$,
$G_c(a)\colonequals \E\norm{a\odot\zeta}_2$ is a norm and
\[
 c_0\mathcal K_c(a)\le G_c(a)\le\mathcal K_c(a),\qquad
 G_c(a)\le c\norm a_1.
\]
If $|a_i|\le c$ and $\norm a_1\le1$, then
$\norm a_2^2\le\mathcal K_c(a)$.
\end{lemma}
\begin{proof}
The sole-hit events among the first $m_c$ coordinates are disjoint,
each with probability $c(1-c)^{m_c-1}\ge c_0c$, and give the head lower bound.
For the tail set $Q^2\colonequals c\sum_{j>m_c}(a_j^*)^2$ and
$Z\colonequals \sum_{j>m_c}\zeta_j(a_j^*)^2$.
Then $\E Z=Q^2$ and $\E Z^2\le Q^4+(a_{m_c}^*)^2Q^2$.
If $Q\ge a_{m_c}^*$, interpolation gives $\E\sqrt Z\ge Q/\sqrt2$;
otherwise the head controls $Q$, since $cm_c\ge1/2$.
For $c=1$, the sole-hit probability is one and there is no difficulty.
Triangle inequality on the head and Jensen on the tail give the upper
bound. The norm assertion and $G_c(a)\le c\norm a_1$ follow directly.
For the last claim, bound head squares by $c$ times head absolute mass,
and use $\norm a_2\le\sqrt c$ on the tail.
\end{proof}

The norm comparison yields a projection bound under second- and
fourth-moment assumptions. Symmetrization allows nonsymmetric inputs.

\begin{lemma}[Centered projection with a separate geometric cap]
\label{lem:CenteredProjection}
Let $0<q\le c\le1$, $c\ge1/n$, and $\ell_0\ge1$.
Suppose the $X_i$ are independent and centered, with
\[
 \E X_i^2=q,\qquad \E X_i^4\le C_4q\ell_0^2,
\]
where $C_4$ is an absolute constant. There are absolute $b,B>0$ such
that the following hold for every deterministic $a\in\R^n$.
\begin{enumerate}[label=(\roman*),ref=(\roman*),leftmargin=*]
\item Absolute projection moment.
\[
 \E |\langle a,X\rangle|
       \ge b\frac{q}{c\ell_0}\mathcal K_c(a).
\]
\item Negative Laplace transform.
If additionally $|a_i|\le c$ and $\|a\|_1\le1$, then
\begin{equation}\label{eq:CenteredProjectionLaplace}
 \log\E  e^{-t|\langle a,X\rangle|}
 \le -\frac{b t q\mathcal K_c(a)}{c\ell_0(1+Bt c\ell_0)}
 \qquad(t\ge0).
\end{equation}
In particular, the right side can be replaced by
$-b'tq\mathcal K_c(a)/(c\ell_0^2)$ for $0\le t\le1$, and by
$-b'tq\mathcal K_c(a)/(c\ell_0)$ for $0\le t\le1/\ell_0$.
\end{enumerate}
The constants may depend on the fixed absolute bound $C_4$.
\end{lemma}
\begin{proof}
Take independent copies $X_i'$ and put $Z_i\colonequals (X_i-X_i')/2$.
The $Z_i$ are independent and symmetric, with
\[
 \E Z_i^2=q/2,\qquad
 \E Z_i^4\le C_5q\ell_0^2,
 \qquad
 \E |\langle a,X\rangle|\ge\E |\langle a,Z\rangle|.
\]
Here $q\le1$ and $\ell_0\ge1$ absorb the $q^2$ term in the fourth moment.
Choose an absolute $B_0\ge\max(1,4C_5)$, set
$M_0\colonequals B_0\ell_0^2$, and let $Y_i\colonequals \min(Z_i^2,M_0)$.
Then $\E Y_i\ge q/2-\E Z_i^4/M_0\ge q/4$.
Conditionally on magnitudes, Rademacher moment interpolation gives
\[
 \E |\langle a,Z\rangle|
 \ge\frac1{\sqrt3}\E \sqrt{\sum_i a_i^2Y_i}
 \ge\sqrt{M_0/3}\E \|a\odot\zeta\|_2,
 \qquad \zeta_i\sim\operatorname{Ber}(q/(4M_0)).
\]
The second inequality follows by applying the chord lower bound for a
concave square root separately to each independent $Y_i/M_0\in[0,1]$,
then using monotonicity to decrease the resulting Bernoulli probabilities.
Represent $\zeta_i=\xi_i\chi_i$, where the two families are independent,
$\xi_i\sim\operatorname{Ber}(c)$ and
$\chi_i\sim\operatorname{Ber}(q/(4M_0c))$.
This is permissible because $q\le c$ and $M_0\ge1$.
Conditioning on $\xi$ and applying Jensen's inequality to the remaining
randomness in $\chi$ gives
\[
 \E \|a\odot\zeta\|_2
 \ge\frac{q}{4M_0c}\E \|a\odot\xi\|_2
 \ge b_0\frac{q}{M_0c}\mathcal K_c(a).
\]
Since $\sqrt{M_0}$ is a constant multiple of $\ell_0$, the projection
bound follows. No symmetry of the input variables is assumed.

For the Laplace estimate write $Y\colonequals |\langle a,X\rangle|$, $m\colonequals \E Y$,
and $v\colonequals \E Y^2=q\|a\|_2^2\le q\mathcal K_c(a)$.
If $\mathcal K_c(a)=0$, the claim is immediate. Otherwise
$m\ge bq\mathcal K_c(a)/(c\ell_0)>0$ and $v/m\le Bc\ell_0$.
The inequality $1-e^{-u}\ge u/(1+u)$ and Cauchy--Schwarz give
\[
 1-\E e^{-tY}\ge\frac{t m^2}{m+tv}.
\]
Using $\log(1-u)\le-u$ proves the displayed bound. The two simplifications
follow from $c\le1$ and $\ell_0\ge1$.
\end{proof}

\subsection{Threshold comparison rules}
\label{sec:ThresholdComparison}

We use two elementary estimates for a comparison rule mixing greedy
signing with a uniform threshold. Write
$\operatorname{clip}(s,-1,1)\colonequals \max\{-1,\min\{s,1\}\}$ and fix either
sign at zero throughout.

\begin{lemma}[Threshold moments]\label{lem:ThresholdMoment}
For $U\sim\operatorname{Unif}[-h,h]$, $h>0$, $t\ge0$, and $z\in\R$,
\begin{equation}\label{eq:ThresholdIdentity}
 \E[U]{e^{-tz\sgn(z+U)}}
 =\cosh(tz)-\operatorname{clip}(z/h,-1,1)\sinh(tz).
\end{equation}
This is at most one when $th\le2$, and equals $e^{-t|z|}$ when
$|z|\ge h$. Under the hypotheses of \cref{lem:CenteredProjection},
including $|a_i|\le c$ and $\norm a_1\le1$, let $Z\colonequals \inp aX$ and choose
$x\colonequals -\sgn Z$ or $x\colonequals -\sgn(Z+U)$ using a fair coin, with the coin,
$U$, and $X$ mutually independent. There is an absolute $b>0$ such that
\begin{equation}\label{eq:ThresholdScalarDrift}
 \log\E e^{txZ}\le-bt\frac{q}{c\ell_0}\mathcal K_c(a)
 \qquad(t\ell_0\le1,\ th\le2).
\end{equation}
\end{lemma}
\begin{proof}
The identity follows by integrating the sign over the threshold. On
$|z|\le h$, its upper bound is equivalent to
$\tanh(t|z|/2)\le|z|/h$, which follows from $th\le2$; outside that
interval the sign is greedy. Hence the half-mixture has moment at most
$(1+\E e^{-t|Z|})/2$. Apply \eqref{eq:CenteredProjectionLaplace}.
The Laplace bound has the form $e^{-u}$, where $u\ge0$ is bounded
absolutely for $t\ell_0\le1$ because $\mathcal K_c(a)\le2c$.
The inequality $\log((1+e^{-u})/2)\le-b_0u$ on a bounded interval
$u\ge0$ proves \eqref{eq:ThresholdScalarDrift}.
\end{proof}

Alongside the scalar estimate, we need inward drift for individual
coordinates.

\begin{lemma}[Threshold covariance and coordinate drift]
\label{lem:ThresholdCoordinate}
Let the $X_i$ be independent and centered with $\E X_i^2=q$, where
$0<q\le c\le1$. Fix $d,a\in\R^n$ with $|a_i|\le c$,
$\norm a_1\le1$, and $\sgn(a_i)=\sgn(d_i)$. Let $J\ge4$ and assume
$\E X_i^4\le qJ^2/128$. Use the half-mixture of the preceding lemma
with width $h=cJ$.
\begin{enumerate}[label=(\roman*),ref=(\roman*),leftmargin=*]
\item Inward means.
For either branch the inward mean
$\mu_i\colonequals -\sgn(d_i)\E(xX_i)$ is nonnegative, and for the half-mixture
\begin{equation}\label{eq:ThresholdInward}
 |a_i|\ge c/2\quad\Longrightarrow\quad\mu_i\ge\frac{3q}{32J}.
\end{equation}
\item Coordinate exponential moments.
If $\lambda>0$ and $\E(X_i^2e^{\lambda|X_i|})\le C_2q$, then every coordinate satisfies
\begin{equation}\label{eq:ThresholdCosh}
 \frac{\E\cosh(\lambda(d_i+xX_i))}{\cosh(\lambda d_i)}-1
 \le-\lambda\tanh(\lambda|d_i|)\mu_i+\frac{C_2\lambda^2q}{2}.
\end{equation}
In particular the growth is at most $C_2\lambda^2q/2$. If also
$|a_i|\ge c/2$, $\tanh(\lambda|d_i|)\ge1/2$, and
$\lambda\le3/(64C_2J)$, the increment ratio is at most
$-3\lambda q/(128J)$.
\end{enumerate}
\end{lemma}
\begin{proof}
For a centered $Y$, an independent copy $Y'$, and bounded nondecreasing $g$,
\begin{equation}\label{eq:ThresholdCovariance}
 \E{Yg(Y)}=\tfrac12\E{(Y-Y')(g(Y)-g(Y'))}\ge0.
\end{equation}
Conditioning on $S\colonequals \sum_{j\ne i}a_jX_j$ applies this identity to both
sign responses, with the direction determined by $a_i$. Sign alignment
gives $\mu_i\ge0$; if $a_i=0$, independence gives zero directly.

For the quantitative bound reflect $X_i$ if necessary, writing $Y$ for
the reflected variable and $\alpha_i\colonequals\nobreak |a_i|\ge\nobreak c/2$. Since
$\E S^2\le qc\le c^2$, one has $\P(|S|\le cJ/2)\ge3/4$.
The fourth-moment assumption gives
$\E(Y^2;|Y|>J/4)\le q/8$ and $\P(|Y|>J/4)\le q/8$. Consequently
\[
 \E{(Y-Y')^2;|Y|,|Y'|\le J/4}
 \ge2q-4\E{Y^2;|Y|>J/4}-4q\P(|Y|>J/4)\ge q.
\]
On these truncation events and $|S|\le cJ/2$, both $S+\alpha_iY$
and $S+\alpha_iY'$ lie inside the linear interval of
$\operatorname{clip}(\cdot/(cJ),-1,1)$.
Their contribution to \eqref{eq:ThresholdCovariance} is at least
$\alpha_i(3/4)q/(2cJ)\ge3q/(16J)$; all remaining contributions
are nonnegative. The greedy half is also nonnegative, proving
\eqref{eq:ThresholdInward}. Finally, Taylor's theorem and
$[\cosh(\lambda(d_i+y))/\cosh(\lambda d_i)]''\le\lambda^2e^{\lambda|y|}$
give a pointwise second-order bound. Substituting $y=xX_i$, taking
expectations, and using the tilted second-moment bound yields
\eqref{eq:ThresholdCosh}. The stated
contraction follows by substituting \eqref{eq:ThresholdInward}.
\end{proof}

\subsection{Capped geometry and coupled moments}
\label{sec:CappedGeometry}

We now relate the capped weights to the projection scale and combine
potential drift with coordinate drift.

\begin{lemma}[Capped gradient geometry]\label{lem:sparsegeometry}
Let $F\colonequals \F_c$, $k_c\colonequals nc$, and $M\colonequals \norm d_\infty$.
With $S_z\colonequals\nobreak \sum_i(u_i^{3/2}+\nobreak v_i^{3/2})$, one has
\begin{equation}\label{eq:SparseGeometryShared}
 \sqrt cS_z\le C[\mathcal K_c(a)+c/\sqrt{k_c}].
\end{equation}
For $R\ge0$ and a sufficiently large absolute $C_H$,
\begin{equation}\label{eq:SparseHighGeometry}
 M-F\le R,\quad F\ge C_H(\sqrt{k_c}/\eta+R)
 \quad\Longrightarrow\quad\mathcal K_c(a)\ge c_0c/\sqrt{k_c}.
\end{equation}
\end{lemma}
\begin{proof}
Use the KKT notation of \cref{lem:kkt} and split at $|d_i|=\theta/2$.
Low coordinates have both weights at most $4/L^2\le2/n$ and contribute
$O(c/\sqrt{k_c})$ to $\sqrt cS_z$.
On high coordinates an uncapped favored weight is at least nine times
the opposite; if capped, its opposite is at most $2c/9$.
Their pair masses are therefore $O(|a_i|)$.
On the head use $\sqrt c|a_i|^{3/2}\le c|a_i|$; on the tail use
$\sum|a_i|^{3/2}\le\norm{a_{\rm tail}}_2\sqrt{\norm{a_{\rm tail}}_1}$.
This proves \eqref{eq:SparseGeometryShared}.

Put $Q\colonequals \sqrt{k_c}/\eta$. If $\theta\ge M/2$, the high-state condition
and $F\le M+2\sqrt2Q$ force $L\ge c_0C_H\sqrt n$.
The low coordinates then carry mass below $1/2$, leaving constant
gradient $\ell_1$ mass. The head--tail split gives
$\mathcal K_c(a)\ge c_0c/\sqrt{k_c}$.
If $\theta<M/2$, let $\omega$ be the total capped mass.
The Euler identity gives $F\le\omega M+(1-\omega)\theta+2\sqrt2Q$.
If $\omega<1/2$, the condition $M\le F+R$ gives
$M\le4R+8\sqrt2Q$, hence $F\le4R+10\sqrt2Q$.
This contradicts $F\ge C_H(Q+R)$ for sufficiently large $C_H$.
Otherwise at most $\lfloor1/c\rfloor$ capped coordinates carry constant
gradient mass, and the head gives $\mathcal K_c(a)\ge c_0c$.
\end{proof}

When the potential is large and no coordinate lies too far ahead, the
potential term contracts. When a coordinate lies farther ahead, the
coordinatewise terms provide the dominant contraction; the remaining
bounded region contributes the additive term in the moment recursion.

\begin{lemma}[Coupling potential and coordinate drift]\label{prop:foster}
Let $F\ge0$ and fix $\lambda>0$, $B_0\ge0$, $\gamma\ge0$,
$P(d)\colonequals e^{\lambda F(d)}$, and rates $\delta_F,\delta_C>0$ and
$\beta\ge0$ satisfying $\delta_C\ge16\delta_F$,
$\beta\le\delta_F/4$, and $\delta_C\le1$.
Suppose a comparison rule satisfies the following two drift bounds,
conditionally on every applicable history.
\begin{enumerate}[label=(\roman*),ref=(\roman*),leftmargin=*]
\item Potential drift.
The rule satisfies
$\E\Delta P\le\beta P$ at every state, and, for every $R\ge0$,
$\E\Delta P\le-\delta_FP$ when $M-F\le R$ and $F\ge H(R)$,
where $M\colonequals \norm d_\infty$ and $H(R)\ge F(0)$.
\item Coordinate drift.
Write $C_i(d)\colonequals \cosh(\lambda d_i)$. The rule satisfies
\[
\begin{aligned}
 \E C_i(d+xX)&\le(1+\gamma)C_i(d)
     &&\text{at every state},\\
 \E C_i(d+xX)&\le(1-\delta_C)C_i(d)
     &&\text{if }|d_i|>F(d)+B_0.
\end{aligned}
\]
\end{enumerate}
Set
\begin{gather}\label{eq:SharedCoupling}
 \Gamma\colonequals \frac{\delta_C+\gamma}{\delta_C},\qquad
 \varepsilon\colonequals \frac{e^{-\lambda B_0}\delta_F}{4n(\delta_C+\gamma)},\qquad
 R\colonequals B_0+\lambda^{-1}\log(32n\Gamma),\\
 \label{eq:SharedV}
 V(d)\colonequals P(d)+\varepsilon\sum_i\cosh(\lambda d_i).
\end{gather}
With $H\colonequals H(R)$ at the chosen $R$, these hypotheses imply
\begin{equation}\label{eq:SharedFoster}
 \E V(d+xX)\le(1-\delta_F/2)V(d)+\delta_Fe^{\lambda H}.
\end{equation}
Starting at $d_0=0$, the deterministic rule minimizing the two successor
values of $V$ satisfies the same recursion. For every prescribed
deterministic integer $T\ge0$ and every $r\ge0$, its normalized discrepancy
$\norm{d_T}_\infty$ is at most $H+\lambda^{-1}[\log(C/\varepsilon)+r]$
with probability at least $1-e^{-r}$.
\end{lemma}
\begin{proof}
Write $\Psi\colonequals V-P$. Coordinates with $|d_i|\le F+B_0$ contribute at most
$n e^{\lambda(F+B_0)}$ before multiplication by $\varepsilon$.
Thus $\E\Delta\Psi\le-\delta_C\Psi+(\delta_F/4)P$.
If $M-F>R$, then $\Psi>4(\delta_F/\delta_C)P$, and
\[
 \E\Delta V\le(\delta_F/2)P-\delta_C\Psi
 \le-(3\delta_F/2)P-(\delta_C/2)\Psi.
\]
If $M-F\le R$ and $F\ge H$, use the contraction of $P$.
Otherwise $F<H$, and adding $(\delta_F/2)V$ leaves at most
$\delta_FP\le\delta_Fe^{\lambda H}$. This proves the recursion.
For each fixed history and fresh input, the minimum of the two successor
values is no larger than their average under the auxiliary randomized
rule. Averaging over the fresh input gives the same conditional bound
along the deterministic rule's own history.
Since $V(0)\le2e^{\lambda H}$ and
$V(d)\ge(\varepsilon/2)e^{\lambda M}$, \cref{lem:terminal} finishes the proof.
\end{proof}

\section{Online discrepancy in the dense regime}
\label{sec:OnlineDiscrepancy}

This section proves \cref{thm:OnlineDiscrepancy} and develops finer
confidence profiles for centered inputs. We first establish a conditional
guarantee allowing past-dependent product laws and nonzero means;
\cref{thm:OnlineDiscrepancy} follows by choosing the confidence parameter.
We then explain why uniform tail control is necessary and derive two
confidence profiles using different regularization scales.

\subsection{Dense drift and conditional inputs}

We prove a conditional version of \cref{thm:OnlineDiscrepancy}.
Let $\Fc_{t-1}$ contain the complete history before the next arrival.
Conditional on this history, assume that the coordinates of $v_t$ are
independent and, almost surely,
\begin{equation}\label{eq:DenseConditionalAssumptions}
 \operatorname{Var}(v_t(i)\mid\Fc_{t-1})=1,\qquad
 \E\bigl[e^{|v_t(i)|/\sigma}\mid\Fc_{t-1}\bigr]\le2.
\end{equation}
The conditional laws and means may depend on the past.
The exponential-moment assumption applies to the uncentered coordinates
$v_t(i)$, so it also controls their conditional means.

\begin{proposition}[Conditional dense guarantee]\label{prop:DenseConditional}
Under \eqref{eq:DenseConditionalAssumptions}, a deterministic rule depending
only on $n,\sigma$ satisfies, for every prescribed deterministic $T\ge1$
and every $r\ge0$,
\begin{equation}\label{eq:DenseConfidenceCoarse}
 \P\left(\norm{\sum_{t=1}^T x_tv_t}_\infty
        >C[\sigma^4\sqrt n+\sigma r]\right)\le e^{-r}.
\end{equation}
It uses $O(n\log(2n\sigma))$ real arithmetic operations per arrival.
\end{proposition}

Use the potential $\F=\F_1$ with cap $c=1$, and set
$p=1$, $Y_t\colonequals v_t/\sigma$, $\nu\colonequals \sigma^{-1}$, and
$d_t\colonequals \sum_{s\le t}x_sY_s$. All one-step expectations below are conditional
on the past, so the current state and the next coordinate laws are fixed.

The next two lemmas control the regularization remainder and the negative
Laplace contribution through the same gradient norm $\norm a_2$.

\begin{lemma}[Dense geometry]\label{lem:dense}
One has $\norm d_\infty\le\F(d)$. If $\F(d)\ge6\sqrt n/\eta$, then,
writing $g\colonequals \norm a_2$,
\[
 g\ge1/(4\sqrt n),\qquad S_z\le4g.
\]
\end{lemma}
\begin{proof}
Placing unit mass on a signed coordinate of maximum magnitude proves the
lower bound for $\F$. Reflecting each negative coordinate of $d$ and
exchanging the corresponding weights $u_i,v_i$ preserves $\F(d)$,
$\norm d_\infty$, $g$, and $S_z$. We may therefore assume $d_i\ge0$,
so that $u_i\ge v_i$.
The upper bound for $\F$ gives $M\ge2\sqrt n/\eta$ in the indicated region.
For $c=1$, all dual weights are strictly below one and $L>\eta M$, whence
$L\ge2\sqrt n$. Thus, $\norm v_2\le1/(4\sqrt n)$, while
$\norm u_2\ge\norm u_1/\sqrt n\ge1/(2\sqrt n)$.
It follows that $\norm{u-v}_2\ge\norm u_2/2\ge\norm{u+v}_2/4$.
Since $\norm{u+v}_1=1$, this is at least $1/(4\sqrt n)$.
Cauchy--Schwarz gives
$S_z\le\sum_i(u_i+v_i)^{3/2}\le\norm{u+v}_2\le4g$.
\end{proof}

\begin{lemma}[Dense negative Laplace estimate]\label{lem:DenseLaplace}
If the coordinates of $Y$ are independent,
$\operatorname{Var}(Y_i)=\nu^2$, and $\E e^{|Y_i|}\le2$, then for
$g\colonequals \norm a_2\le1$ and $0\le t\le1$,
\begin{equation}\label{eq:DenseLaplace}
 \log\E e^{-t|\inp aY|}\le-ct\nu^3g.
\end{equation}
The coordinates need not be centered.
\end{lemma}
\begin{proof}
First suppose they are centered. The exponential moment gives
$\E|Y_i|^6\le2\cdot6!$, and hence
$\E Y_i^4\le(\E Y_i^2\,\E Y_i^6)^{1/2}\le C\nu$.
For $Z\colonequals |\inp aY|$, independence and centering give
$\E Z^2=\nu^2g^2$ and $\E Z^4\le C\nu g^4$.
Interpolation yields $\E Z\ge m\colonequals c\nu^{5/2}g$.
For $t,g,\nu\le1$,
$m+t\nu^2g^2\le C\nu^2g$, so \eqref{eq:negative} gives
\eqref{eq:DenseLaplace}. The case $g=0$ is immediate.

For noncentered coordinates, take an independent copy $Y'$ and put
$S\colonequals \inp aY$, $S'\colonequals \inp a{Y'}$. The triangle inequality gives
\[
 (\E e^{-t|S|})^2
 =\E e^{-t(|S|+|S'|)}\le\E e^{-t|S-S'|}.
\]
The independent coordinates $Y_i-Y_i'$ are centered, have variance
$2\nu^2$, and have uniformly bounded sixth moments.
The preceding calculation therefore bounds the logarithm on the right
by $-ct\nu^3g$. Dividing by two proves the claim.
The copy is used only in the proof.
\end{proof}

\begin{proof}[Proof of \cref{prop:DenseConditional} and \cref{thm:OnlineDiscrepancy}]
Use \eqref{eq:parameters} and an upper root approximation with
$\widehat L-L\le c_e\nu^3/\sqrt n$.
Let $\widehat a$ be the corresponding approximate gradient and put
$e\colonequals a-\widehat a$. At each arrival, choose
$x\colonequals-\sgn\inp{\widehat a}{Y}$, with a fixed sign at ties.
By \cref{lem:kkt}, $\norm e_1\le2c_e\nu^3/\sqrt n$.
The Fenchel and signing-error inequalities, followed by H\"older, give
\[
 \log\E e^{\lambda\Delta\F}
 \le-c\lambda\nu^3g+C\lambda\eta S_z+C\lambda\norm e_1,
\]
where \cref{lem:DenseLaplace} bounds the negative term.
Choose the absolute $\lambda>0$ small enough for the product moments.
In the region $\F\ge H\colonequals 6\sqrt n/\eta$, \cref{lem:dense} gives
$S_z\le4g$ and $g\ge1/(4\sqrt n)$.
Taking $c_\eta,c_e$ small absorbs both positive terms and yields
\[
 \E e^{\lambda\Delta\F}\le e^{-\delta},
 \qquad \delta\colonequals c\nu^3/\sqrt n.
\]
In the low region, Lipschitz continuity gives
\[
 \E e^{\lambda\F(d+xY)}
 \le e^{\lambda H}\E e^{\lambda\norm Y_\infty}
 \le2n e^{\lambda H}.
\]
These conditional estimates apply at every time. Thus
$V\colonequals e^{\lambda\F}$ satisfies \cref{lem:terminal} with contraction
$1-e^{-\delta}\ge\delta/2$, injection $2ne^{\lambda H}$, and $q_0=1$.
Since $\F(0)<H$, for every $r\ge0$ the normalized threshold is at most
\[
 H+\lambda^{-1}[\log(5n/\delta)+r]
 \le C[\sigma^3\sqrt n+r].
\]
Multiplying by $\sigma$ proves \cref{prop:DenseConditional}.
Taking $r=c\sigma^3\sqrt n$ proves \cref{thm:OnlineDiscrepancy}.
The arithmetic cost is verified in \cref{sec:implementation}.
\end{proof}

\begin{remark}[Necessity of uniform tail control]
\label{rem:UniformTailControl}
Some uniform tail control is necessary, even for \iid entries.
For a prescribed $T\ge2$, let each entry be zero with probability $1-1/T$
and take each value $\pm\sqrt T$ with probability $1/(2T)$.
It is symmetric, centered, and has unit variance.
The number $N_i$ of nonzero arrivals in coordinate $i$ is $\Bin(T,1/T)$, so
\[
 \P{N_i\text{ is odd}}=\frac{1-(1-2/T)^T}{2}=\Theta(1).
\]
When $N_i$ is odd, the signed sum in coordinate $i$ is an odd multiple of
$\sqrt T$, regardless of the signing, and therefore has magnitude at least
$\sqrt T$.
Independence across coordinates therefore gives
\[
 \P\left(\min_{x\in\{\pm1\}^T}
    \norm{\sum_{t=1}^T x_tv_t}_\infty\ge\sqrt T\right)
 \ge1-e^{-cn},
\]
where $c>0$ is an absolute constant.
Taking $T\gg n$ excludes a horizon-independent $O(\sqrt n)$ bound with
constants uniform over all unit-variance laws. The example is consistent
with \cref{thm:OnlineDiscrepancy}: its $\psi_1$-norm is
$\sqrt T/\log(T+1)$, which grows with $T$.
\end{remark}

The tail-parameter and confidence obstructions are proved together in
\cref{sec:TailSharpness}; in particular, \cref{prop:SigmaSharpness}
shows why the fourth power in the main guarantee is necessary at its
stated confidence level.

\subsection{Finer confidence profiles for centered inputs}
\label{sec:DenseConfidence}

Centering allows the variance-sensitive remainder from
\cref{sec:VarianceRefinements} to improve the tail-parameter dependence at
other confidence levels. These profiles use different regularization
parameters from the main theorem.

\begin{theorem}[Dense confidence profiles]\label{thm:DenseConfidence}
Let $\alpha\in\{1,2\}$ and $\ell\colonequals (\log(e\sigma))^{1/\alpha}$.
Conditional on the complete past, suppose the coordinates of $v_t$ are
independent, centered, have variance one, and satisfy
$\E\bigl[e^{|v_t(i)/\sigma|^\alpha}\mid\Fc_{t-1}\bigr]\le2$ almost surely.
There are deterministic online rules depending only on $n,\sigma,\alpha$
such that, for every prescribed deterministic $T\ge1$ and $r\ge0$,
\begin{align}
 \P\left(\norm{\sum_{t=1}^T x_tv_t}_\infty
       >C\sigma[\ell^2\sqrt n+r]\right)&\le e^{-r},
       \label{eq:DenseProfileFixed}\\
 \P\left(\norm{\sum_{t=1}^T \widetilde x_tv_t}_\infty
       >C\sigma\ell[\sqrt n+r]\right)&\le e^{-r}.
       \label{eq:DenseProfileModerate}
\end{align}
The two rules use different parameter choices. Each requires
$O(n\log(2n\sigma))$ real arithmetic operations per arrival.
\end{theorem}

The second rule improves the leading term by a factor $\ell$, at the cost
of an additional factor $\ell$ in the coefficient of $r$.

\begin{proof}
Normalize as above. Apply \cref{lem:CenteredProjection} with
$q=\nu^2$, $c=1$, and $\ell_0=\ell$; the moments follow from
\cref{lem:variancerefinement}, and $\mathcal K_1(a)\asymp g$.
This gives Laplace rate $-ct\nu^2g/\ell^2$ for $t\le1$ and
$-ct\nu^2g/\ell$ for $t\le c_*/\ell$.
For \eqref{eq:DenseProfileFixed}, choose
$\eta=c_\eta/\ell^2$, $\lambda=c_\lambda$, and
$\norm e_1\le c_e\nu^2/(\ell^2\sqrt n)$.
For \eqref{eq:DenseProfileModerate}, choose
$\eta=c_\eta/\ell$, $\lambda=c_\lambda/\ell$, and
$\norm e_1\le c_e\nu^2/(\ell\sqrt n)$.
H\"older, \eqref{eq:varianceremainder}, and \cref{lem:dense} absorb the
remainder and error in the region $\F(d)\ge H\colonequals 6\sqrt n/\eta$.
In this region, both choices give a logarithmic drift bounded above by
$-c\nu^2g/\ell^2$.
With $V(d)\colonequals e^{\lambda\F(d)}$ and
$\delta\colonequals c\nu^2/(\ell^2\sqrt n)$, this yields
\[
 \E\bigl[V(d_t)\mid\Fc_{t-1}\bigr]
 \le e^{-\delta}V(d_{t-1})
 \qquad\text{if }\F(d_{t-1})\ge H.
\]

For the low region, a sharper injection avoids an unnecessary
$\log\sigma$ loss. Globally, dropping the negative Laplace contribution
and using $S_z\le1$ gives
\[
 \E\bigl[V(d_t)\mid\Fc_{t-1}\bigr]
 \le e^\beta V(d_{t-1}),
 \qquad \beta\colonequals C\nu^2/\ell^2.
\]
With $\rho\colonequals 1-e^{-\delta}$, the high and low estimates combine to give
\[
 \E\bigl[V(d_t)\mid\Fc_{t-1}\bigr]
 \le(1-\rho)V(d_{t-1})+(e^\beta-1+\rho)e^{\lambda H}.
\]
Since $\beta$ is bounded by an absolute constant and $\rho\ge\delta/2$,
\[
 \frac{e^\beta-1}{\rho}
 \le C\frac{\beta}{\rho}
 \le C\sqrt n.
\]
Together with $\F(0)<H$, \cref{lem:terminal} therefore bounds the normalized threshold by
$H+\lambda^{-1}[\log(C(1+\sqrt n))+r]$.
The logarithmic correction has no additional dependence on $\sigma$.
The scale $\eta^{-1}$ determines the potential threshold $H$, while
$\lambda^{-1}$ determines the coefficient of the confidence parameter $r$.
The two parameter choices, followed by multiplication by $\sigma$, give
the assertions. Root tolerances and arithmetic costs are detailed in
\cref{sec:implementation}.
\end{proof}

For either value of $\alpha$, taking $r=\sqrt n$ in
\eqref{eq:DenseProfileModerate} gives a discrepancy bound of
$C\sigma\ell\sqrt n$ with failure probability at most $e^{-\sqrt n}$.
The factor $\sigma\log(e\sigma)$, and its sub-Gaussian counterpart
$\sigma\sqrt{\log(e\sigma)}$, are necessary for horizon-independent
guarantees at fixed dimension and constant confidence; see
\cref{rem:TailParameterFactors}.

\section{Online discrepancy in the sparse regime}
\label{sec:SparseOnlineDiscrepancy}

We prove \cref{thm:SparseOnlineDiscrepancy} and the sub-Gaussian refinements
of \cref{thm:SparseRefinements}. We first establish two confidence profiles
for centered sparse inputs, then develop the sub-Gaussian envelope and
compare the resulting guarantees with the lower bounds.

We retain the physical mask probability $p=k/n$ throughout the analysis.
A generic dense reduction would normalize by $\sqrt p$. For marks with
$\psi_1$-norm at most $\sigma$, the generic estimate
$\norm{\xi_iw_i/\sqrt p}_{\psi_1}\le\sigma/\sqrt p$ gives, by
\cref{prop:DenseConditional}, the leading raw scale
\[
 \sqrt p\left(\frac{\sigma}{\sqrt p}\right)^4\sqrt n
 =p^{-2}\sigma^4\sqrt k.
\]
The arguments below work directly for every
$1\le k\le n$, including $p=1$, and require only centered marks.
They also allow the next product law to depend on the complete past.

Fix an integer $n\ge1$ and a real $1\le k\le n$, and put
$p\colonequals k/n$. For $\alpha\in\{1,2\}$ and $\sigma\ge1$, set
\[
 \nu\colonequals \sigma^{-1},\qquad \ell\colonequals (\log(e\sigma))^{1/\alpha}.
\]
Let $\Fc_{t-1}$ contain the full history before arrival $t$. Almost surely,
conditionally on $\Fc_{t-1}$, the marks $w_t(i)$ are independent and satisfy
\begin{equation}\label{sp:eq:assumptions}
 \E{w_t(i)\mid\Fc_{t-1}}=0,\qquad
 \E{w_t(i)^2\mid\Fc_{t-1}}=1,\qquad
 \E{e^{|w_t(i)/\sigma|^\alpha}\mid\Fc_{t-1}}\le2.
\end{equation}
The fresh masks are independent $\operatorname{Ber}(p)$ variables,
independent of the fresh marks conditionally on the same history. The input
is $v_t\colonequals \xi_t\odot w_t$. No symmetry or identical distribution
across coordinates is required; the conditional mark laws may depend on
the past. The auxiliary randomness in the randomized comparison rules is
independent of the fresh input and the past.
Constants below are positive and absolute,
uniform over $\alpha\in\{1,2\}$; their values can change between displays.

\begin{theorem}[Centered sparse confidence profiles]\label{thm:SparseConfidence}
Under \eqref{sp:eq:assumptions}, there are deterministic online signing rules,
independent of the horizon and input laws, such that, for every prescribed
deterministic $T\ge1$ and every $r\ge0$,
\begin{align}
 \P\left\{\left\|\sum_{t\le T}x_tv_t\right\|_\infty
 >C\sigma[\ell^2\sqrt k+\log(2n)+r]\right\}&\le e^{-r},
 \label{eq:SparseProfileFixed}\\
 \P\left\{\left\|\sum_{t\le T}\widetilde x_tv_t\right\|_\infty
 >C\sigma\ell[\sqrt k+\log(2n)+r]\right\}&\le e^{-r}.
 \label{eq:SparseProfileModerate}
\end{align}
The two bounds use different parameter choices and need not hold for the same rule.
\end{theorem}

The randomized rules below establish drift for an explicit coupled potential.
Minimizing its successor values gives the deterministic implementations.
These are terminal-time bounds, not one event controlling all prefixes.

The profiles are sparse counterparts of \cref{thm:DenseConfidence}, with
$\sqrt k$ in place of $\sqrt n$ and an additional $\log(2n)$ correction.
For $k=n$, this correction is absorbed by $\sqrt n$, recovering the dense
profiles up to absolute constants. The second profile reduces the leading
coefficient from $\sigma\ell^2$
to $\sigma\ell$, at the cost of multiplying the $\log(2n)+r$ correction
by $\ell$.

\begin{corollary}[A coarse sparse confidence profile]\label{cor:SparseAllLevels}
Under \eqref{sp:eq:assumptions} with $\alpha=1$, for each $1\le k\le n$
there is a deterministic rule depending only on $n,k,\sigma$ such that,
for every prescribed deterministic $T\ge1$ and every $r\ge0$,
\begin{equation}\label{eq:SparseConfidenceCoarse}
 \P\left(\norm{\sum_{t\le T}x_tv_t}_\infty>
 C[\sigma^4\sqrt k+\sigma\log(2n)+\sigma r]\right)\le e^{-r}.
\end{equation}
\end{corollary}
\begin{proof}
Use \eqref{eq:SparseProfileFixed} and
$[\log(e\sigma)]^2\le C\sigma^3$.
Taking $r=c\sigma^3\sqrt k$ and absorbing $\log(2n)$ when
$k\ge K_0\log^2(2n)$ also gives \cref{thm:SparseOnlineDiscrepancy}.
\end{proof}

\subsection{A centered comparison rule}

Use normalized inputs $X_t\colonequals v_t/\sigma$ and normalized state
$d_t\colonequals \sum_{s\le t}x_sX_s$. Let $F\colonequals \F_p$ be the common potential in \eqref{eq:potential},
with geometric cap equal to $p$.
Write $a\colonequals \nabla F(d)=u-v$ and take $b\colonequals a$ in the analytical
comparison rule. Given the new input $X$, draw a branch coin and an
independent uniform random variable:
\begin{equation}\label{sp:eq:rule}
 x\colonequals \begin{cases}
 -\sgn\inp bX,&\text{with probability }\theta,\\
 -\sgn(\inp bX+U),\quad U\sim\operatorname{Unif}[-pJ,pJ],
     &\text{with probability }1-\theta.
 \end{cases}
\end{equation}
Use any one fixed value in $\{-1,1\}$ for $\sgn(0)$. Uniform-threshold ties
have probability zero. The parameter choices are
\[
 (J,\theta,\lambda,\eta,\rho)\colonequals
 \begin{cases}
 (K\ell,1/2,c_\lambda/\ell,c_\eta/\ell,\nu^2/\ell),
    &\text{for \eqref{eq:SparseProfileModerate}},\\
 (L,a_0/\ell,\epsilon/L,c_\eta/\ell^2,\nu^2/\ell^2),
    &\text{for \eqref{eq:SparseProfileFixed}}.
 \end{cases}
\]
Here $K,L$ are sufficiently large absolute constants; $a_0,\epsilon,
c_\lambda,c_\eta$ are sufficiently small absolute constants, fixed in the
order justified below. In the constant-exponent case, $\lambda=\epsilon/L$
is an absolute constant. The parameter $\lambda$ controls the exponential
moments and also enters the coupled potential used by the deterministic rule.

For both choices, fix $B_0\colonequals C_0(2/\eta+1/\lambda)$ with a sufficiently
large absolute $C_0$. Its order is $\ell$ and $\ell^2$, respectively.
By \cref{lem:kkt}, the favored weight is already capped when $|d_i|>F(d)$.
The additional buffer ensures that coordinates with $|d_i|>F(d)+B_0$
satisfy $\lambda|d_i|\ge1$ and, in the constant-exponent case,
$|d_i|\gtrsim\ell^2$, as needed for cosh contraction.

The exact deterministic rule minimizes $V(d+xX)$ over $x\in\{-1,1\}$,
with fixed tie-breaking. The potential $V$ and its coupling weight are
specified in \eqref{eq:SharedV} and \eqref{eq:SharedCoupling}.
Certified approximate comparisons preserve the terminal bounds; see
\cref{sec:implementation}.
The algorithm does not sample an independent copy of a mark or estimate a
conditional law. In the auxiliary comparison rule, the branch coin can be
drawn after observing $X$, because its conditional probability is fixed in advance.

\subsection{Centered projections and coordinate drift}

All expectations below are conditional on a fixed past. The moments in
\cref{lem:variancerefinement} apply to $Y_i\colonequals w_i/\sigma$. In addition,
for each fixed sufficiently large absolute $K$, tail integration gives
$\E(Y_i^2\mathbf1_{\{|Y_i|>K\ell\}})\le\varepsilon_0\nu^2$,
where $\varepsilon_0$ can be made arbitrarily small by increasing $K$.
The fixed small-exponent moments $\E e^{c|Y_i|}$ and
$\E(Y_i^2e^{c|Y_i|})$ are bounded absolutely.

Write $\mathcal K_p$ for the common functional in
\cref{sec:SharedSparseGeometry}, with $c=p$.
The masked moments satisfy $\E X_i^2=p\nu^2$ and
$\E X_i^4\le Cp\nu^2\ell^2$. Applying \cref{lem:CenteredProjection}
with $q=p\nu^2$, $c=p$, $\ell_0=\ell$, and using
\cref{lem:BernoulliNorm} for the second moment gives
\begin{equation}\label{sp:eq:headtail}
 \E|\inp aX|\ge c\nu^2\mathcal K_p(a)/\ell,\qquad
 \E(\inp aX)^2=p\nu^2\norm a_2^2\le p\nu^2\mathcal K_p(a).
\end{equation}

\paragraph{Moderate parameters.}
Take $J=K\ell$ and $\theta=1/2$ in \eqref{sp:eq:rule}.
For sufficiently large absolute $K$, the fourth moment is at most
$(p\nu^2)J^2/128$. The tilted second moment in
\cref{lem:variancerefinement} permits \cref{lem:ThresholdCoordinate}
with $q=p\nu^2$, $c=p$, and $\lambda=c_\lambda/\ell$.
Thus the inward mean is at least $3p\nu^2/(32K\ell)$ when
$|a_i|\ge p/2$, while every coordinate has expected relative cosh increase at most
$\gamma\colonequals Cp\lambda^2\nu^2$. When $|d_i|>F(d)+B_0$, \cref{lem:kkt}
gives $|a_i|\ge p/2$, and $\lambda B_0\ge1$. Choosing $c_\lambda$
sufficiently small gives contraction by a factor $1-\delta_C$, with
$\delta_C\colonequals c_C\lambda\rho p$, where $\rho\colonequals \nu^2/\ell$.

For $Z\colonequals \inp aX$, apply \cref{lem:ThresholdMoment} with $h=pK\ell$.
After decreasing the admissible absolute constant $c_*$ so that
$c_*\le\min\{1,2/K\}$, both $t\ell\le1$ and $tpK\ell\le2$ hold throughout
$0\le t\le c_*/\ell$. Thus
\begin{equation}\label{sp:eq:modscalar}
 \log\E e^{txZ}\le-ct\rho\mathcal K_p(a)
 \qquad(0\le t\le c_*/\ell).
\end{equation}
Choose $c_\lambda$ with $2c_\lambda\le c_*$, so that the scalar estimate
applies at $2\lambda$. Taking $c_\eta$ sufficiently small for the remainder
estimate, H\"older with two factors in \eqref{eq:fenchel} and
\cref{lem:variancerefinement} yields
\begin{equation}\label{sp:eq:Fdrift}
 \log\E e^{\lambda\Delta F}
 \le-c\lambda\rho\mathcal K_p(a)
       +C\lambda\eta\nu^2\sqrt pS_z.
\end{equation}

\paragraph{Constant exponent: direct nonsymmetric coordinate analysis.}

To deduce contraction from the preceding coordinate estimate, we take
$\lambda=O(\ell^{-1})$. For \eqref{eq:SparseProfileFixed}, we instead establish
coordinate contraction with constant $\lambda$, then recover scalar
contraction despite the smaller greedy fraction.
Fix a sufficiently small absolute $\epsilon>0$, then a sufficiently large
absolute $L$, and put $\lambda=\epsilon/L$. At a coordinate with
$|b_i|\ge p/2$, reflect $(d_i,b_i,Y_i)$ if necessary; the projection, the
cosh term, and the centered moment assumptions are preserved. By sign
alignment, we may therefore assume $b_i>0$ and $d_i>0$.
Put $c_i\colonequals b_i/p\in[1/2,1]$ and
$S\colonequals \sum_{j\ne i}b_jX_j$, $s\colonequals S/p$. Then $s$ is centered,
independent of $Y\colonequals Y_i$, and $\E s^2\le\nu^2$. Define
\[
 f_d(y)\colonequals \E_s\operatorname{clip}((s+c_i y)/L,-1,1),\qquad
 f_g(y)\colonequals \E_s\sgn(s+c_i y),\qquad f_0\colonequals f_d(0).
\]
Put $H\colonequals |Y|$ and $A_*\colonequals 4L$, and introduce the
nonnegative comparison function
\begin{equation}\label{sp:eq:wshape}
 w(h)\colonequals \begin{cases}h^2/(64L),&h\le A_*,\\
 h/(1+\lambda h),&h>A_*.
 \end{cases}
\end{equation}
The first branch includes $h=A_*$, also when $H$ has an atom there.
We compare both signing branches with the \emph{same} $\E w(H)$;
separate worst-case bounds would lose an extra factor $\ell$.

\begin{lemma}[Uniform-threshold contraction and greedy cost]\label{sp:lem:sharpcoordcore}
Choose a sufficiently small absolute $\epsilon>0$ and then a
sufficiently large absolute $L$. The constants can be fixed so that
\begin{align}
 \E{\cosh(\lambda Y)-1-f_d(Y)\sinh(\lambda Y)}
   &\le-c\lambda\E w(H),\label{sp:eq:dithercontract}\\
 \E{\cosh(\lambda Y)-1-f_g(Y)\sinh(\lambda Y)}
   &\le C\lambda\ell\E w(H),\label{sp:eq:greedycost}\\
 \E w(H)&\ge c\nu^2/\ell^2.\label{sp:eq:wlower}
\end{align}
\end{lemma}
\begin{proof}
First, centering of $s$ gives
\begin{equation}\label{sp:eq:fzero}
 |f_0|\le\E|\operatorname{clip}(s/L,-1,1)-s/L|\le\nu^2/L^2.
\end{equation}
For $|y|\le L/2$, almost everywhere,
$f_d'(y)=(c_i/L)\P(|s+c_i y|<L)\ge1/(4L)$.
By monotonicity, with $q(y)\colonequals \sgn(y)(f_d(y)-f_0)$,
\begin{equation}\label{sp:eq:fdlocal}
 q(y)\ge\tfrac18\min(|y|/L,1),\qquad
 q(y)\ge |y|/(32L)\quad(|y|\le A_*).
\end{equation}
For $h\colonequals |y|>A_*$, a wrong threshold sign requires
$|s|\ge c_i h-L\ge h/4$, so
\begin{equation}\label{sp:eq:fdlarge}
 1-\sgn(y)f_d(y)\le32\nu^2/h^2.
\end{equation}
The greedy response similarly satisfies
$1-\sgn(y)f_g(y)\le8\nu^2/h^2$ for $h>0$.

Write $D_d(y)\colonequals \cosh(\lambda y)-1-f_d(y)\sinh(\lambda y)$.
Since $\E Y=0$, adding $\lambda f_0y$ to $D_d(y)$ does not change its expectation.
For $h\le A_*$ the resulting expression is
\[
 \cosh(\lambda h)-1-q(y)\sinh(\lambda h)
                    -f_0[\sinh(\lambda y)-\lambda y].
\]
Use $\tanh(u/2)\le u/2$ and $\sinh u-u\le u^2\sinh u$.
For small $\epsilon=\lambda L$, the negative term controlled by
\eqref{sp:eq:fdlocal} absorbs the cosh term and the absolute value of the
centering remainder, giving $-\lambda h^2/(64L)=-\lambda w(h)$ as an upper bound.
For $h>A_*$, \eqref{sp:eq:fdlarge} gives
\[
 D_d(y)+\lambda f_0y\le-\frac{\lambda h}{1+\lambda h}
  +\frac{32\nu^2\sinh(\lambda h)}{h^2}
  +\frac{\lambda\nu^2h}{L^2}.
\]
The crossing error is bounded by
\begin{align}
 \nu^2\E\left[\frac{\sinh(\lambda H)}{H^2};H>A_*\right]
 &\le\lambda\nu^2\E{e^{\lambda H}/H;H>A_*}\notag\\
 &\le\lambda\nu^3
       \left(\E{e^{2\lambda H}/H^4;H>A_*}\right)^{1/2}
 \le C\lambda\nu^3/A_*^2.\label{sp:eq:crossing}
\end{align}
Here Cauchy--Schwarz uses the factors $H$ and $e^{\lambda H}/H^2$.
The centering error is at most $\lambda\nu^4/(L^2A_*)$.
Weighted Cauchy--Schwarz and \cref{lem:variancerefinement} give
\begin{equation}\label{sp:eq:wdenom}
 \E w(H)\ge\frac{(\E H^2)^2}{\E{H^4/w(H)}}
 \ge\frac{c\nu^2}{64L+\ell+\lambda\ell^2}.
\end{equation}
The quotient at zero is defined as zero. The ratios of the two errors to
$\lambda\E w(H)$ are at most constants times
\[
 \frac{\nu(64L+\ell+\lambda\ell^2)}{L^2},\qquad
 \frac{\nu^2(64L+\ell+\lambda\ell^2)}{L^3}.
\]
Fix a sufficiently small $\epsilon$ for the local algebra and then choose
$L$ sufficiently large, retaining $\lambda=\epsilon/L$.
Both ratios become arbitrarily small uniformly in $\nu$, since
$\sup_{0<\nu\le1}\nu\ell^j<\infty$ for each fixed $j$.
This proves \eqref{sp:eq:dithercontract}; \eqref{sp:eq:wdenom} also proves
\eqref{sp:eq:wlower} because $L$ is now an absolute constant.

For the greedy branch put $g_0\colonequals f_g(0)$ and add $\lambda g_0y$.
Monotonicity gives $\sgn(y)(f_g(y)-g_0)\ge0$, hence for $h\le A_*$
the modified integrand is at most
$\cosh(\lambda h)-1+|\sinh(\lambda y)-\lambda y|
 \le C\lambda^2h^2$.
For $h>A_*$ it is at most
$\lambda h+8\nu^2\sinh(\lambda h)/h^2$.
The small-mark contribution and \eqref{sp:eq:crossing} are at most
$C\lambda\E w(H)$. For the remaining tail mean, split again at
$R_*\colonequals K\ell$ with absolute $K\ge A_*$ sufficiently large. For
$A_*<H\le R_*$, $H=(1+\lambda H)w(H)\le C\ell w(H)$.
Tail integration above $R_*$ gives
$\E(H;H>R_*)\le C\nu^2/\ell\le C\ell\E w(H)$.
This proves \eqref{sp:eq:greedycost}.
\end{proof}

Choose $\theta=a_0/\ell$ with sufficiently small absolute $a_0>0$.
Lemma~\ref{sp:lem:sharpcoordcore} gives
\begin{equation}\label{sp:eq:sharpinfcoord}
 \E{\cosh(\lambda Y)-1-
       \{(1-\theta)f_d(Y)+\theta f_g(Y)\}\sinh(\lambda Y)}
 \le-c\lambda\nu^2/\ell^2.
\end{equation}
Write $m(y)\colonequals (1-\theta)f_d(y)+\theta f_g(y)$. Conditional on
an active mask, the averaged sign satisfies
$\E{x\mid Y=y,\xi_i=1}=-m(y)$. Expanding the cosh term and restoring
the mask gives the exact identity
\[
 \frac{\E\cosh(\lambda(d_i+xX_i))}{\cosh(\lambda d_i)}
 =1+p\E\bigl[\cosh(\lambda Y)-1
       -\tanh(\lambda d_i)m(Y)\sinh(\lambda Y)\bigr].
\]
At finite positive discrepancy the factor $\tanh(\lambda d_i)$
therefore introduces a defect of at most
\[
 2e^{-2\lambda|d_i|}\E\sinh(\lambda H)
 \le C\lambda\nu e^{-2\lambda|d_i|},
\]
by Cauchy--Schwarz and the fixed exponential moment.
It is absorbed when $|d_i|\ge C\ell^2$, since
$\log(1/\nu)\le\ell^2$. With $\eta=c_\eta/\ell^2$ and small enough
$c_\eta$, $|d_i|>F+B_0$ ensures this condition and the cap.
Using the identity above gives
\begin{equation}\label{sp:eq:sharpcoord}
 \E\cosh(\lambda(d_i+xX_i))
 \le(1-\delta_C)\cosh(\lambda d_i),\qquad
 \delta_C\colonequals c_Cp\lambda\rho,\quad\rho\colonequals \nu^2/\ell^2.
\end{equation}
For every coordinate, \eqref{eq:ThresholdCovariance} gives nonnegative inward mean
for both branches. Taylor's formula and the fixed tilted second moment
bound the expected relative cosh increase by
$\gamma\colonequals Cp\lambda^2$.

\paragraph{Constant exponent: recovering the full potential rate.}
Since both branches sign large projections greedily, the following estimate
avoids an additional loss from the small greedy probability.

\begin{lemma}[Large projections cost no mixture probability]\label{sp:lem:sharpscalar}
For the constant-exponent rule and any $b$ satisfying the cap and mass
bounds, for every $0<t\le t_0$ with an absolute $t_0>0$,
\begin{equation}\label{sp:eq:sharpscalar}
 \log\E e^{t x\inp bX}\le-ct\nu^2\mathcal K_p(b)/\ell^2.
\end{equation}
\end{lemma}
\begin{proof}
First fix $X$ and average only over the comparison randomness.
Write $z\colonequals |\inp bX|$. By \cref{lem:ThresholdMoment}, the threshold
branch has scalar moment at most one when $tpL\le2$.
For $z>pL$, \emph{both} branches sign greedily and the moment is $e^{-tz}$.
For $z\le pL$, the greedy fraction supplies deficit at least
$\theta(1-e^{-tz})\ge ct z/\ell$ if $tL$ is bounded.
For $pL<z\le R\ell$, the common greedy deficit is
$1-e^{-tz}\ge tz/(1+tz)\ge c_Rtz/\ell$.
All deficits are nonnegative. Averaging now over the fresh input,
\eqref{sp:eq:headtail} gives
\[
 \E{z;z>R\ell}\le\frac{\E z^2}{R\ell}
 \le\frac{p\nu^2\mathcal K_p(b)}{R\ell}.
\]
Choose one sufficiently large absolute $R$ so that
$\E(z;z\le R\ell)\ge c\nu^2\mathcal K_p(b)/\ell$.
Integrating the pointwise deficit and using $\log(1-u)\le-u$ proves
\eqref{sp:eq:sharpscalar}. If $\mathcal K_p(b)=0$ it holds trivially.
\end{proof}

With $b\colonequals a$, combine the scalar estimate, \eqref{eq:fenchel}, and
\cref{lem:variancerefinement} by H\"older's inequality. This gives
\eqref{sp:eq:Fdrift}, now with $\rho=\nu^2/\ell^2$ and constant $\lambda$.
The parameter $\epsilon$ is chosen small enough that all moments at
$3\lambda$ satisfy the stated restrictions. This is the step that preserves
$\eta\asymp\ell^{-2}$ rather than the weaker $\ell^{-3}$ choice.

\paragraph{The common coupling.}
For either parameter choice, the scalar drift, the common remainder bound,
and \cref{lem:sparsegeometry} give the following conditional bounds for
$P\colonequals e^{\lambda F}$, with $M\colonequals\norm d_\infty$:
\[
\begin{aligned}
 \E P(d+xX)&\le(1+\beta)P(d)
   &&\text{at every state},\\
 \E P(d+xX)&\le(1-\delta_F)P(d)
   &&\text{if }M-F(d)\le R,\quad F(d)\ge C_H(\sqrt k/\eta+R).
\end{aligned}
\]
Here $R\ge0$ is arbitrary, and the rates satisfy
\begin{equation}\label{sp:eq:rates}
 \delta_F\colonequals c_F\lambda\rho p/\sqrt k,\qquad
 \delta_C\colonequals c_C\lambda\rho p,\qquad
 \beta\le Cc_\eta\lambda\rho p/\sqrt k.
\end{equation}
Choose $c_F$ below $c_C/16$, then $c_\eta$ small enough for
$\beta\le\delta_F/4$. The fixed $B_0$ has order $\ell$ in the moderate case and
$\ell^2$ for the constant exponent. The growth rates are
$\gamma\colonequals Cp\lambda^2\nu^2$ and $\gamma\colonequals Cp\lambda^2$, respectively.
Apply \cref{prop:foster} with
$H=C_H(\sqrt k/\eta+R)$. Its $\Gamma$ is bounded absolutely in the
moderate case, and is at most $C\sigma^2\ell^2$ in the constant-exponent
case. Thus $\log\Gamma\le C\ell^2$; this contribution is absorbed by
$\ell^2\sqrt k$ since $k\ge1$. The normalized terminal threshold is therefore
$C\ell[\sqrt k+\log(2n)+r]$ and
$C[\ell^2\sqrt k+\log(2n)+r]$, respectively.
Rescaling by $\sigma$ proves \cref{thm:SparseConfidence}.
The main sparse theorem, \cref{thm:SparseOnlineDiscrepancy}, follows as in
the proof of \cref{cor:SparseAllLevels}.

\paragraph{Determinism, approximation, and computation.}

For each parameter choice, use the deterministic minimizer of the full
$V$ in \eqref{eq:SharedV}, with fixed tie-breaking. The proof of
\cref{prop:foster} supplies the derandomization. Approximate comparisons
with relative error at most $\delta_F/8$ preserve the terminal bounds.
Implementation and certified complete-potential comparisons are detailed
in \cref{sec:implementation}.

\subsection{Sub-Gaussian confidence bounds}
\label{sec:SparseSubGaussian}

We state the full sub-Gaussian envelope before proving its first bound.
The three profiles trade off the leading variance term and the dimension
correction.

\begin{theorem}[Sub-Gaussian envelope]\label{thm:SparseEnvelope}
Under the conditional product assumptions of \cref{thm:SparseRefinements},
set $\sigma_*\colonequals \max(\sigma,2)$,
$A\colonequals 8\sigma_*\sqrt{\log(e\sigma_*)}$,
$\nu\colonequals \max(k/A^2,1)$, and $R\colonequals \min(A^2,k)$, and define
\begin{equation}\label{eq:SparseEnvelopeProfiles}
\begin{aligned}
 U_0&\colonequals \sqrt k+A\log(2n),\\
 U_G&\colonequals A\left[\sqrt k+\sqrt{\log(en/k)}\right],\\
 U_{\rm red}&\colonequals A\sqrt{\log(eR)}
       \left[\sqrt\nu+\sqrt{\log(en/\nu)}\right].
\end{aligned}
\end{equation}
There is a deterministic online algorithm depending only on $n,k,\sigma$
such that, uniformly over these input laws, for every prescribed
deterministic integer $T\ge0$ and every $r\ge0$,
\begin{equation}\label{eq:SparseSubGaussianEnvelope}
 \P\left(\norm{\sum_{t=1}^T x_tv_t}_\infty>
 C\left[\min\{U_0,U_G,U_{\rm red}\}+Ar\right]\right)\le e^{-r}.
\end{equation}
Here $C>0$ is an absolute constant. The algorithm selects a branch
minimizing $U_0,U_G,U_{\rm red}$ before the first arrival and keeps it
throughout. It uses fixed tie-breaking and requires no knowledge of the
input laws, the horizon, or the confidence parameter.
\end{theorem}

The $U_0$ bound retains an absolute constant on the leading $\sqrt k$
term. The $U_G$ bound yields \cref{thm:SparseRefinements} and its
Gaussian specialization in \cref{cor:SparseGaussian}.
The reduced profile $U_{\rm red}$ can improve on both other bounds when $\sigma$
grows with $n$; \cref{sec:SparseLimits} gives an explicit regime.

\paragraph{Separating variance from the tail parameter.}
We prove the $U_0+Ar$ branch of \cref{thm:SparseEnvelope} using a hard
weight cap. The other two branches use the curvature-breakpoint
construction in \cref{sec:SparseGaussianUpper}.
We restate its conditional assumptions to fix the normalization.

Let $n\ge1$, $1\le k\le n$, $p\colonequals k/n$, and $\sigma>0$.
Set $\sigma_*\colonequals \max(\sigma,2)$ and $\ell\colonequals \sqrt{\log(e\sigma_*)}$.
Conditionally on the complete past
$\Fc_{t-1}$, suppose that the marks $w_t(i)$ are independent and
\begin{equation}\label{vg:eq:model}
 \E{w_t(i)\mid\Fc_{t-1}}=0,\qquad
 \E{w_t(i)^2\mid\Fc_{t-1}}=1,\qquad
 \E{e^{w_t(i)^2/\sigma_*^2}\mid\Fc_{t-1}}\le2.
\end{equation}
The fresh masks $\xi_t(i)$ are independent $\operatorname{Ber}(p)$
variables, independent of the fresh marks conditionally on the same past.
The conditional mark laws may vary with time, coordinate, and past.
Write $v_t(i)\colonequals \xi_t(i)w_t(i)$ and $S_t\colonequals \sum_{s\le t}x_sv_s$.
Symmetry is not assumed. The original moment assumption implies the displayed bound with $\sigma_*$.

Set, before the first arrival,
\begin{equation}\label{vg:eq:normalization}
 \begin{gathered}
 A\colonequals 8\sigma_*\ell,\quad X_i\colonequals \xi_iw_i/A,\quad q\colonequals p/A^2,\\
 c\colonequals \max\{q,1/n\},\quad k_c\colonequals nc=\max\{k/A^2,1\}.
 \end{gathered}
\end{equation}
Thus $0<q\le c\le1$ and $1\le k_c\le n$. The state $d\colonequals S/A$ is
always normalized in the proof. All one-step expectations below condition
on a fixed past. In particular all coefficients selected from $d$ are
fixed in those expectations.
The normalization by $A$ gives fourth and tilted moment bounds
proportional to $q$. Choosing $c=\max(q,1/n)$ yields
$A\sqrt{k_c}=\max(\sqrt k,A)$; the additional $A$-term is absorbed by
$A\log(2n)$.

\paragraph{Variance-sensitive projection and remainder.}
The next lemma collects the normalized moment bounds used in the projection
and remainder estimates.

\begin{lemma}\label{vg:lem:moments}
Uniformly over the conditional laws,
\[
 \E X_i=0,\quad \E X_i^2=q,\quad
 \E X_i^4\le\frac5{64}q,\quad
 \E{(X_i^2+X_i^4)e^{4X_i^2}}\le Cq.
\]
In particular, $\E(X_i^2e^{|X_i|})\le Cq$.
\end{lemma}
\begin{proof}
Put $L\colonequals \log(e\sigma_*)\ge1$. Splitting at
$|w|=2\sigma_*\sqrt L$, and using $\P(|w|>t)\le2e^{-t^2/\sigma_*^2}$,
gives
\[
 \E w^4\le4\sigma_*^2L+
 2\sigma_*^4(16L^2+8L+2)e^{-4L}\le5\sigma_*^2L.
\]
Indeed the second term divided by $\sigma_*^2L$ is
$2e^{-2}(16L+8+2/L)e^{-2L}\le52e^{-4}<1$.
Since $A^2=64\sigma_*^2L$, the fourth-moment assertion follows.

For $H\colonequals |w|/A$, on $H\le1$ one has
$(H^2+H^4)e^{4H^2}\le2e^4H^2$. On $H>1$, use
$\E e^{64LH^2}\le2$ and the decreasing function
$(h^2+h^4)e^{-(64L-4)h^2}$ to bound the expectation by
$4e^{4-64L}\le CA^{-2}$. Restoring the mask gives the tilted estimate.
The remaining moment claims follow from elementary domination.
None of these arguments uses symmetry.
\end{proof}

Let $F\colonequals \F_c$ be the potential in \eqref{eq:potential} and put
$a\colonequals \nabla F(d)$, $M_d\colonequals \norm d_\infty$, and
$S_z\colonequals \sum_i(u_i^{3/2}+v_i^{3/2})$.
The geometric statements in \cref{lem:kkt,lem:sparsegeometry} apply
with $k_c=nc$ and require no relation between $c$ and $p$.
For a fixed coordinate $i$, put $H\colonequals |w_i|/A$.
The common scalar remainder obeys $h_z(H)\le cH$ and
\eqref{eq:scalarbound}, with cutoff $B_z\colonequals \sqrt c/(2\eta\sqrt z)$.

\begin{lemma}[Variance-sized remainder]\label{vg:lem:remainder}
For $0\le s\le1$ and sufficiently small absolute $\eta>0$,
\[
 \log\E e^{sU(X)}\le
 Cs\eta(q/c)\mathcal K_c(a)+Cs\eta q/\sqrt{k_c}.
\]
\end{lemma}
\begin{proof}
The bounds \eqref{eq:scalarbound} imply globally $h_z(H)\le2\eta cH^2$:
outside the local range, $H>B_z\ge1/(2\eta)$.
Thus the pair remainder is at most $4\eta cH^2$. Choosing parameters
so its exponential tilt is at most $e^{4H^2}$, the local contribution
for one weight is $C\eta(q/\sqrt c)z^{3/2}$. On the complement,
\[
 pc\E{He^{4H^2};H>B_z}
 \le pcB_z^{-3}\E{H^4e^{4H^2}}
 \le C\eta^3(q/\sqrt c)z^{3/2}.
\]
Use $e^{su}-1\le su e^{su}$, independence across coordinates, and
\eqref{eq:SparseGeometryShared}. The proof uses magnitudes, not symmetry.
\end{proof}

The leading remainder term has the same $(q/c)\mathcal K_c(a)$ scale as
the projection drift, so an absolute choice of $\eta$ suffices to absorb it.

\paragraph{A centered comparison rule.}
Let $\sgn_+(z)$ be $1$ for $z\ge0$ and $-1$ for $z<0$.
At a fixed state write $Z\colonequals \inp aX$. For analysis only, use an independent
fair coin to choose between
\begin{equation}\label{vg:eq:competitor}
 x_G\colonequals -\sgn_+(Z),\qquad
 x_U\colonequals -\sgn_+(Z+Z_0),\quad Z_0\sim\operatorname{Unif}[-4c,4c].
\end{equation}
The auxiliary variables are independent of the past and fresh input.
Call the resulting sign $x$. Neither branch depends on a mark law.

\begin{lemma}[Potential drift]\label{vg:lem:Fdrift}
For small absolute $t>0$ and then sufficiently small absolute $\eta>0$,
\begin{equation}\label{vg:eq:Fdrift}
 \log\E e^{t[F(d+xX)-F(d)]}
 \le Ct\eta q/\sqrt{k_c}-b_Ft(q/c)\mathcal K_c(a),
\end{equation}
where $b_F>0$ is absolute and independent of sufficiently small $\eta$.
\end{lemma}
\begin{proof}
Apply \cref{lem:ThresholdMoment} with $\ell_0=1$ and $h=4c$;
\cref{vg:lem:moments} supplies the required moments. For $s\le1/2$
this gives $\log\E e^{sxZ}\le-bs(q/c)\mathcal K_c(a)$.
H\"older in \eqref{eq:fenchel}, at parameter $2t\le1/2$, and
\cref{vg:lem:remainder} prove the claim after choosing $\eta$ small
enough to absorb its $\mathcal K_c$ term.
\end{proof}

\paragraph{Coordinate drift.}
Apply \cref{lem:ThresholdCoordinate} with $J=4$:
\cref{vg:lem:moments} gives $\E X_i^4\le5q/64\le q/8$.
The inward means are nonnegative, and when $|a_i|\ge c/2$ the
half-mixture has $\mu_i\ge\alpha q$, where $\alpha\colonequals 3/128$.
The common Taylor estimate and the tilted second moment give
\[
 \frac{\E\cosh(\lambda(d_i+xX_i))}{\cosh(\lambda d_i)}-1
 \le-\lambda\tanh(\lambda|d_i|)\mu_i+C\lambda^2q.
\]
Choose $\lambda\le\min(1/4,\alpha/(4C))$ and
$D_0\colonequals \lambda^{-1}\operatorname{arctanh}(1/2)$.
The growth rate is $\gamma\colonequals C_Cq$, $C_C\colonequals C\lambda^2$; when
$|a_i|\ge c/2$ and $|d_i|\ge D_0$, contraction holds with
$\delta_C\colonequals c_Cq$, $c_C\colonequals \alpha\lambda/4$.

\paragraph{The common coupling and rescaling.}
Fix the small absolute $\lambda$ in the coordinate estimate and then take
$\eta$ small enough for the remainder absorption. The scalar drift and
\cref{lem:sparsegeometry} give
\[
 \delta_F\colonequals f_0q/\sqrt{k_c},\qquad
 \delta_C\colonequals c_Cq,\qquad \gamma\colonequals C_Cq,\qquad \beta\le\delta_F/4,
\]
where $f_0$ is chosen below $c_C/16$.
Choose an absolute $B_0\ge D_0$. For the application of \cref{prop:foster},
write $R_{\mathrm{cpl}}\colonequals B_0+\lambda^{-1}\log(32n\Gamma)$ and take
$H=C_H(\sqrt{k_c}/\eta+R_{\mathrm{cpl}})$.
Here $\Gamma=O(1)$ and $\varepsilon\asymp1/(n\sqrt{k_c})$.
The factor $q$ cancels from both $\Gamma$ and $\varepsilon$, so the
coupling introduces no additional $\log(1/q)$ term.
Thus $H=O(\sqrt{k_c}+\log(2n))$, and the deterministic minimizer of the
common full $V$ satisfies
\[
 \P\left(\norm{S_T}_\infty>
 CA[\sqrt{k_c}+\log(2n)+r]\right)\le e^{-r}.
\]
Since $A\sqrt{k_c}=\max(\sqrt k,A)$ and $\log(2n)\ge\log2$,
this proves the $U_0+Ar$ bound. The remaining branches and their
fixed selection are proved in \cref{sec:SparseGaussianUpper}.
The same conditional recursion holds almost surely at every step, also
for mark laws with atoms and in the cases $p=1$ and $n=1$.
The fixed relative comparison tolerance $\delta_F/8$ and the
implementation costs are given in \cref{sec:implementation}.

\subsection{Sub-Gaussian refinements from two breakpoints}
\label{sec:SparseGaussianUpper}

We prove the $U_G$ and $U_{\rm red}$ branches of
\cref{thm:SparseEnvelope} by one calculation.
The $U_G$ bound also proves \cref{thm:SparseRefinements};
\cref{cor:SparseGaussian} follows by specialization.
All expectations below condition on the complete past, fixing the state,
weights, and the next coordinate laws. The fresh mark coordinates are
independent and, for a fixed numerical $M\ge1$ and deterministic $0<\tau\le1$,
satisfy
\begin{equation}\label{sg:eq:moment}
 \E Y_i=0,\qquad \E Y_i^2=\tau,\qquad
 \E{(Y_i^2+Y_i^4)e^{4Y_i^2}}\le M\tau.
\end{equation}
The fresh $\Ber(p)$ masks are independent, also independent of the fresh
marks under the same conditioning. Put $X_i\colonequals \xi_iY_i$, $q\colonequals p\tau$, and
consider exactly the two choices
\[
 c=p\quad\hbox{or}\quad c=\max(q,1/n),\qquad
 \nu\colonequals nc,\qquad R\colonequals p/c.
\]
In both cases $1/n\le c\le p$ and $1\le R\le n$.
The two rules use this $c$ as a curvature breakpoint rather than a hard
cap. Constants in the normalized calculation depend only on $M$;
the application fixes $M=128$.

Write
\[
 b\colonequals \sqrt{\nu}\ge1,\qquad S\colonequals \sqrt{\log(e/c)},\qquad
 \mathcal D\colonequals b+1+\kappa S.
\]
Choose a small fixed $\lambda>0$, a large fixed $L_0\ge2$, and then
a sufficiently small fixed $\kappa_0>0$. Set
\begin{equation}\label{sg:eq:params}
 L\colonequals L_0R,\qquad
 \kappa\colonequals \frac{\kappa_0}{\sqrt{\log(eR)}},\qquad
 \eta\colonequals \frac{\eta_0\kappa b}{b+\kappa+\kappa^2S},
\end{equation}
where $\eta_0>0$ is fixed sufficiently small last.
In particular $\eta\le\eta_0\kappa\le\kappa$, and
\begin{equation}\label{sg:eq:eta-budget}
 \frac{\eta\mathcal D}{b}
 =\eta_0\frac{\kappa b+\kappa+\kappa^2S}
                 {b+\kappa+\kappa^2S}\le\eta_0
 \quad(\kappa\le1).
\end{equation}
In particular, $\eta q/b\le\eta_0q/\mathcal D$, which matches the positive
remainder to the eventual high-state drift.

We use the notation
$\alpha\colonequals \kappa\sqrt{\log(eL)}$.
Since
\[
 \alpha\le\kappa_0\sqrt{1+\log L_0},
\]
it is uniformly small after choosing $\kappa_0$.
We may impose
$\kappa_0\le[8\sqrt{1+\log L_0}]^{-1}$ and decrease it further
for the low-error absorption. The fixed numerical constants
$\lambda,L_0,\kappa_0,\eta_0$ depend only on $M$.
All error constants below can be bounded uniformly for
$0<\lambda,\kappa,\eta\le1$ and $\alpha\le1/8$ before fixing $L_0$.
The successive parameter choices therefore introduce no dependence of
these constants on $p,\tau,c,R$.
For either breakpoint, we will prove a normalized terminal bound of
$C_M[(b+S)/\kappa+r]$ with failure probability at most $e^{-r}$
at every prescribed deterministic horizon.

\paragraph{Projection moments.}
For $0<t\le1$, decreasing absolute coefficients $a_j^*$ padded to
length $n$, and $m_t\colonequals \min(n,\floor{1/t})$, write
\[
 K_t(a)\colonequals t\sum_{j\le m_t}a_j^*
       +\sqrt t\left(\sum_{j>m_t}(a_j^*)^2\right)^{1/2},
 \qquad G_t(a)\colonequals \E\norm{a\odot\zeta_t}_2,
\]
where the coordinates of $\zeta_t$ are independent $\Ber(t)$.
For $t\ge1/n$, $K_t=\mathcal K_t$ and the head--tail argument in
\cref{lem:BernoulliNorm} gives the comparison below: the head has
constant $1/e$; if the tail is smaller than its largest coefficient,
it is at most twice the head, and otherwise interpolation applies.
For $t<1/n$, sole-hit events give
$G_t\ge K_t/e$, and the triangle inequality gives $G_t\le K_t$.
Thus, throughout $0<t\le1$,
\begin{equation}\label{sg:eq:Bern}
 K_t(a)/(3e)\le G_t(a)\le K_t(a).
\end{equation}

\begin{lemma}\label{sg:lem:projection}
For $B\colonequals \sum_i a_i\xi_iY_i$ under \eqref{sg:eq:moment},
\[
 m_B\colonequals \E|B|\ge\gamma_M(q/c)K_c(a).
\]
If $\norm a_1\le1$, then
\[
 \E B^2=q\norm a_2^2\le\sqrt M\,m_B,\qquad
 \E e^{-\lambda|B|}\le1-\tfrac12\lambda m_B
 \quad(0\le\lambda\le1/\sqrt M).
\]
\end{lemma}
\begin{proof}
Let $Y_i'$ be independent copies of the respective conditional
coordinate laws, independent of the fresh marks and masks. Use the same masks in
$B'\colonequals \sum_i a_i\xi_iY_i'$. The triangle inequality and conditional
Rademacher interpolation give, with $D_i\colonequals Y_i-Y_i'$,
\[
 m_B\ge\frac12\E|\sum_i a_i\xi_iD_i|
 \ge\frac1{2\sqrt3}\E\sqrt{\sum_i a_i^2\xi_iD_i^2}.
\]
Put $J\colonequals M+3$. Centering gives $\E D_i^2=2\tau$ and
$\E D_i^4\le2J\tau$. Since $(x-J)_+\le x^2/(4J)$,
$V_i\colonequals \min(D_i^2,J)$ satisfies $\E V_i\ge3\tau/2\ge \tau$.
For fixed $s_0\ge0$, the function $\sqrt{s_0+a_i^2z}$ lies above
its endpoint chord on $z\in[0,1]$.
Replacing each independent $\xi_iV_i/J$ by a Bernoulli
variable with the same mean therefore only decreases the expectation
of the square root. Decrease those means to $q/J$ to get
\[
 \E\sqrt{\sum_i a_i^2\xi_iD_i^2}\ge\sqrt J\,G_{q/J}(a).
\]
Represent a $\operatorname{Ber}(q/J)$ mask by independent
$\operatorname{Ber}(q)$ and $\operatorname{Ber}(1/J)$ masks.
Conditional Jensen for the convex Euclidean norm gives
$G_{q/J}\ge G_q/J$. Together with \eqref{sg:eq:Bern} this proves the
first inequality with, for example,
$\gamma_M=(6e\sqrt{3J})^{-1}$.

For $c=\max(q,1/n)$, $K_q=(q/c)K_c$ exactly: either $c=q$,
or both heads contain all coordinates. For $c=p$, independent thinning
and Jensen for the convex Euclidean norm give $G_q\ge\tau G_p$, so the
same argument gives $m_B\ge\gamma_M\tau K_p$.
For each coordinate, fourth-moment interpolation gives
$\mu_i\colonequals \E|Y_i|\ge \tau/\sqrt M$, while Cauchy gives $\mu_i\le\sqrt \tau$.
If $j$ maximizes $|a_j|$, original centered independence and Jensen imply
$m_B\ge p\mu_j|a_j|\ge q\norm a_\infty/\sqrt M$.
Thus $\E B^2\le q\norm a_\infty\norm a_1\le\sqrt M\,m_B$.
Apply $e^{-x}\le1-x+x^2/2$ for the final assertion.
\end{proof}

The stronger projection pays for the reduced breakpoint. It uses only
centering and the fourth-moment bound; the exponential moment is used below.

\paragraph{A common potential for the two breakpoints.}
Define a convex $f$ by continuity at $c$ and
\[
 \begin{gathered}
 f(0)\colonequals 0,\qquad f(z)\colonequals -4\sqrt{cz}/\eta\ (0<z\le c),\\
 f'(z)\colonequals -2/\eta+(2/\kappa)(\sqrt{1+\log(z/c)}-1)\ (z\ge c).
 \end{gathered}
\]
For a normalized state $d$ define
\begin{equation}\label{sg:eq:potential}
 F(d)\colonequals \max_{\substack{u_i,v_i\ge0\\\sum_i(u_i+v_i)=1}}
 \left\{\sum_i d_i(u_i-v_i)-\sum_i[f(u_i)+f(v_i)]\right\}.
\end{equation}
The ideal rule chooses the smaller of the two successor potentials.
We first analyze this continuous potential, then transfer the estimates
to the potential with logarithmic curvature bins used in the implementation.
Let $\psi\colonequals (f')^{-1}$, $a_i\colonequals u_i-v_i$, $z_i\colonequals \max(u_i,v_i)$,
and $\bar z_i\colonequals \min(u_i,v_i)$ at the unique optimizer.
Its multiplier satisfies
$u_i=\psi(d_i-\theta)$, $v_i=\psi(-d_i-\theta)$.
Since $\eta\le\kappa$, $\psi$ is convex, with an upward derivative jump
at $f'(c)=-2/\eta$. Since $c\ge1/n$, Jensen and the Euler identity give
\begin{align}
 \theta&\ge2\sqrt{2\nu}/\eta,&
 \bar z_i&\le1/(2n),&
 \operatorname{sgn}(a_i)&=\operatorname{sgn}(d_i), \label{sg:eq:geometry}\\
 F(0)&=4\sqrt{2\nu}/\eta,&
 0\le F(d)-\theta&\le Cb/\eta,&
 \norm d_\infty&\le F(d)+2S/\kappa.               \label{sg:eq:domination}
\end{align}
For the Euler bound, below $c$ the summand is $2\sqrt{cz}/\eta$;
above $c$ it is at most $2c/\eta+z/\kappa$.
Cauchy and at most $1/c$ upper weights give $Cb/\eta$ because
$b\ge1$ and $\eta\le\kappa$. The last inequality tests a simplex vertex
and uses $f(1)\le2S/\kappa$.
The domination cost $S/\kappa$ is retained in the final threshold.

For $0<z<1$ use the capped scalar conjugate
\[
 L_z(t)\colonequals \sup_{0\le r\le1}\{rt-D_f(r,z)\}
       =\int_0^t\min\{1,\psi(f'(z)+s)\}\,ds.
\]
Write $P_z(H)\colonequals L_z(H)$, $D_z(H)\colonequals -L_z(-H)$ and
$h_z(H)\colonequals \max\{L_z(H)-zH,L_z(-H)+zH\}$.
Then $0\le h_z(H)\le H$, $0\le D_z(H)\le zH$, $0\le P_z(H)\le H$.
The KKT identity and dropping only the simplex mass constraint give
\begin{equation}\label{sg:eq:scalar-comparison}
 F(d+xX)-F(d)\le\sum_i[L_{u_i}(xX_i)+L_{v_i}(-xX_i)].
\end{equation}
Here the normalized fresh vector is $X_i=\xi_iY_i$.

\paragraph{Low-coordinate remainders and their product tilt.}
The heavy set is
$I\colonequals \{i:z_i\ge Lc\}=\{i:z_i\ge L_0p\}$; its complement is low.
Let
\[
 B\colonequals \sum_{i\notin I}a_iX_i,\quad K\colonequals K_c(a_{I^c}),\quad
 U_i\colonequals \xi_i[h_{u_i}(H_i)+h_{v_i}(H_i)],\quad
 U\colonequals \sum_{i\notin I}U_i,\qquad H_i\colonequals |Y_i|.
\]
For this restricted projection, put $m_B\colonequals \E|B|$.
We give the variance-sized scalar argument explicitly to keep the
constants uniform as $L=L_0R$ grows.

For $z\le\min(Lc,1)$ put
$s_z\colonequals \sqrt{1+\log_+(z/c)}$, $a_z\colonequals \kappa s_z\le\alpha$.
Define the uncapped inverse flow by $r_z(t)\colonequals\psi(f'(z)+t)$.
It satisfies
\[
 r_z(t)\le z e^{a_zt+\kappa^2t^2/4},\qquad
 r_z(-t)\ge z e^{-a_zt}\quad(t\ge0).
\]
Below $c$ its logarithmic speed is at most $\eta\le\kappa$;
above $c$ the square-root log factor grows at rate $\kappa/2$.
Integration, with the cap only decreasing the outward displacement,
gives
\[
 h_z(H)\le Cz(a_z+\kappa^2H)H^2
                         e^{a_zH+\kappa^2H^2/4}.
\]
Every fixed polynomial moment of degree at least two with a fixed
quadratic exponent strictly less than $4$ and fixed bounded linear
exponent is $O_M(\tau)$ by \eqref{sg:eq:moment}.
The pair tilt is at most $e^{2\lambda H}$. Thus, for uniformly small
$\alpha$, either low weight has
\begin{equation}\label{sg:eq:scalar-uniform}
 \E{(1+H)h_z(H)e^{\lambda(h_{u_i}(H)+h_{v_i}(H))}}
 \le C_M\tau\alpha z.
\end{equation}
For $z\le3c/4$ there is the stronger bound
\begin{equation}\label{sg:eq:scalar-small}
 \E{(1+H)h_z(H)e^{\lambda(h_{u_i}(H)+h_{v_i}(H))}}
 \le C_M\tau\eta z^{3/2}/\sqrt c.
\end{equation}
Indeed the exact lower inverse-square formula works up to
$B_z\colonequals (8\eta)^{-1}\sqrt{c/z}$ and gives
$h_z(H)\le C\eta z^{3/2}H^2/\sqrt c$ there.
Globally, for $z\le c$,
\[
 P_z(H)\le zH(1+\eta H/2)^2e^{\kappa H+\kappa^2H^2/4},
 \qquad h_z(H)\le P_z(H)+zH.
\]
On $H>B_z$ use $H\le H^2/B_z$ in the retained-$z$ envelope.
All remaining tilted moments have degree at least two, so the tail
is at most $C_M\tau z/B_z$. This proves \eqref{sg:eq:scalar-small}, including
the extra magnitude needed for the tilt. This proves the estimate under \eqref{sg:eq:moment}, with constants uniform in $L$.

For a low coordinate define
\[
 D_i^L\colonequals \frac{\eta}{\sqrt c}
       \sum_{\substack{z\in\{u_i,v_i\}\\z\le3c/4}}z^{3/2}
       +\alpha\sum_{\substack{z\in\{u_i,v_i\}\\z>3c/4}}z.
\]
Then $D_i^L\le\alpha(u_i+v_i)$ and masking gives
\begin{equation}\label{sg:eq:coordinate-cost}
 \E{U_ie^{\lambda U_i}}
   +\E{|X_i|U_ie^{\lambda U_i}}\le C_MqD_i^L.
\end{equation}
The restricted deterministic geometry has no $L$ loss:
\begin{equation}\label{sg:eq:restricted}
 c\sum_{\rm middle}z\le CK,\qquad
 \sqrt c\sum_{\rm small}z^{3/2}\le CK+Cc/b.
\end{equation}
Every middle weight is favored and at most $3|a_i|$ because its opposite
is at most $c/2$; there are $O(1/c)$ middle weights. Head/tail splitting
proves the first inequality. For $\nu\ge2$, inner coordinates
$|d_i|\le\theta/2$ have both weights at most
$16c/(\eta^2\theta^2)\le2/n$ and contribute $Cc/b$ to the second.
On outer coordinates the favored/opposite ratio is at least nine,
using the inverse-square ratio if both are below $c$, or the favored
weight at least $c$ otherwise. Each small weight there is at most
$C|a_i|$ and at most $c$. The head is bounded by $c\sum|a_i|$;
the tail by $\sqrt c\,\norm{a_{\rm tail}}_2$ using total mass at most one.
For $1\le\nu<2$ the entire small-weight sum is at most $c$.
This proves \eqref{sg:eq:restricted} for the restricted low set.

\begin{lemma}[Low-coordinate remainder and product tilt]
\label{sg:lem:low-tilt}
Under \eqref{sg:eq:moment}, with the parameter choices in
\eqref{sg:eq:params}, the low-coordinate quantities defined above satisfy
\begin{equation}\label{sg:eq:low}
 \E e^{\lambda U}\le C_M,\qquad
 \E{Ue^{\lambda U}}
 \le C_M\left[\alpha(q/c)K+\eta q/b\right].
\end{equation}
Moreover,
\begin{equation}\label{sg:eq:actual-tilt}
 \E{|B|e^{\lambda U}}\le C_Mm_B.
\end{equation}
Here $C_M$ depends only on $M$.
\end{lemma}
\begin{proof}
Independence, \eqref{sg:eq:coordinate-cost}, and \eqref{sg:eq:restricted}
give \eqref{sg:eq:low}.
For example the sum of single-coordinate weighted moments is bounded
by the right side, while also being at most $C_M\alpha q$.
Productization bounds the total normalizer and the weighted sum.
Lemma~\ref{sg:lem:projection} pays for the $(q/c)K$ term.

For each low coordinate, put $Z_i\colonequals \E e^{\lambda U_i}$ and
define the tilted marginal by
\[
 \widetilde\E g(X_i)\colonequals Z_i^{-1}\E{g(X_i)e^{\lambda U_i}}.
\]
The joint tilted law is the product of these tilted marginals.
For a Lipschitz $g$, subtract $g(0)$ in the exact change of measure:
\[
 |\widetilde\E g(X_i)-\E g(X_i)|
 \le\frac{\operatorname{Lip}(g)}{Z_i}
 \{\E{|X_i|(e^{\lambda U_i}-1)}+(Z_i-1)\E|X_i|\}
 \le C_M\lambda qD_i^L\operatorname{Lip}(g).
\]
Replacing low coordinates successively,
\[
 |\widetilde\E|B|-m_B|
 \le C_M\lambda q\sum_{i\notin I}|a_i|D_i^L
 \le C_M\lambda\alpha q\norm{a_{I^c}}_\infty
 \le C_M\lambda\alpha m_B.
\]
The last inequality is the original centered Jensen bound from
Lemma~\ref{sg:lem:projection}; it does not assume centered tilted laws.
Multiplying by the bounded product normalizer proves
\eqref{sg:eq:actual-tilt}.
The calculation retains the physical mask rate and uses centering only
under the original conditional laws.
\end{proof}

\paragraph{Heavy coordinates and their actual deficits.}
The heavy coordinates supply inward deficits that will absorb the
collision and opposite-weight errors. Comparing these errors to the
same actual deficits keeps the estimates uniform over the mark laws.
For $i\in I$ put
\[
 s_i\colonequals \sqrt{1+\log(z_i/c)},\qquad
 h_i^0\colonequals \frac{z_i}{1+\kappa s_i},\qquad
 \delta_i\colonequals \lambda^{-1}\E{1-e^{-\lambda D_{z_i}(H_i)}}.
\]
These expected deficits are analytical quantities, not inputs to the rule.
Recall that $\tau/\sqrt M\le\mu_i=\E H_i\le\sqrt \tau$ for each coordinate.
The inward logarithmic speed is at most $\kappa s_i$, whence
$z_iH/(1+\kappa s_iH)\le D_{z_i}(H)\le z_iH$.
The inequality $1-e^{-x}\ge x/(1+x)$ and Cauchy yield
\begin{equation}\label{sg:eq:heavy-deficit}
 \delta_i\ge
 \frac{z_i\mu_i^2}{\mu_i+(\kappa s_i+\lambda z_i)\tau}
 \ge c_M\mu_i h_i^0\ge c_M\tau h_i^0.
\end{equation}
Define the nonnegative amplitude and central term
\[
 S_i(H)\colonequals \tfrac12(P_{z_i}+D_{z_i}-P_{\bar z_i}-D_{\bar z_i})(H),
 \qquad C_i(H)\colonequals S_i(H)-D_{z_i}(H)+P_{\bar z_i}(H).
\]
Monotonicity gives $S_i\ge0$ and convexity gives $C_i\ge0$.
The outgoing inverse and the preceding inward lower bound imply
\[
 S_i(H)\le D_{z_i}(H)(1+\kappa s_iH)
                  e^{\kappa s_iH+\kappa^2H^2/4}.
\]
Put $Q_i^0(H)\colonequals (1-e^{-\lambda D_{z_i}(H)})/\lambda$. Then
\[
 D_{z_i}(H)\le(1+\lambda H)Q_i^0(H).
\]
Both estimates hold also at $H=0$. Cauchy--Schwarz gives
\begin{equation}\label{sg:eq:amplitude}
 \left(\E{\sqrt{S_i(H)}e^{\lambda H}}\right)^2
 \le C_M\delta_i(z_i/c)^\epsilon,\qquad \epsilon\colonequals 8\kappa^2.
\end{equation}
For an explicit integration check, bound
$(1+aH)(1+\lambda H)\le e^{(a+\lambda)H}$ for $a\colonequals \kappa s_i$,
and then
$(2a+3\lambda)H\le H^2/4+8a^2+18\lambda^2$.
The remaining quadratic exponent is at most $H^2/2$ and
$\E e^{H^2/2}\le1+M\tau/2$. Thus the constant is independent of $\tau,c,R$.

Set
\[
\begin{aligned}
 N_d&\colonequals \sum_{i\in I}\delta_i,&
 A_H&\colonequals \sum_{i\in I}\xi_iD_{z_i}(H_i),\\
 V_H&\colonequals \sum_{i\in I}\xi_iH_i,&
 N_0&\colonequals \sum_{i\in I}\xi_i,\\
 J_H&\colonequals 2\sum_{i<j,\ i,j\in I}\xi_i\xi_j\min(S_i,S_j).
\end{aligned}
\]
For the opposite-weight contribution, put
\[
 Q\colonequals \sum_{i\in I}\xi_iP_{\bar z_i}(H_i).
\]

\begin{lemma}[Heavy-coordinate errors and actual deficits]
\label{sg:lem:heavy-deficits}
Under \eqref{sg:eq:moment}, with the parameter choices in
\eqref{sg:eq:params}, the heavy-coordinate quantities defined above satisfy
\[
\begin{aligned}
 \E e^{\lambda V_H}&\le C_M,\\
 \E{N_0e^{\lambda V_H}}&\le C_M/L_0,\\
 \E{(Q+J_H)e^{\lambda V_H}}&\le C_MpN_d/L_0,\\
 \E e^{-\lambda A_H}&\le1-\tfrac12\lambda pN_d.
\end{aligned}
\]
Here $C_M$ depends only on $M$.
\end{lemma}
\begin{proof}
The threshold $Lc=L_0p$ gives $p|I|\le1/L_0$; independence then
yields the first two bounds.
Both physical mask factors and \eqref{sg:eq:amplitude} give
\begin{align}
 \E{J_He^{\lambda V_H}}
 &\le C_Mp^2c^{-\epsilon}
      \left(\sum_i\sqrt{\delta_i}z_i^{\epsilon/2}\right)^2\notag\\
 &\le C_MpN_d\,R L^{-(1-\epsilon)}.              \label{sg:eq:collision}
\end{align}
The last step uses $\sum_i z_i^\epsilon\le(Lc)^{\epsilon-1}$.
The factor $R=p/c$ is essential. Under \eqref{sg:eq:params},
\[
 \epsilon\log L
 \le8\kappa_0^2(1+\log L_0)\le1/8,\qquad
 R L^{-(1-\epsilon)}=L_0^{-1}L^\epsilon
 \le e^{1/8}/L_0.
\]
Thus the collision error is $C_MpN_d/L_0$.

For the opposite-weight term $Q$, the lower-branch outgoing envelope gives
\[
 \E{P_{\bar z_i}(H_i)e^{\lambda H_i}}
 \le\bar z_i(\mu_i+C_M\tau)\le C_M\mu_i\bar z_i.
\]
Indeed, after subtracting the term $\bar z_iH_i$, the envelope has a
factor $H_i^2$ times a fixed polynomial and an exponential with quadratic
coefficient less than $4$. Its expected contribution is therefore
$O_M(\bar z_i\tau)$ by \eqref{sg:eq:moment}, and is absorbed using
$\mu_i\ge\tau/\sqrt M$.
For $r_i\colonequals z_i/c\ge L$, monotonicity of
$(1+\kappa\sqrt{1+\log r})/r$ gives
\[
 \frac{\bar z_i}{h_i^0}
 \le\frac{1+\kappa\sqrt{1+\log r_i}}{2\nu r_i}
 \le\frac{1+\alpha}{2\nu L}\le C/L_0.
\]
Hence $\E(Qe^{\lambda V_H})\le C_MpN_d/L_0$.
Finally $\delta_i\le\mu_i z_i\le\sqrt \tau\,z_i$, so $\lambda pN_d\le\lambda$, and
\begin{equation}\label{sg:eq:heavy-laplace}
 \E e^{-\lambda A_H}
 =\prod_{i\in I}(1-\lambda p \delta_i)
 \le1-\tfrac12\lambda pN_d.
\end{equation}
The collision, opposite, and negative-Laplace bounds all use the same
actual coordinatewise deficits, including when the mark laws differ.
\end{proof}

\paragraph{From the scalar comparison to terminal moments.}
For $i\in I$, set $\varepsilon_i\colonequals \sgn(a_iY_i)$, with a fixed
sign at zero. The corresponding scalar contribution is
$\xi_i[C_i(H_i)+x\varepsilon_iS_i(H_i)]$; no conditional symmetry of
the signs is used. Since the ideal rule minimizes the successor potential,
its increment is at most the minimum of the two scalar upper bounds in
\eqref{sg:eq:scalar-comparison}. Thus, pointwise,
\[
 \Delta F\le U+\sum_{i\in I}\xi_iC_i
                    -\left|B+\sum_{i\in I}\xi_i\varepsilon_iS_i\right|.
\]
Apply
$\sum b_i-|\sum\varepsilon_i b_i|\le2\sum_{i<j}\min(b_i,b_j)$
to the heavy amplitudes and to the additional term $B$.
The cross terms involving $B$ are at most $2N_0|B|$.
Using $C_i=S_i-D_{z_i}+P_{\bar z_i}$ yields
\[
 \Delta F\le-A_H-|B|+Q+J_H+2N_0|B|+U,\qquad
 \Delta F\le U+V_H.
\]
The second inequality chooses the low-linear minimizing sign and uses
the capped scalar pair bound $H_i$ on each heavy coordinate.
The mean value theorem therefore gives
\[
 e^{\lambda\Delta F}
 \le e^{-\lambda(A_H+|B|)}
 +\lambda(Q+J_H+2N_0|B|+U)e^{\lambda(U+V_H)}.
\]
Heavy and low coordinates are independent at the fixed current state.
In particular the conflict term is at most $C_Mm_B/L_0$ by
\eqref{sg:eq:actual-tilt}. Combining
\cref{sg:lem:projection,sg:lem:low-tilt,sg:lem:heavy-deficits},
\begin{align*}
 \E e^{\lambda\Delta F}
 &\le1-c_M\lambda(pN_d+m_B)\\
 &\quad+C_M\lambda\left[
       (pN_d+m_B)/L_0+\alpha(q/c)K+\eta q/b\right].
\end{align*}
First take $L_0$ large, then $\kappa_0$ small so that
$\alpha(q/c)K$ is absorbed using Lemma~\ref{sg:lem:projection}.
We obtain the global inequality
\begin{equation}\label{sg:eq:global-drift}
 \E e^{\lambda\Delta F}
 \le1-c_M\lambda(pN_d+m_B)+C_M\lambda\eta q/b.
\end{equation}

Let $H_*\colonequals C_Hb/\eta$, with large fixed $C_H$.
When $F\ge H_*$, \eqref{sg:eq:domination} makes $\theta$ large enough
that the inner optimizer mass on $|d_i|\le\theta/2$ is at most $1/4$.
If heavy favored mass is at least $1/4$, \eqref{sg:eq:heavy-deficit} gives
\[
 pN_d\ge c_Mq/(1+\kappa S).
\]
Otherwise the total heavy mass is less than $1/2$, so a fixed positive
outer low mass remains. On those coordinates
$u_i+v_i\le C|a_i|$, as in \eqref{sg:eq:restricted}. Thus
$\norm{a_{I^c}}_1\ge c_1$. The head or tail yields
$K_c(a_{I^c})\ge c_2 c/\sqrt{\nu}$, and
$m_B\ge c_Mq/b$.
Both cases imply the high-state bound
\begin{equation}\label{sg:eq:high}
 pN_d+m_B\ge c_Mq/\mathcal D,\qquad
 \mathcal D=b+1+\kappa S.
\end{equation}
The same geometry covers $\nu$ near one because $\theta$ is in the
large-state region; no assumption $\nu\ge2$ is used here.

Choose $\eta_0$ small in \eqref{sg:eq:params}.
Equation~\eqref{sg:eq:eta-budget} makes the global positive residual in
\eqref{sg:eq:global-drift} at most a small constant times $q/\mathcal D$.
With $\delta\colonequals c_M\lambda q/\mathcal D$, after fixing its small coefficient,
\begin{equation}\label{sg:eq:Foster}
 \E{e^{\lambda F(d_{t+1})}\mid\mathcal F_t}
 \le(1-\delta)e^{\lambda F(d_t)}
                    +C_M\delta e^{\lambda H_*}.
\end{equation}
Above $H_*$ this is the high-state contraction. Below it, the global
positive residual is $O_M(\delta)$, so the injection contains the same
factor $\delta$. There is no inverse-$q$ normalization loss.
Choose $C_H$ so that $F(0)\le H_*$. Iteration gives
$\sup_{T\ge0}\E e^{\lambda F(d_T)}\le C_Me^{\lambda H_*}$.
Since $\lambda>0$ is fixed independently of all model and confidence
parameters, \cref{lem:terminal} and \eqref{sg:eq:domination} yield
\begin{equation}\label{sg:eq:all-confidence}
 \P\left(\norm{d_T}_\infty>
 C_M\left[\frac{b+S}{\kappa}+r\right]\right)\le e^{-r}
 \qquad(T\ge0,\ r\ge0).
\end{equation}
Indeed $b/\eta=\eta_0^{-1}(b/\kappa+1+\kappa S)$, so the state
threshold and the domination cost are $O_M((b+S)/\kappa)$.
The confidence cost is $r$, without an extra $\kappa^{-1}$ factor.
If $L_0p>1$, the heavy set is empty and the same low argument applies.
All estimates were conditional on an arbitrary past, so their iteration
uses only the tower property, with the same constants at every history.

\paragraph{The implemented potential and finite accuracy.}
Let $f_b$ agree with $f$ on $[0,c]$ and prescribe, for integers $j\ge0$,
\[
 f_b''(z)\colonequals \frac1{\kappa z\sqrt{1+j}}
 \qquad(ce^j<z<ce^{j+1}),
\]
with $f_b$ and $f_b'$ continuous. Define $F_b$ by \eqref{sg:eq:potential}
with $f_b$ in place of $f$. The implementation compares certified upper
values of $F_b(d+X)$ and $F_b(d-X)$, each with additive error at most
\begin{equation}\label{sg:eq:tolerance}
 \varepsilon\colonequals c_{\rm tol}q/\mathcal D
 =c_{\rm tol}\frac{p\tau}{\sqrt{nc}+1+\kappa\sqrt{\log(e/c)}}.
\end{equation}
It chooses the smaller upper value, with fixed tie-breaking.

\begin{lemma}[Transfer to the implemented potential]\label{lem:BinnedImplementation}
Under \eqref{sg:eq:moment} and the parameter choices in
\eqref{sg:eq:params}, the exact successor-minimizing rule for $F_b$
satisfies \eqref{sg:eq:Foster} with $F_b$ in place of $F$, with constants
depending only on $M$. For sufficiently small fixed $c_{\rm tol}>0$,
the certified rule with tolerance \eqref{sg:eq:tolerance} retains
contraction at least $\delta/2$ and an injection of order
$\delta e^{\lambda H_*}$. Both rules satisfy
\eqref{sg:eq:all-confidence}, with constants depending only on $M$.
\end{lemma}

\Cref{app:BinnedImplementation} proves the lemma and the certified
evaluation bounds. The comparison tolerance is fixed before the first
arrival and depends on neither the horizon nor the confidence parameter.
\Cref{sec:implementation} summarizes the resulting operation counts.

\paragraph{Rescaling and choosing a branch.}
For the assumptions of \cref{thm:SparseEnvelope},
put $Y_i\colonequals w_t(i)/A$ and $\tau\colonequals A^{-2}$, using
$\sigma_*=\max(\sigma,2)$ and $A=8\sigma_*\sqrt{\log(e\sigma_*)}$.
These marks satisfy \eqref{sg:eq:moment} with $M=128$, conditionally
at every history. To check this, put $Z\colonequals w_t(i)^2$ and
$D\colonequals A^2/\sigma_*^2=64\log(e\sigma_*)$.
On $Z\le A^2$,
$(Z+Z^2/A^2)e^{4Z/A^2}\le2e^4Z$.
For $t\colonequals Z/\sigma_*^2>D$, the ratio of this integrand to $e^t$ is
\[
 \sigma_*^2t(1+t/D)e^{-(1-4/D)t}.
\]
It decreases for $t\ge D$ and is at most
$2\sigma_*^2D e^{-D+4}<1$ for $\sigma_*\ge2$.
Thus the conditional expectation is at most $2e^4+2<128$;
dividing by $A^2$ proves the claim.
Whenever the class is nonempty, Jensen gives $\sigma\ge(\log2)^{-1/2}$.
Thus $A\asymp\sigma\sqrt{\log(e+\sigma)}$ with absolute constants;
replacing $\sigma$ by $\sigma_*$ changes the final scale only by an
absolute factor.

For $c=p$, $R=1$ and $\kappa$ is an absolute constant. Hence
$A(b+S)/\kappa\asymp U_G$. For $c=\max(q,1/n)$,
$\nu=\max(k/A^2,1)$ and $R=\min(A^2,k)$, giving
$A(b+S)/\kappa\asymp U_{\rm red}$.
The normalized confidence term in \eqref{sg:eq:all-confidence} becomes
$Ar$ for both rules. Their cost is
$O(n\log(2n\sigma_*)\log\log(4n))$ elementary real operations per
arrival, with a tolerance independent of $T$ and $r$.
The $c=p$ rule, together with the preceding comparison for $A$, proves
\cref{thm:SparseRefinements}.

Together with the $U_0+Ar$ bound from \cref{sec:SparseSubGaussian},
choose the smallest of $U_0,U_G,U_{\rm red}$ before the first arrival,
using only $n,k,\sigma$ and fixed ties, and run that branch throughout.
Taking the maximum of the three absolute constants proves
\cref{thm:SparseEnvelope}.
Constant-factor comparisons of the base profiles suffice. The choice
uses neither the confidence parameter nor the observations, and it
requires no union bound. Its update cost is the cost of the selected
branch, with the separate qualifications in \cref{sec:implementation}.
For standard Gaussian marks take $\sigma_*=2$, since
$\E e^{G^2/4}=\sqrt2\le2$, and use the $c=p$ branch.
This proves \cref{cor:SparseGaussian} with an algorithm
depending only on $n,k$.

\subsection{Dimension corrections and remaining gaps}
\label{sec:SparseLimits}
For a fixed tail bound, the general centered sub-exponential profiles
have an additive correction of order $\log n+r$. Unit-variance Laplace
marks force this order, up to an additive absolute constant, for
$r\ge\log2$ at every prescribed $T\ge n$, by
\eqref{eq:LaplaceNecessaryConfidence}. Thus, already for this fixed mark
law, a pure $O(\sqrt k)$ bound at fixed nontrivial confidence requires
$k\gtrsim\log^2 n$.

For every fixed sub-Gaussian tail bound, the $U_G$ branch already
reduces the dimension correction to order $\sqrt{\log(2n)}$.
For Gaussian marks, \cref{cor:SparseGaussian} matches the necessary
$\sqrt k+\sqrt{\log(2n)}$ dimension order in
\cref{cor:GaussianDimensionLower} at fixed confidence and every prescribed
$T\ge n$. For $r\ge\log2$, the necessary dimension-and-confidence correction has order
$\sqrt{\log(2n)+r}$, whereas our fixed-rule upper profile contains
$\sqrt{\log(2n)}+r$, in addition to the common $\sqrt k$ term.
Whether a horizon-independent rule can attain the necessary correction
remains open.

For general sub-Gaussian marks, \cref{prop:JointSubGaussianLower}
gives separate parameter obstructions; \cref{prop:RareActivationLower}
shows that their sum need not be sufficient.
For example, when $k=n$ and $\sigma^2=n/\log(n+1)$, rare activations
force $\Omega(\sqrt n\log\log n)$ discrepancy at sufficiently long
prescribed horizons, whereas our envelope gives $O(\sqrt n\log n)$
at fixed confidence. \Cref{prop:RareActivationUpper} gives a matching
upper bound for this particular atom law; the gap remains for a uniform
guarantee over the entire admissible class.

The reduced-breakpoint branch also improves the joint upper scale.
For $Q\colonequals \log(2n)$, $\sigma^2=Q$, and $k\asymp Q^2$, one has
$A\asymp\sqrt{Q\log Q}$, $\nu\asymp Q/\log Q$,
and $R\asymp Q\log Q$. Consequently
$U_{\rm red}\asymp Q\log Q$, whereas
$\min(U_0,U_G)\asymp Q^{3/2}\sqrt{\log Q}$.
At fixed confidence, the improvement factor is of order
$\sqrt{Q/\log Q}$, which tends to infinity as $n\to\infty$.
For sufficiently large $r$, the common $Ar$ term can dominate the base bound.

\Cref{sec:implementation} gives certified implementations and operation
counts for these successor-comparison rules. The arguments above do not
establish the same bounds for gradient signing.

\section{Implementation of the signing rules}
\label{sec:implementation}

We describe the potential evaluations and accuracy needed to implement
the signing rules. The operation counts use exact real arithmetic and unit-cost
evaluations of the indicated elementary functions. They do not include
the bit cost of realizing these primitives and do not assert a uniform
bit-complexity bound for arbitrary real-valued inputs.

\paragraph{Dense gradient rules.}
The dense signing rule needs only an upper approximation of its KKT root.
For the main theorem use
\begin{equation}\label{eq:rootaccuracy}
 \widehat L-L\le c\sigma^{-3}/\sqrt n.
\end{equation}
The derivative bound in \cref{lem:kkt} gives the signing error required
in \eqref{eq:errormgf}. For the two centered dense profiles, replace
$\sigma^{-3}$ by $\sigma^{-2}/\ell^2$ and $\sigma^{-2}/\ell$,
respectively, with $\ell=(\log(e\sigma))^{1/\alpha}$.
The bracket
$\eta\norm d_\infty+1<L\le\eta\norm d_\infty+\sqrt{2n}$
has width at most $\sqrt{2n}$; bisection therefore reaches these tolerances
in $O(\log(2n\sigma))$ evaluations, each costing $O(n)$ arithmetic operations.
These costs count arithmetic, comparisons, and square roots.

\paragraph{Centered sparse full-potential comparisons.}
For the general centered sparse profiles and $U_0$, the deterministic
rules compare the complete coupled potential $V$ at the two candidate
states. With the notation of \cref{prop:foster}, write
$V_\pm\colonequals V(d\pm X)$ and compute certified upper values satisfying
\[
 V_\pm\le\widehat V_\pm\le(1+\varepsilon_{\rm cmp})V_\pm,
 \qquad \varepsilon_{\rm cmp}\colonequals\delta_F/8.
\]
Select the smaller upper value, using fixed tie-breaking. The tolerance
concerns each positive successor value, rather than their difference.
This preserves contraction by at least $\delta_F/4$ and changes the
injection only by an absolute factor. No exact comparison of
transcendental values is needed. The comparison tolerance is fixed before
the first arrival and does not decrease with the horizon.

Certified dual upper values of the capped potential, together with
relative upper approximations of the positive exp/cosh terms, provide
these comparisons. \Cref{app:CappedImplementation} gives the scalar dual
formula and sufficient root and evaluation tolerances.

Write $M_*\colonequals 1+\norm d_\infty+\norm X_\infty$.
For the general centered profiles, $c=p$ and
$\delta_F\asymp\lambda\rho p/\sqrt k$ with the parameters in
\cref{sec:SparseOnlineDiscrepancy}. For the $U_0$ branch,
$c=\max(p/A^2,1/n)$ and $\delta_F\asymp(p/A^2)/\sqrt{nc}$.
The sufficient root tolerances in both cases have inverse polynomial
size in $n,\sigma$ and logarithmic tail factors. Bisection and two
successor evaluations use $O(n\log(2+n\sigma+M_*))$ elementary real
operations per arrival, including exp/cosh evaluations.
This model differs from the arithmetic-only dense gradient count.

\paragraph{Successor comparisons for two breakpoints.}
The $U_G$ and $U_{\rm red}$ rules use the potential $F_b$ with logarithmic
curvature bins from \cref{sec:SparseGaussianUpper}.
For either breakpoint $c$, put $q\colonequals p/A^2$ and
$\mathcal D\colonequals \sqrt{nc}+1+\kappa\sqrt{\log(e/c)}$.
Each normalized successor value $F_b(d\pm X)$ is certified from above
with additive error at most $c_{\rm tol}q/\mathcal D$.
The root bracket has width $O((\sqrt{nc}+\sqrt{\log(e/c)})/\kappa)$,
independent of the current input magnitude. Certified successor upper
values give the required comparison, using
$O(n\log(2n\sigma_*)\log\log(4n))$ elementary real operations per
arrival and $O(n+\log(2n))$ stored numbers.
The transfer to this implemented rule is established in
\cref{lem:BinnedImplementation}; \cref{app:BinnedImplementation} gives
the root bracket, capped scalar conjugates, and certified evaluation.
The Gaussian specialization fixes $\sigma_*=2$ and uses $c=p$, so its
rule depends only on $n,k$ and has cost
$O(n\log(2n)\log\log(4n))$.

The envelope algorithm in \cref{thm:SparseEnvelope} selects its branch
before the first arrival and has that branch's operation count.
For the general centered sparse profiles and $U_0$, the tolerances are
fixed in advance, while the operation count depends on the current state
and input through $M_*$. None of the tolerances or branch choices depends
on the confidence parameter or horizon.

All rules require their stated parameters, but no input laws, means,
or horizon. For arithmetic initialization of the tail factors,
$\log(e\sigma)$ can be replaced by a comparable quantity obtained by
doubling, with square roots for the sub-Gaussian factor.

\section{Lower bounds, sharpness, and further questions}
\label{sec:FurtherQuestions}

We study the dimension and sparsity scales, tail-parameter dependence,
and the contribution of large individual marks. Throughout this section,
$d_t\colonequals \sum_{s=1}^t x_sv_s$ denotes the unnormalized signed sum.
All lower bounds allow randomized signing rules with unrestricted computation.

\subsection{Fresh blocks and distributional lower bounds}
\label{sec:BlockLowerBounds}

A fresh suffix gives a lower bound after any earlier online history.
Let $Q(Y,b)\colonequals \sup_a\P(|Y-a|\le b)$ be the concentration function of a real
random variable. Put $g(u)\colonequals \P(|Z|\le u)$ for $u\ge0$, where $Z\sim N(0,1)$,
and let $c_G>0$ satisfy $g(c_G)=1/2$; numerically,
$c_G=0.67448975\ldots$.

\begin{samepage}
\begin{lemma}[Fresh product block]\label{lem:FreshProductBlock}
Let all entries of $U_1,\ldots,U_m\in\R^n$ be independent.
For every deterministic $d\in\R^n$ and $b\ge0$,
\begin{equation}\label{eq:ProductBlock}
 \P\left(\exists\varepsilon\in\{\pm1\}^m:
       \norm{d+\sum_{j=1}^m\varepsilon_jU_j}_\infty\le b\right)
 \le\sum_{\varepsilon\in\{\pm1\}^m}
       \prod_{i=1}^n Q\left(\sum_{j=1}^m\varepsilon_jU_j(i),b\right).
\end{equation}
\end{lemma}
\end{samepage}
\begin{proof}
For a fixed signing, coordinate independence factors the probability,
and each factor is bounded by its concentration function.
Take a union bound over signings.
\end{proof}

For an online rule at a prescribed $T\ge m$, condition on the complete
history at $T-m$. If the suffix is independent of that history and has the
stated product law, every possible continuation lies in the same union.
Thus the bound also controls $\P(\norm{d_T}_\infty\le b)$, even if the last $m$ signs
are chosen jointly after observing the block.
The one-step conditional assumptions of the upper theorems do not suffice
for the required independence of the suffix from the history.
For $m=T=n$, the argument also gives an offline bound. For $T>n$, an offline
algorithm may choose its earlier state using the suffix, so the argument
no longer applies.

For mutually independent standard Gaussian marks and $\Ber(p)$ masks, define
$q_b(0)\colonequals 1$ and $q_b(r)\colonequals g(b/\sqrt r)$ for $r>0$.
Conditional on the masks, a fixed signing has independent Gaussian
coordinates with variances equal to the row counts.
A Gaussian interval has greatest mass when centered at its mean:
the centered interval probability is even in the shift and decreases
with its absolute value, by differentiation.
Averaging the independent counts $N_i\sim\Bin(m,p)$ in
\cref{lem:FreshProductBlock} gives
\begin{equation}\label{eq:GaussianSuffix}
 \P(\norm{d_T}_\infty\le b)\le2^m\bigl[\E{q_b(N)}\bigr]^n,\qquad
 N\sim\Bin(m,p),\quad T\ge m.
\end{equation}
All right sides may be truncated at one.

\begin{proof}[Proof of Theorem~\ref{thm:GaussianLowerBound}]
Take $m=n$. For part~\ref{item:GaussianLowerBoundDense}, $p=1$ and
$b=\sqrt n/2$, so \eqref{eq:GaussianSuffix} becomes
$[2g(1/2)]^n<0.767^n$, since $2g(1/2)=0.765849\ldots$.
More generally the same calculation gives $[2g(c)]^n$ at threshold
$c\sqrt n$.

For part~\ref{item:GaussianLowerBoundSparse}, put $p\colonequals k/n$ and $b=\sqrt k/2$.
The Gaussian density bound and Cauchy--Schwarz give
\[
 \E{q_b(N)}
 \le e^{-k}+\sqrt{2/\pi}\,b\,
                 \bigl(\E{N^{-1}\I{N>0}}\bigr)^{1/2}.
\]
For every integer $r\ge1$,
\[
 \frac1r\le\frac1{r+1}
       +\frac1{(r+1)(r+2)}
       +\frac8{(r+1)(r+2)(r+3)}.
\]
Indeed, the remainder after the first two terms is
$2/[r(r+1)(r+2)]$, bounded by the third.
For $j\ge1$,
\begin{align*}
 \E\frac1{(N+1)\cdots(N+j)}
 &=\frac1{(j-1)!}\int_0^1(1-p+pt)^n(1-t)^{j-1}\,dt\\
 &\le\frac1{(j-1)!}\int_0^\infty e^{-ku}u^{j-1}\,du=k^{-j}.
\end{align*}
Thus $\E(N^{-1}\I{N>0})\le k^{-1}+k^{-2}+8k^{-3}$, and the exponential
base in \eqref{eq:GaussianSuffix} is at most
\[
 2e^{-k}+\sqrt{2/\pi}\sqrt{1+1/k+8/k^2}.
\]
It decreases with $k$ and is $0.8344478723\ldots<0.835$ at $k=16$.
Both parts use the full last-$n$ block, preserving the joint-sign allowance.
\end{proof}

Reducing the threshold from $\sqrt k/2$ to $\sqrt k/4$ retains the exponential
base $0.835$ for every $k\ge2$, and yields an exponential bound throughout
$k\ge1$.

\begin{corollary}[Gaussian lower bound at all sparsity levels]
\label{cor:GaussianAllLevels}
For independent standard Gaussian marks and independent $\Ber(k/n)$ masks,
every prescribed $T\ge n$ and every $1\le k\le n$ satisfy
\begin{equation}\label{eq:GaussianAllLevels}
 \P(\norm{d_T}_\infty\le\sqrt k/4)<0.988^n
\end{equation}
for every randomized online rule. For $k\ge2$, the right side can be
replaced by $0.835^n$. Both statements allow the last $n$ signs to be
chosen jointly.
\end{corollary}
\begin{proof}
For $N\sim\Bin(n,k/n)$ and every $b\ge0$,
$N^{-1/2}\I{N>0}\le\sqrt{2/(N+1)}$ and Jensen's inequality give
$\E q_b(N)\le e^{-k}+2b/\sqrt{\pi k}$. The fresh-$n$ estimate
\eqref{eq:GaussianSuffix} therefore implies
\begin{equation}\label{eq:GaussianSimpleBlock}
 \P(\norm{d_T}_\infty\le b)\le
 \left[2e^{-k}+\frac{4b}{\sqrt{\pi k}}\right]^n.
\end{equation}
Now set $b=\sqrt k/4$. For $k\ge2$, $2\E q_b(N)$ is at most
$2e^{-2}+1/\sqrt\pi=0.834860\ldots<0.835$.
For $1\le k\le2$, put $a\colonequals \sqrt{2/\pi}/4$.
Since $q_b(r)\le a\sqrt k$ for $r\ge1$,
\[
 \E q_b(N)\le F(k)\colonequals e^{-k}+a\sqrt k(1-e^{-k}).
\]
Differentiation gives $F'(k)<0$ on $[1,2]$, since
$a<0.20<0.27<2\sqrt k/(e^k+2k-1)$ there.
Hence $2F(k)\le2F(1)=0.9879384996\ldots<0.988$.
\end{proof}

\begin{corollary}[Independent marks with finite third moments]
\label{cor:ThirdMomentLower}
Let $1\le k\le n$ and let the marks $w_t(i)$ be independent, centered,
of variance one, with $\E|w_t(i)|^3\le\rho$ uniformly.
Independently use independent $\Ber(k/n)$ masks.
For every prescribed $T\ge n$, every $c>0$, and every randomized online rule,
\begin{equation}\label{eq:ThirdMomentLower}
 \P(\norm{d_T}_\infty\le c\sqrt k)
 \le\left[2g(c)+\frac{4C_{\mathrm{BE}}\rho}{\sqrt k}\right]^n,
\end{equation}
where $C_{\mathrm{BE}}=0.561$ is a valid absolute Berry--Esseen constant.
For each fixed $0<c<c_G$, the condition $k\ge C(c)\rho^2$ makes this at
most $e^{-\gamma(c)n}$ for a positive $\gamma(c)$.
The last $n$ signs may still be chosen jointly.
\end{corollary}
\begin{proof}
For a fixed signing of the last $n$ arrivals, the summands in each coordinate
are independent and centered; their variances sum to $k$, and their third
absolute moments sum to at most $k\rho$.
The nonidentically distributed Berry--Esseen theorem
\cite[Theorem~7]{Tyurin09} bounds the Kolmogorov distance between the normalized
sum and the standard normal by $C_{\mathrm{BE}}\rho/\sqrt k$, uniformly in
the signs. An interval probability involves two CDF endpoints, so its
concentration function at radius $c\sqrt k$ is at most
$g(c)+2C_{\mathrm{BE}}\rho/\sqrt k$.
Left limits handle atoms, since the comparison normal CDF is continuous.
Apply \cref{lem:FreshProductBlock} with $m=n$.
Since $2g(c)<1$ for $c<c_G$, taking $k\ge C(c)\rho^2$ gives the last assertion.
\end{proof}

This conclusion requires neither symmetry nor identical coordinate laws.
It includes dense inputs by setting $k=n$.
If fourth moments are at most $K$, one may take $\rho=\sqrt K$ by
Cauchy--Schwarz, so $k\ge C(c)K$ suffices.

The next result drops centering and the third-moment assumption for a fixed
common mark law, at the cost of a law-dependent threshold for $k$.

\begin{proposition}[A fixed law with finite variance]
\label{prop:FiniteVarianceLower}
Fix a mark law $W$ with $\operatorname{Var}(W)=1$.
Use independent copies as all marks and independent $\Ber(k/n)$ masks.
For every $0<c<c_G$, there are $k_0(W,c)<\infty$ and $\gamma(c)>0$ such
that, for every $n\ge k\ge k_0(W,c)$ and every prescribed $T\ge n$,
\begin{equation}\label{eq:FiniteVarianceLower}
 \P(\norm{d_T}_\infty\le c\sqrt k)\le e^{-\gamma(c)n}.
\end{equation}
This holds for every randomized online rule, including joint choice of
the last $n$ signs. The law need not be centered, symmetric, or have a
finite third moment.
\end{proposition}
\begin{proof}
Write $\mu\colonequals \E W$, $p\colonequals k/n$, and $Y_p\colonequals \xi W-p\mu$.
Then $\operatorname{Var}(Y_p)=p[1+(1-p)\mu^2]\ge p$.
For sums of $n$ independent copies with arbitrary fixed signs, put
$s^2\colonequals n\operatorname{Var}(Y_p)\ge k$.
For any fixed $\varepsilon>0$ and all sufficiently large $k$, the event
$|Y_p|>\varepsilon s$ is impossible when $\xi=0$.
When $\xi=1$, it implies $|W|>\varepsilon\sqrt k/2$.
Consequently the Lindeberg ratio is at most
\[
 \frac{\E{Y_p^2\I{|Y_p|>\varepsilon s}}}
      {\operatorname{Var}(Y_p)}
 \le2\E{(W^2+\mu^2)\I{|W|>\varepsilon\sqrt k/2}}\longrightarrow0.
\]
The bound is uniform in $p$ and the signs.
A sequence contradicting uniform normal approximation as $k=np\to\infty$
would itself satisfy the Lindeberg--Feller central limit theorem
\cite[Theorem~3.4.10]{Durrett19}.
Continuity of the normal CDF therefore gives a uniform Kolmogorov error.
Let $\delta_W(K)$ be the supremum of these errors over all $n,p$ with
$np\ge K$ and all signings; then $\delta_W(K)\to0$ and is nonincreasing.
The concentration function is invariant under translation, so the signed
means do not affect it. Since $s^2\ge k$,
$Q(\sum_j\varepsilon_j\xi_jW_j,c\sqrt k)\le g(c)+2\delta_W(k)$.
Choose $k_0$ so that $\delta_W(k_0)\le(1/2-g(c))/4$.
The block lemma then yields \eqref{eq:FiniteVarianceLower} with
$\gamma(c)\colonequals -\log(g(c)+1/2)>0$.
\end{proof}

The threshold $k_0(W,c)$ cannot be chosen uniformly over all unit-variance
mark laws.
If $W_n$ is zero with probability $1-n^{-3}$ and takes each value
$\pm n^{3/2}$ with probability $n^{-3}/2$, it has variance one.
For dense inputs at $T=n$, all entries vanish with probability at least
$1-1/n$. Thus no threshold independent of the mark law can make
\eqref{eq:FiniteVarianceLower} hold uniformly over this class.
These extensions express robustness of the block obstruction; they do not
identify the optimal proportional constants studied in
\cite{Fiedler26,JoLee26}.

\subsection{Tail and confidence dependence}
\label{sec:TailConfidenceLower}
\label{sec:TailSharpness}

\begin{proposition}[Sharpness at the stated confidence]\label{prop:SigmaSharpness}
For the sub-exponential class in \cref{thm:OnlineDiscrepancy}, the fourth
power of $\sigma$ cannot be reduced while retaining failure probability
$e^{-c\sigma^3\sqrt n}$ with absolute constants.
The analogous assertion holds for \cref{thm:SparseOnlineDiscrepancy}
with $\sqrt k$ in place of $\sqrt n$.
\end{proposition}
\begin{proof}
For $\sigma\ge2$, let $W\colonequals BL$, where $B\sim\Ber(2/\sigma^2)$ and,
independently, $L$ is Laplace with scale $\sigma/2$.
Then $W$ is symmetric with variance one and
$\norm W_{\psi_1}=(\sigma/2)(1+2/\sigma^2)\le\sigma$.
For $T=1$, signs cannot change coordinate magnitudes. A single coordinate
therefore gives
\[
 \P{\norm{v_1}_\infty>D}\ge
 \frac2{\sigma^2}e^{-2D/\sigma}.
\]
An error at most $e^{-c\sigma^3\sqrt n}$ requires
$D\ge(c/2)\sigma^4\sqrt n-\sigma\log\sigma+(\sigma/2)\log2$.
Even $n=1$ excludes any smaller power as $\sigma\to\infty$.
For sparse inputs, fix an admissible $n=2k$ and $p=1/2$, and use the same
marks. The extra mask changes the displayed prefactor to $\sigma^{-2}$.
This requires $D\ge(c/2)\sigma^4\sqrt k-\sigma\log\sigma$.
\end{proof}

This is sharpness of the joint threshold and confidence dependence.
It does not assert a fourth-power obstruction at constant confidence,
or for the smaller class of sub-Gaussian marks. For fixed tail parameter,
a unit-variance Laplace coordinate at $T=1$ also gives
$\P{\norm{v_1}_\infty>C\sqrt n}\ge e^{-\sqrt2C\sqrt n}$.
Thus failure $e^{-cn}$ at threshold $C\sqrt n$ is impossible for the full
sub-exponential class.

\begin{remark}[Tail-parameter factors at constant confidence]
\label{rem:TailParameterFactors}
Let $W$ take values $0,\pm A$ with total nonzero probability $A^{-2}$,
where $A\ge\sqrt2$. It is symmetric with variance one, and for
$\alpha\in\{1,2\}$,
\[
 \sigma_A\colonequals \norm W_{\psi_\alpha}
 =\frac A{[\log(1+A^2)]^{1/\alpha}},\qquad
 A\asymp\sigma_A[\log(e\sigma_A)]^{1/\alpha}
 \quad(A\to\infty).
\]
For independent copies and independent $\Ber(p)$ masks,
the number of nonzero inputs in each coordinate is $\Bin(T,p/A^2)$.
Odd parity forces magnitude at least $A$ for every signing, so
\begin{equation}\label{eq:AtomicParity}
 \P(\norm{d_T}_\infty<A)\le\left[\frac{1+(1-2p/A^2)^T}{2}\right]^n.
\end{equation}
For any fixed dimension and $p>0$, choose $T=\lceil A^2/p\rceil$.
A fixed coordinate has odd parity with probability at least
$(1-e^{-2})/2>e^{-1}$.
Taking $n=1,p=1$ gives the dense obstruction; any fixed $n,k$ gives its
sparse counterpart. Thus the factors $\sigma\log(e\sigma)$ and
$\sigma\sqrt{\log(e\sigma)}$ in the moderate-confidence profiles are
necessary at fixed dimension and confidence.
This does not establish joint optimality in all parameters.
\end{remark}

\begin{samepage}
\begin{lemma}[Singletons after an arbitrary past]\label{lem:SingletonSuffix}
Let $0<p\le1$, and let all marks be independent copies of
a law $W$ with an even density nonincreasing on $[0,\infty)$.
Use independent $\Ber(p)$ masks, independent of the marks.
For every prescribed $T\ge\lfloor1/p\rfloor$, every $b\ge0$, and every randomized online rule,
\begin{equation}\label{eq:SingletonSuffix}
 \begin{gathered}
 \P(\norm{d_T}_\infty\le b)\le[1-s\P(|W|>2b)]^n
       \le\exp\left[-\frac ne\,\P(|W|>2b)\right],\\
 s\colonequals mp(1-p)^{m-1},\qquad m\colonequals \lfloor1/p\rfloor.
 \end{gathered}
\end{equation}
The last $m$ signs may be chosen jointly after observing that suffix.
\end{lemma}
\end{samepage}
\begin{proof}
Condition on the complete history before the last $m$ arrivals, with
current state $d$.
A coordinate is hit exactly once in this block with probability $s$,
independently across coordinates.
If its unique mark is $W_i$, successful signing requires
$\bigl||d_i|-|W_i|\bigr|\le b$, regardless of how the signs are coupled.
The density of $|W_i|$ is nonincreasing on the positive half-line.
The interval $[\max\{0,|d_i|-b\},|d_i|+b]$ has length at most $2b$,
so its probability is at most $\P(|W|\le2b)$.
These necessary coordinate conditions are independent given $d$.
This proves the first inequality without a union bound over signs.

For fixed $m$, the function $mp(1-p)^{m-1}$ increases on
$(1/(m+1),1/m]$, with infimum $(m/(m+1))^m>e^{-1}$.
For $p=1$, $m=1$ and $s=1$ directly.
Thus $s\ge e^{-1}$; use $1-u\le e^{-u}$ and average over the past.
\end{proof}

For standard Gaussian marks and $n\ge3$, the elementary tail bound,
obtained by integration by parts,
$\P(|Z|>x)\ge\sqrt{2/\pi}\,x e^{-x^2/2}/(1+x^2)$ gives, at every
prescribed $T\ge\floor{1/p}$,
\begin{equation}\label{eq:GaussianLogLower}
 \P(\norm{d_T}_\infty\le\tfrac12\sqrt{\log n})
 \le\exp\left[-\frac1{e\sqrt{2\pi}}\sqrt{\frac n{\log n}}\right].
\end{equation}
\begin{corollary}[Gaussian dimension lower bound]
\label{cor:GaussianDimensionLower}
Fix an integer $n\ge3$ and a parameter $1\le k\le n$.
Let all marks be independent standard
Gaussian variables and, independently, let all masks be independent
$\Ber(k/n)$ variables. For every prescribed deterministic terminal time
$T\in\N$ with $T\ge n$ and every randomized online signing rule,
\begin{equation}\label{eq:GaussianCombinedLower}
 \P\left(\norm{d_T}_\infty\le\tfrac18(\sqrt k+\sqrt{\log n})\right)
 \le0.988^n+
       \exp\left[-\frac1{e\sqrt{2\pi}}\sqrt{\frac n{\log n}}\right].
\end{equation}
\end{corollary}
\begin{proof}
Combine \eqref{eq:GaussianLogLower} with \cref{cor:GaussianAllLevels}.
Indeed, if $\sqrt k\ge\sqrt{\log n}$, the threshold is at most $\sqrt k/4$;
otherwise it is at most $\sqrt{\log n}/4$.
\end{proof}
Thus $\sqrt k+\sqrt{\log n}$ is necessary up to constants at every
prescribed $T\ge n$, with probability tending to one as $n\to\infty$
uniformly over $1\le k\le n$. This combined bound concerns ordinary online rules; the shorter
singleton block does not preserve the full last-$n$ joint-sign allowance.

For unit-variance Laplace marks, $\P(|W|>x)=e^{-\sqrt2x}$, so
\begin{equation}\label{eq:LaplaceSuffix}
 \P(\norm{d_T}_\infty\le b)\le\exp[-(n/e)e^{-2\sqrt2b}].
\end{equation}
If $\P(\norm{d_T}_\infty\le b)\ge1-e^{-r}$ for $r\ge\log2$, then
$-\log(1-e^{-r})\le2e^{-r}$ implies
\begin{equation}\label{eq:LaplaceNecessaryConfidence}
 b\ge\frac{\log n+r-\log(2e)}{2\sqrt2}.
\end{equation}
The additive $\log n+r$ term in the sub-exponential sparse profiles is
therefore necessary already for this fixed mark law and at every
prescribed $T\ge n$.
For standard Gaussian marks and every prescribed $T\ge\floor{1/p}$,
the bound $\P(|Z|>x)\ge\tfrac12e^{-x^2}$ for $x\ge0$ similarly gives
\[
 4b^2\ge\log n+r-\log(4e)
\]
whenever the success probability is at least $1-e^{-r}$, with $r\ge\log2$.
These are online obstructions after an arbitrary past.
The simpler zero-state single-hit argument in the introduction also
shows the offline necessity of the $\log^2 n$ sparsity scale for the full
sub-exponential class at the special horizon $\lfloor n/k\rfloor$.

\paragraph{A joint sub-Gaussian lower profile.}
The next statement compares the sub-Gaussian parameters at a common prescribed horizon.

\begin{samepage}
\begin{proposition}[Joint parameter obstruction]\label{prop:JointSubGaussianLower}
Let $\sigma\ge2$, $r\ge2$, $1\le k\le n$, and $p\colonequals k/n$.
Use independent copies of each mark law with independent $\Ber(p)$ masks.
Let $T$ be a prescribed deterministic integer satisfying
\begin{equation}\label{eq:JointLowerHorizon}
 T\ge\max\{n,\ n\sigma^2/(4k),\ \sigma^2\log(e\sigma)/k\}.
\end{equation}
Suppose that for every symmetric unit-variance mark law with
$\psi_2$-norm at most $\sigma$, there is a randomized causal rule achieving
threshold $b$ at this same prescribed $T$ with failure probability at most
$e^{-r}$. The rule may be chosen separately for each law and may know that
law. Then necessarily
\begin{equation}\label{eq:JointSubGaussianLower}
 b\ge\frac1{60}\left[
 \sqrt k+\sigma\sqrt{\log(e\sigma)+\log(2n)+r}\right].
\end{equation}
Different admissible laws supply the three contributions.
\end{proposition}
\end{samepage}
\begin{proof}
The rare-atom, thinned-Gaussian, and standard-Gaussian laws below use,
respectively, the conditions $T\ge\sigma^2\log(e\sigma)/k$,
$T\ge n\sigma^2/(4k)$, and $T\ge n$ in \eqref{eq:JointLowerHorizon}.
Set $H\colonequals \log(e\sigma)$ and $A\colonequals \sigma\sqrt H$.
First let $W$ equal $\pm A$ with probabilities $1/(2A^2)$ and vanish
otherwise. It has variance one and
$\E e^{W^2/\sigma^2}=1+(e^H-1)/A^2\le2$, since
$e\sigma\le\sigma^2\log(e\sigma)$ for $\sigma\ge2$.
With $a\colonequals p/A^2$, \eqref{eq:AtomicParity} gives
\[
 \P(\norm{d_T}_\infty<A)\le\exp\left[-\frac n2(1-e^{-2/n})\right]<0.65.
\]
Here $Ta\ge1/n$ by \eqref{eq:JointLowerHorizon}, and
$n(1-e^{-2/n})\ge1-e^{-2}$. Thus $b\ge A$ is necessary.

Next take $W\colonequals (\sigma/2)BZ$, where $B$ is
$\operatorname{Ber}(4/\sigma^2)$ and $Z\sim N(0,1)$ are independent.
This law is symmetric with variance one, and
$\E e^{W^2/\sigma^2}=1+(4/\sigma^2)(\sqrt2-1)<2$.
The effective activation rate is $a\colonequals 4p/\sigma^2$.
Apply \cref{lem:SingletonSuffix} at this rate, to Gaussian marks
$(\sigma/2)Z$, using $m=\lfloor1/a\rfloor\le T$.
It gives
\[
 \P(\norm{d_T}_\infty\le b)\le\exp\left[-\frac ne\P(|Z|>4b/\sigma)\right].
\]
The elementary bound $\P(|Z|>x)\ge\tfrac12e^{-x^2}$ holds for $x\ge0$:
the difference between the two sides has negative derivative and tends
to zero. At $b=(\sigma/20)\sqrt{\log(2n)+r}$, the negative log of
this success bound is at least
$\mu\colonequals \frac n{2e}(2n)^{-1/25}e^{-r/25}$.
For $n\ge1,r\ge2$,
$\mu e^r\ge e^{23/25}/(2\,2^{1/25})>7/6$, whereas
$-\log(1-e^{-r})<(7/6)e^{-r}$. Hence
$b\ge(\sigma/20)\sqrt{\log(2n)+r}$ is necessary.

Finally use standard Gaussian marks, whose $\psi_2$-norm is
$\sqrt{8/3}<2$. At $b=\sqrt k/20$, the base in
\eqref{eq:GaussianSimpleBlock} is at most
$2/e+1/(5\sqrt\pi)<0.849<1-e^{-2}$.
Thus $b\ge\sqrt k/20$ is necessary.
The maximum of the three necessary thresholds dominates
\eqref{eq:JointSubGaussianLower}.
\end{proof}

\paragraph{Rare activations.}
Rare activations can force successive levels of imbalance beyond the
separate obstructions in \cref{prop:JointSubGaussianLower}, even when the
law and horizon are known.
A related isolated-intersection growth mechanism appears in
Altschuler and Tikhomirov \cite[Section~1.2.1 and Lemmas~2.2--2.3]{Altschuler26Binary}
for nonnegative binary vectors with fixed support size. Their
Proposition~2.5 also gives terminal lower bounds after a preceding history.
Here we consider independent Bernoulli--Rademacher coordinates, with
random support sizes and signed nonzero entries.

In the next proposition, activation probability $q$ gives $\rho=nq$
expected active coordinates per arrival. We use $h$ stages in blocks of
length $m\asymp1/q$. From $N$ coordinates at one height, a growth stage
produces roughly $\rho N^2/n$ at the next height. The proof first establishes
this block estimate, then initializes by parity and iterates.

\begin{proposition}[Rare-activation obstruction]\label{prop:RareActivationLower}
Let $n\ge2^{64}$ and $0<q\le1/n$. Let $v_t(i)\colonequals B_t(i)\epsilon_t(i)$,
where all $B_t(i)\sim\Ber(q)$ and all Rademacher variables
$\epsilon_t(i)$ are independent. Put
\begin{equation}\label{eq:RareActivationParameters}
 \begin{gathered}
 \rho\colonequals nq,\quad m\colonequals \ceil{1/q},\quad L\colonequals 16+\log_2(1/\rho),\\
 h\colonequals h(n,\rho)=\max\left\{1,1+\floor{\log_2\frac{\log_2 n}{4L}}\right\}.
 \end{gathered}
\end{equation}
For every randomized online signing rule and every prescribed
deterministic integer $T\ge hm$,
\begin{equation}\label{eq:RareActivationLower}
 \P\left(\norm{\sum_{t\le T}x_tv_t}_\infty\ge h\right)\ge0.99.
\end{equation}
The bound also holds conditionally on the complete history through
$T-hm$. The rule may know the law and the horizon.
\end{proposition}
\begin{proof}
We first prove a one-block estimate from any integer state having
$N$ coordinates of magnitude at least $j\ge1$, where $2\le N\le n$.
In a fresh block of $m$ arrivals, a coordinate is retained if it is hit
exactly once. Its retention probability is
$r_0\colonequals mq(1-q)^{m-1}\ge1/4$: here $mq\ge1$ and
$(1-q)^{1/q}\ge1/4$ for $q\le1/2$.
Conditional on retention, its hit time is uniform among the $m$ arrival times,
independently across coordinates.

Split the $N$ coordinates into deterministic sets of sizes
$\floor{N/2}\ge N/3$ and $\ceil{N/2}\ge N/2$.
Let $B$ be the number of distinct arrival times occupied by retained coordinates
from the first set. Since $N\le n\le m$, each exposure hits a previously
unoccupied arrival time with probability at least $r_0/2\ge1/8$.
Consequently $\P(B<N/48)\le e^{-N/192}$.
Match at most one retained coordinate from the second set to each such
arrival time. Until half the $B$ arrival times are matched, each new
coordinate creates a match with probability at least $B/(8m)$.
The target $NB/(32m)$ is at most $B/32$, below the stopping level $B/2$.
Thus stopped binomial domination and Chernoff's inequality give
\[
 \P\left(D<\frac{NB}{32m}\,\middle|\,
       \text{first-set retention labels and hit times}\right)
 \le e^{-NB/(128m)},
\]
where $D$ is the number of matches. Thus, except with probabilities
$e^{-N/192}$ and $e^{-N^2/(6144m)}$, there are at least
$N^2/(1536m)$ disjoint pairs, each pair retained at a distinct arrival time.
The pairing uses only masks and indices.

Each paired coordinate keeps its initial value until its unique hit
and is untouched afterwards. Conditional on the complete history before
the block and all block masks, the mark signs oriented relative to these
initial states are independent Rademacher variables. Opposing signs in
a pair, an event of probability $1/2$, force one coordinate to magnitude
at least $j+1$ at the block's end for either common signing. These events
are conditionally independent across pairs. A Chernoff bound gives at least
$D/4$ pairs that force growth, except with probability $e^{-D/16}$.
If $N'$ is the number of coordinates at height at least $j+1$ after
the block, these three estimates imply
\begin{equation}\label{eq:RareLevelStep}
 \P\left(N'\ge\frac{N^2}{8192m}\,\middle|\,
          \text{history before the block}\right)
 \ge1-3e^{-N^2/(32768m)}.
\end{equation}
We used $N\le m$ to absorb the first error into the displayed bound.
Since $m\le2/q$, the next height has at least
$\rho N^2/(16384n)$ coordinates, except with probability
$3e^{-\rho N^2/(65536n)}$.

For the first block, a coordinate's terminal parity is its initial
parity plus that of $\Bin(m,q)$. The odd-count probability is
$[1-(1-2q)^m]/2\ge(1-e^{-2})/2\ge3/8$, and its complement is
also at least $3/8$. Conditional on the initial state these parity
events are independent, regardless of the signs chosen by the rule.
Hence at least $n/4$ coordinates have magnitude at least one, except
with probability $e^{-n/48}$.

Set $a_1\colonequals 1/4$ and $a_{j+1}\colonequals \rho a_j^2/16384$. Explicitly,
\[
 a_j=\frac{16384}{\rho}
          \left(\frac{\rho}{65536}\right)^{2^{j-1}}.
\]
Conditional on success through stage $j$, there are at least $na_j$ coordinates
at height $j$, and the next stage fails with conditional probability at most
$3e^{-na_{j+1}/4}$ by \eqref{eq:RareLevelStep}.
If $h\ge2$, its definition gives
$2^{h-1}\le(\log_2 n)/(4L)$ and therefore
$na_h\ge(16384/\rho)n^{3/4}$.
The sequence $a_j$ decreases, so all needed sets contain at least two
coordinates. The total failure probability is at most
\[
 e^{-n/48}+3(h-1)e^{-4096n^{3/4}}<0.01
 \qquad(n\ge2^{64}).
\]
For $h=1$ the parity bound alone suffices.
For any prescribed $T\ge hm$, condition on time $T-hm$ and apply
the same argument to the fresh final $h$ blocks. The state is integer,
and causality keeps it independent of those fresh inputs.
This proves the conditional terminal bound.
\end{proof}

\paragraph{Consequences for unit-variance marks.}
For $n\ge2^{64}$, let a mark equal $\pm a$, each with probability $1/(2a^2)$, and
otherwise zero, where $a\ge1$. It is centered with variance one, and
its $\psi_2$-norm is at most $\sigma$ exactly when
\begin{equation}\label{eq:RareAtomAdmissibility}
 1+\frac{e^{a^2/\sigma^2}-1}{a^2}\le2,
 \qquad\text{equivalently}\qquad
 a^2\le\sigma^2\log(1+a^2).
\end{equation}
Under independent physical masks with $p=k/n$, dividing the inputs by $a$
gives activation $q=k/(na^2)$. If $a^2\ge k$, then
\cref{prop:RareActivationLower} applies with $\rho=k/a^2$ and
$m=\ceil{na^2/k}$, forcing raw discrepancy at least $ah$ with
probability at least $0.99$ at every prescribed $T\ge hm$.

In particular, take $a=\sqrt n$, $k=n$, and
$\sigma_n^2\colonequals n/\log(n+1)$. The exponential moment in
\eqref{eq:RareAtomAdmissibility} equals two, so $\sigma_n$ is the
exact $\psi_2$-norm. With $h=h(n,1)$ from
\eqref{eq:RareActivationParameters}, every rule has discrepancy at
least $\sqrt n\,h=\Omega(\sqrt n\log\log n)$ with probability
at least $0.99$ at every prescribed $T\ge nh$.
The profile suggested by the separate obstructions at fixed confidence
\[
 F_*\colonequals \sqrt k+\sigma\sqrt{\log(e\sigma)+\log(2n)}
\]
is only $O(\sqrt n)$ here and therefore cannot be a uniform upper
bound even at fixed confidence.

At $k=n$, an additional term of order $\sqrt k\log\log n$ would account
for this obstruction. The next example shows that the corrected profile
$F_*+\sqrt k\log\log(e^en)$ is still insufficient uniformly.
For $n\ge2^{256}$ retain the same marks and $\sigma_n$, but take
$k=\floor{n/\log n}$. Then $\rho=k/n\asymp1/\log n$,
$m=\ceil{n^2/k}$, and
$h=\log_2\log_2 n-\log_2\log_2\log n+O(1)$.
The lower bound remains $\Omega(\sqrt n\log\log n)$ at every prescribed
$T\ge hm$, while the corrected profile is $O(\sqrt n)$.
These are terminal obstructions for causal rules; they do not give
offline lower bounds for the entire suffix.

\paragraph{A matching upper bound.}
The lower bound is attained, up to absolute constants, for the dense
witness law by a simple priority rule. At each arrival, this rule selects
an active coordinate of largest current absolute discrepancy, breaking
ties by smallest index. If the selected discrepancy is nonzero, it chooses
the sign that moves this coordinate toward zero. If the support is empty
or all active discrepancies are zero, it chooses $+1$.

\begin{proposition}[Matching upper bound for the rare-activation law]
\label{prop:RareActivationUpper}
Let $n\ge1$ be an integer and let all entries of the input vectors be
independent with
\[
 \P(v_t(i)=\sqrt n)=\P(v_t(i)=-\sqrt n)=\frac1{2n},\qquad
 \P(v_t(i)=0)=1-\frac1n.
\]
Starting from $d_0=0$, the deterministic priority rule above satisfies
\begin{equation}\label{eq:RareActivationUpper}
 \P\left(\norm{d_T}_\infty>
 1000\sqrt n\,[1+\log\log(en)]\right)<\frac1{100}
\end{equation}
at every prescribed deterministic integer $T\ge0$.
The rule requires $O(n)$ elementary arithmetic operations and comparisons
per arrival and $O(n)$ stored numbers, without knowledge of the horizon.
\end{proposition}

\Cref{app:RareActivationUpper} gives the proof.
Together with \cref{prop:RareActivationLower}, this determines the
fixed-confidence discrepancy order $\Theta(\sqrt n\log\log n)$ for this
law at sufficiently long prescribed horizons.
For the entire centered unit-variance class with
$\norm W_{\psi_2}\le\sigma_n$, however, \cref{thm:SparseEnvelope}
still gives the uniform upper bound $O(\sqrt n\log n)$ at fixed confidence.
The remaining factor $\log n/\log\log n$ concerns that full class,
rather than the special law just solved.

\subsection{Further questions}
\label{sec:OpenQuestions}
We discuss three remaining questions, concerning offline Gaussian signing,
Gaussian confidence, and the joint dependence on sparsity and tails.

\paragraph{Offline discrepancy.}
For fixed $\beta>1$ and $T=\ceil{n^\beta}$, the offline bound announced in
\cref{sec:RelatedWork} concerns the scale $O_\beta(\sqrt{n/\log(T/n)})$
for independent standard Gaussian inputs. Its proof is deferred to separate
work. The broader question is the smallest discrepancy attainable with high
probability by randomized polynomial-time algorithms in this regime.
This returns to \cite[Question~1]{Turner20}, whose
generalized Karmarkar--Karp guarantee for Gaussian inputs applies in a much
lower-dimensional regime.

The small-margin algorithmic scale announced by Huang and Sellke in
Sun's survey \cite[Table~1]{Sun26ICM} is
$\alpha_{\ALG}(\kappa)=\Theta(\kappa^2\log(1/\kappa))$.
With $\alpha=n/T$ and target margin $\kappa=D/\sqrt T$, this formally
corresponds to $D\asymp\sqrt{n/\log(T/n)}$.
The IAMP framework of Huang, Sellke and Sun
\cite[Theorem~1.16(a)]{Huang26} controls the empirical distribution of
absolute normalized constraint values at fixed aspect ratio; the repair
needed to satisfy every constraint is deferred to companion work
\cite[p.~160]{Huang26}.
Applying this approach along polynomial horizons, where $n/T\to0$ and
the target margin $D/\sqrt T\to0$, requires quantitative control of the
dimension, error bounds, and running time, together with that repair.

Conditional computational barriers concern smaller offline targets.
Vafa and Vaikuntanathan~\cite[Theorem~3 and Corollary~2]{Vafa25}
consider symmetric-perceptron margin
$\kappa(\tau)\colonequals \tau^{-1/2-\varepsilon}$ for fixed
$\varepsilon>0$, where $\tau\colonequals T/n$.
This corresponds to discrepancy $\sqrt n\,(T/n)^{-\varepsilon}$.
Their reduction applies even to solvers with inverse-polynomial success
probability and uses sufficiently large polynomial column counts depending
on $\varepsilon$.
Under their stated worst-case lattice-hardness assumptions, it excludes
uniform polynomial-time solvers at those targets. It does not exclude
a bound of order $\sqrt{n/\log(T/n)}$ at fixed polynomial horizons.
Their logarithmic variant, under stronger subexponential lattice-hardness
assumptions, concerns discrepancy $\sqrt n/\log^{1+c}(T/n)$ for $c>3$
and uses superpolynomial column counts \cite[Remark~1]{Vafa25}.
It likewise does not rule out the scale $\sqrt{n/\log(T/n)}$ at fixed polynomial horizons.

\paragraph{Gaussian confidence.}
The Gaussian confidence cost in \cref{cor:SparseGaussian} is linear in $r$.
The square-root target matches the necessary dimension and confidence
scales established in \cref{sec:TailConfidenceLower}.
We conjecture that a single horizon-independent rule can attain it.
In both \cref{conj:GaussianConfidence,q:JointTailDependence}, we require
the algorithm to use polynomially many operations per arrival in the
real evaluation model of
\cref{sec:implementation}.

\begin{conjecture}[Gaussian confidence profile]\label{conj:GaussianConfidence}
There exists an absolute constant $C>0$ such that, for every positive
integer $n$ and every $1\le k\le n$, there is a deterministic online
algorithm depending only on $n,k$ with the following guarantee.
Under the Gaussian-mark and independent-mask assumptions of
\cref{cor:SparseGaussian},
\[
 \P\left(\norm{\sum_{t=1}^T x_tv_t}_\infty>
 C\left[\sqrt k+\sqrt{\log(2n)+r}\right]\right)\le e^{-r}
\]
for every prescribed deterministic integer $T\ge0$ and every $r\ge0$.
The same rule applies to all horizons and confidence parameters.
\end{conjecture}

Taking $k=n$ and $r=n$ in \cref{conj:GaussianConfidence} would give
$O(\sqrt n)$ terminal discrepancy with failure probability at most $e^{-n}$.
The guarantee concerns each prescribed terminal time separately.
The Laplace obstruction in \eqref{eq:LaplaceNecessaryConfidence} excludes
this confidence profile for the full sub-exponential class.

\paragraph{Joint parameter dependence.}
For general centered unit-variance sub-Gaussian marks, the optimal joint
dependence on $n,k,\sigma$, and confidence remains open.
The envelope in \cref{thm:SparseEnvelope} provides complementary upper
bounds, while \cref{prop:RareActivationLower} rules out the sum of the
separate obstructions in \cref{prop:JointSubGaussianLower} as a uniform
target. \Cref{prop:RareActivationUpper} attains the necessary order for
the atom law used in the lower bound. The following question asks whether
the same order holds uniformly over the admissible iid class at those parameters.

\begin{question}[Joint parameter dependence]\label{q:JointTailDependence}
For each integer $n\ge16$, put $\sigma_n^2\colonequals n/\log(n+1)$.
Does there exist an absolute constant $C>0$ and, for each such $n$, a
deterministic online algorithm depending only on $n$ such that, uniformly
over all centered unit-variance laws $W$ with
$\norm W_{\psi_2}\le\sigma_n$, for inputs whose entries $v_t(i)$ are
independent copies of $W$,
\[
 \P\left(\norm{\sum_{t=1}^T x_tv_t}_\infty>
 C\sqrt n\,\log\log(e^e n)\right)\le\frac1{10}
\]
at every prescribed deterministic integer $T\ge0$, without knowledge of
$W$ or $T$?
\end{question}

The atom law in \cref{prop:RareActivationUpper}
is centered with variance one and has exact $\psi_2$-norm $\sigma_n$.
Even though $k=n$ and there is no additional masking, each arrival has
only one nonzero coordinate on average.
Its fixed-confidence discrepancy order is now determined, while
\cref{thm:SparseEnvelope} gives the uniform upper bound
$O(\sqrt n\log n)$ for the full admissible class.
A positive answer to \cref{q:JointTailDependence} would close the
remaining gap for the iid class; it would not settle the broader conditional model.
Whether mixtures of different amplitude scales create a stronger
obstruction remains open.

\appendix
\numberwithin{equation}{section}
\section{Certified implementation of the signing rules}
\label[appendix]{app:CertifiedImplementation}

We give the evaluation details supporting \cref{sec:implementation}.
The first part treats the coupled capped-potential rules; the second
proves \cref{lem:BinnedImplementation} for the two curvature breakpoints.
All operation counts use the real evaluation model stated in
\cref{sec:implementation}.

\subsection{Certified comparisons for the capped potential}
\label[appendix]{app:CappedImplementation}

The deterministic rules for the general centered sparse profiles and
$U_0$ compare the positive successor values of the coupled potential $V$.
The following scalar dual formula supplies certified upper values of
its capped-potential component.

For geometric cap $c$, put $c_0\colonequals \sqrt c/\eta$ and
\[
 h_c(q)\colonequals \begin{cases}2\sqrt c-cq,&q\le c^{-1/2},\\q^{-1},&q>c^{-1/2}.
 \end{cases}
\]
Let $L_*$ denote the exact root in \eqref{eq:root} for the state $d$.
The scalar dual objective
\[
 G_d(L)\colonequals c_0\left[L+\sum_{j=1}^{2n}h_c(L-g_j/c_0)\right],
 \qquad g\colonequals (d,-d),
\]
satisfies $G_d(L_*)=\F_c(d)$ and
$0\le G_d(\widehat L)-\F_c(d)\le c_0(\widehat L-L_*)$
for $\widehat L\ge L_*$, since its derivative lies in $[0,c_0]$ there.
An upper root error at most $c_{\rm acc}\delta_F/(\lambda c_0)$, together
with certified relative upper errors at most $c_{\rm acc}\delta_F$ for
the positive exp/cosh terms, therefore certifies the required full-$V$
comparison: for either successor,
\[
 V\le\widehat V
 \le(1+c_{\rm acc}\delta_F)e^{c_{\rm acc}\delta_F}V
 \le(1+\delta_F/8)V
\]
for a sufficiently small absolute $c_{\rm acc}>0$.

The common root bracket has width at most
$\eta M_*/\sqrt c+\sqrt{2n}$ at either candidate, where
$M_*=1+\norm d_\infty+\norm X_\infty$.
The values of $c,\delta_F$, and $\lambda$ in
\cref{sec:implementation} make the sufficient root tolerances inverse
polynomial in $n,\sigma$ and the logarithmic tail factors. Bisection and
two successor evaluations therefore have the operation count stated there.
The error bounds concern positive successor values, so even an exact
or near tie requires no exact comparison of transcendental values.

\subsection{Transfer and certified evaluation for two breakpoints}
\label[appendix]{app:BinnedImplementation}

We retain the normalized notation of \cref{sec:SparseGaussianUpper},
the moment condition \eqref{sg:eq:moment}, and the parameter choices
in \eqref{sg:eq:params}.
In particular, $f_b$ is the implemented scalar regularizer defined there,
$F_b$ is its potential, and $\psi_b\colonequals(f_b')^{-1}$.
Scalar quantities formed from $f_b$ carry a superscript $b$.

\begin{proof}[Proof of \cref{lem:BinnedImplementation}]
This binned curvature is between one and $\sqrt2$ times the continuous
one. The inverse $\psi_b$ is convex, and the
geometry \eqref{sg:eq:geometry}--\eqref{sg:eq:domination} survives with
absolute constants.
More explicitly, the upward derivative jumps of $\psi_b$ at the images
of $c$ and the bin boundaries under $f_b'$ preserve Jensen's multiplier
and opposite-weight bounds.
Above $c$, the binned Euler summand satisfies
\[
 zf_b'(z)-f_b(z)
 =2c/\eta+\int_c^z r f_b''(r)\,dr
 \le2c/\eta+\sqrt2\,z/\kappa.
\]
The lower branch and $F_b(0)=4\sqrt{2\nu}/\eta$ are unchanged.
Summation gives an Euler-gap bound of the same order, and integration of $f_b'$
gives $f_b(1)\le CS/\kappa$. The restricted inner/outer geometry and
the high-state mass split use precisely these retained properties.
No equality of the continuous and binned potentials is assumed.

The binned outgoing logarithmic speed is no larger than
$\kappa\sqrt{1+\log(r/c)}$, and the inward speed from $z\ge c$
is at most $\kappa s_z$. Below $c$ the branch is unchanged.
The envelopes and the buffered small-weight estimates in the argument
for low coordinates therefore remain valid with the same constants, including the uniform
$\alpha z$ estimate and the actual-projection tilt.
For heavy terms use the actual binned inward increment
$D_z^b(H)\colonequals -L_z^b(-H)$ and define
\[
 \delta_i^b\colonequals
 \lambda^{-1}\E\bigl[1-e^{-\lambda D_{z_i}^b(H_i)}\bigr],
 \qquad H_i\colonequals |Y_i|.
\]
The expectation uses the respective conditional coordinate law given the fixed past.
The inward increments satisfy
$zH/(1+\kappa s_zH)\le D_z^b(H)\le zH$, and $P_z^b$ has the same
outgoing upper envelope. The multiplicative comparison of $S_i^b$ with
$D_{z_i}^b$ and \eqref{sg:eq:amplitude} therefore use
$\delta_i^b$ directly. Both collision
error and negative drift are compared to this same actual binned deficit.
No comparison to a mismatched continuous deficit is needed.
With all optimizer-dependent quantities formed from $F_b$, the preceding
bounds yield the binned counterpart of \eqref{sg:eq:global-drift}.
The retained geometry also gives the binned counterpart of
\eqref{sg:eq:high} whenever $F_b(d)\ge H_*$.
Equation~\eqref{sg:eq:eta-budget} therefore yields \eqref{sg:eq:Foster}
with $F_b$ in place of $F$, with constants depending only on $M$.

If $M_d\colonequals \norm d_\infty$, the multiplier $\theta_b$ of $F_b(d)$
is bracketed by
\[
 M_d-f_b'(1)<\theta_b\le M_d-f_b'(1/(2n)),
\]
with width at most
$C_M[b/\eta+S/\kappa]\le C_MB_{\rm red}$,
$B_{\rm red}\colonequals (b+S)/\kappa$.
For $g=(d,-d)$, use the capped conjugate and dual upper value
\[
\begin{aligned}
 f_b^*(s)&\colonequals \sup_{0\le z\le1}\{sz-f_b(z)\},\\
 G_b(\theta)&\colonequals \theta+\sum_{j=1}^{2n}f_b^*(g_j-\theta).
\end{aligned}
\]
The root equation is $\sum_{j=1}^{2n}\psi_b(g_j-\theta_b)=1$.
Above $\theta_b$, $G_b$ has derivative in $[0,1]$.
Consequently,
\[
\begin{aligned}
 G_b(\theta_b)&=F_b(d),\\
 0\le G_b(\widehat\theta)-F_b(d)&\le\widehat\theta-\theta_b
 \qquad(\widehat\theta\ge\theta_b).
\end{aligned}
\]
Thus the upper endpoint of a bisection bracket of width at most
$\varepsilon$ gives a certified upper value with additive error at most
$\varepsilon$.
Compute both successor values from above with certified total errors
at most $\varepsilon$ from \eqref{sg:eq:tolerance}.
Choose the smaller upper value with fixed ties. Its exact potential
exceeds the exact minimum by at most $\varepsilon$.
For sufficiently small fixed $c_{\rm tol}$, multiplication by
$e^{\lambda\varepsilon}$ retains a recursion of the form
\eqref{sg:eq:Foster} for $F_b$, with contraction at least $\delta/2$
and an adjusted injection constant depending only on $M$.

Certified evaluation only uses curvature bins intersecting $[c,1]$.
There are $O(\log(2/c))=O(\log(2n))$ such bins.
Inverse evaluations use $\min\{1,\psi_b(s)\}$: the cap applies for
$s\ge f_b'(1)$, the unchanged lower branch for $s\le f_b'(c)$,
and only the intermediate range needs a bin lookup.
The clipped inverse sum has the same root: all terms are positive,
so any term equal to one makes the sum exceed one.
Locating the relevant bin, using its endpoints or their images under
$f_b'$, costs $O(\log\log(4n))$ operations. Bisection needs
\[
 O_M(\log(B_{\rm red}\mathcal D/q))=O_M(\log(2n/\tau))
\]
iterations: $q\ge \tau/n$, $b\le\sqrt n$, $S\le\sqrt{\log(en)}$,
$R\le n$ and $\kappa^{-1}=O_M(\sqrt{\log(en)})$.
Each iteration uses $2n$ scalar terms per successor.
Within a bin the scalar formulas use arithmetic, square roots,
logarithms and exponentials. Storage is $O(n+\log(2n))$.
No latent mask or mark reconstruction, distributional oracle or
prescribed horizon is needed. The same implemented rule therefore satisfies \eqref{sg:eq:all-confidence}
for every $r$ and every prescribed deterministic $T$.
\end{proof}

\section{A matching upper bound for the rare-activation law}
\label[appendix]{app:RareActivationUpper}

We prove \cref{prop:RareActivationUpper} by dividing the arrivals and
state by $\sqrt n$. In this appendix, write
$U_t(i)\colonequals \xi_t(i)\varepsilon_t(i)$, where all masks are
independent $\Ber(1/n)$ and all marks are independent Rademacher variables,
also independent of the masks. Put
$d_t\colonequals \sum_{s\le t}x_sU_s$ and
$a_i(t)\colonequals |d_t(i)|$, starting at $d_0=0$.
The priority rule in \cref{sec:TailConfidenceLower} is invariant under this
positive rescaling. Let $\Fc_t$ contain the complete history through time $t$.
All logarithms are natural unless a base is indicated.

The proof first establishes an exponential moment bound at every
prescribed time. In each stage of a finite suffix, initially high
coordinates are shown, with high probability, to visit below the relevant
level at least once. Subsequent occupancy is then controlled by counting
entries. A quadratic bound on conditional expected entry counts yields
a squaring recursion for deterministic occupancy bounds.
The suffix length is used only in the analysis;
the signing rule does not know it or the horizon.

Related priority corrections for exceptional rows and doubly exponential
level estimates appear in
\cite[Algorithm~3.2 and Section~3.5]{Altschuler26Binary}, for nonnegative
binary vectors with fixed support size.

\subsection{Coordinate transitions and an exponential moment}

Order the coordinates by decreasing $a_i(t)$ with index ties, and let
$r_i$ be the resulting rank, starting at one. At a fixed state with
$a_i>0$, coordinate $i$ is selected exactly when it is active and all
$r_i-1$ earlier coordinates are inactive. Conditional on the masks,
if $i$ is active but not selected, the selected mark is independent of
$\varepsilon_i$; hence $i$ moves inward or outward with equal probability.
Writing $p\colonequals 1/n$ and
$\theta_i\colonequals p(1-p)^{r_i-1}$ gives the exact transition probabilities
\begin{align}
 \P(a_i'=a_i+1\mid\Fc_t)&=(p-\theta_i)/2,\label{atom:eq:out}\\
 \P(a_i'=a_i-1\mid\Fc_t)&=(p+\theta_i)/2.
\end{align}
For $n=1$, the factor $(1-p)^0$ is interpreted as one. For every $n$,
\[
 \theta_i\ge\frac p4,
 \qquad
 \P(a_i'=a_i+1\mid\Fc_t)
 \le\frac{r_i-1}{2n^2}.
\]
The first bound uses $(1-1/n)^{n-1}\ge1/4$ for $n\ge2$;
the second uses $1-(1-p)^r\le pr$.

Put $\alpha\colonequals 1/8$ and $c\colonequals 1/64$. For a positive height,
\begin{align}
 \E\bigl[e^{\alpha(a_i'-a_i)}\mid\Fc_t\bigr]
 &=1+p(\cosh\alpha-1)-\theta_i\sinh\alpha\notag\\
 &\le1+p(\alpha^2-\alpha/4)
 =1-\frac c n.\label{atom:eq:coordinate-mgf}
\end{align}
Here $\cosh\alpha-1\le\alpha^2$ and $\sinh\alpha\ge\alpha$.
At height zero, $a_i'=1$ with probability $p$, independently of the
selected sign. Define
\[
 M_t\colonequals \sum_{i=1}^n e^{\alpha a_i(t)}.
\]
Summing the positive-height drift and the zero-height contribution gives
\begin{equation}\label{atom:eq:Mdrift}
 \E\bigl[M_{t+1}\mid\Fc_t\bigr]
 \le(1-c/n)M_t+(c+e^\alpha-1)
 \le(1-c/n)M_t+\frac14.
\end{equation}
For the last inequality, $e^{1/8}-1\le1/7$ suffices. Since $M_0=n$,
iteration proves
\begin{equation}\label{atom:eq:uniformM}
 \E M_t\le16n\qquad\text{for every }t\ge0.
\end{equation}

\begin{lemma}[Maximum on a finite deterministic interval]\label{atom:lem:maxM}
For any deterministic integer $s\ge0$, integer length $m\ge0$, and $B>0$,
\[
 \P\!\left\{\max_{s\le t\le s+m}M_t>B\right\}
 \le\frac{16n+m/4}{B}.
\]
\end{lemma}
\begin{proof}
By \eqref{atom:eq:Mdrift}, $M_t+\frac14(s+m-t)$ is a nonnegative
supermartingale on that interval. Stop at its first $M_t>B$ or the
terminal time; bounded optional stopping and \eqref{atom:eq:uniformM}
give the displayed bound. Equivalently, sum the drift bound $\E\Delta M
\le1/4$ to the bounded stopping time. No overshoot bound is needed.
\end{proof}

\subsection{Level entries and clearance}
For integer $j\ge1$, let
\[
 N_j(t)\colonequals \#\{i:a_i(t)\ge j\}.
\]
For $j\ge2$ define the number of entries into level $j$ in one step by
\[
 C_j(t)\colonequals \#\{i:a_i(t)=j-1,\ a_i(t+1)=j\}.
\]
An eligible coordinate has rank at most $N_{j-1}(t)$. Summing
\eqref{atom:eq:out} over such coordinates therefore proves the conditional
statewise inequality
\begin{equation}\label{atom:eq:entries}
 \E\bigl[C_j(t)\mid\Fc_t\bigr]\le\frac{N_{j-1}(t)^2}{2n^2}.
\end{equation}
The estimate holds conditionally at each state, before taking expectations.

For a deterministic start $s$, integer level $j\ge1$ and coordinate $i$,
write $\tau\colonequals \inf\{u\ge0:a_i(s+u)<j\}$. For every integer $w\ge0$,
killing the process on this first exit and applying
\eqref{atom:eq:coordinate-mgf} while it is alive gives
\begin{equation}\label{atom:eq:clear}
 \P(\tau>w\mid\Fc_s)
 \le e^{\alpha a_i(s)}(1-c/n)^w
 \le e^{\alpha a_i(s)}e^{-w/(64n)}
 \quad\text{if }a_i(s)\ge j.
\end{equation}
If $a_i(s)<j$, then $\tau=0$. To verify the killed estimate, multiply
$e^{\alpha a_i(s+u)}$ by $\mathbf1_{\{\tau>u\}}$ and iterate;
discarding paths killed in the next step can only decrease its conditional
expectation. On $\{\tau>w\}$, $e^{\alpha a_i(s+w)}\ge1$, so the
conditional expectation of the killed process bounds the survival probability.

For $j\ge2$, the following deterministic observation is central. Suppose every
coordinate initially at height at least $j$ visits a height below $j$
between $s$ and $s+w$. Then for every $t\ge s+w$,
\begin{equation}\label{atom:eq:entrytracking}
 N_j(t)\le\sum_{u=s}^{t-1}C_j(u).
\end{equation}
Indeed each coordinate high at time $t$ either started below $j$ or has
since visited below $j$. Because increments have size at most one, it
must have an entry from $j-1$ to $j$ after time $s$. Distinct coordinates
give distinct counted entries. Repeated entries only make the right-hand
side larger. From the zero state, \eqref{atom:eq:entrytracking} holds for every
$t\ge0$ with $s=0$, without any waiting period.

\subsection{A finite suffix and the level recursion}
The following deliberately large constants cover every $n$ without a
separate small-dimension argument. Define deterministic analysis parameters
\begin{align*}
 R&\colonequals 10^6+\log(n+1),&
 J&\colonequals \lceil\log_2 R\rceil,&
 w&\colonequals \lceil1024nR\rceil,\\
 m&\colonequals Jw,& B&\colonequals nR^4,& j_0&\colonequals \lceil64\log R\rceil.
\end{align*}
We use the elementary bounds
\begin{equation}\label{atom:eq:parameters}
 J\le\sqrt R,\qquad w\le1025nR,\qquad n\le e^R.
\end{equation}
Set
\[
 \beta_0\colonequals R^{-4},\qquad \beta_{\ell+1}\colonequals R^3\beta_\ell^2,
 \qquad b_\ell\colonequals n\beta_\ell\quad(0\le\ell\le J).
\]
The exact solution is
\begin{equation}\label{atom:eq:recursion}
 \beta_\ell=R^{-(2^\ell+3)},\qquad b_J<1,
\end{equation}
because $2^J\ge R$ and $n\le e^R$.

First suppose $T\ge m$. Take the deterministic suffix beginning at
$s_0\colonequals T-m$ and set $s_\ell\colonequals s_0+\ell w$, so $s_J=T$.
We will establish, for $\ell=0,\ldots,J$, the following induction bound
outside the bad events estimated below:
\begin{equation}\label{atom:eq:induction}
 N_{j_0+\ell}(t)\le b_\ell
 \quad\text{for all integers }t\in[s_\ell,T].
\end{equation}
The potential bound and the choice of $j_0$ give the base case.
Each interval of length $w$ allows initially high coordinates to visit
below the next level, while stopped entry counts bound subsequent
occupancy. Since $b_J<1$, the final level is empty at time $T$.

Let $\mathcal G_M$ be the event $M_t\le B$ for every integer
$t\in[s_0,T]$. Lemma~\ref{atom:lem:maxM} gives
\begin{equation}\label{atom:eq:badM}
 \P(\mathcal G_M^c)
 \le\frac{16n+m/4}{nR^4}
 \le\frac{16}{R^4}+\frac{257J}{R^3}.
\end{equation}
On $\mathcal G_M$, for the entire suffix,
\begin{equation}\label{atom:eq:baselevel}
 N_{j_0}(t)\le M_te^{-\alpha j_0}
 \le nR^4R^{-8}=b_0.
\end{equation}

For each $\ell=0,\ldots,J-1$, let $\mathcal G_{\mathrm{clear},\ell}$
be the event that every coordinate initially high at level
$j_0+\ell+1$ at time $s_\ell$ visits below that level by time
$s_{\ell+1}$. Applying \eqref{atom:eq:clear} conditional on $\Fc_{s_\ell}$
and then a union bound gives
\begin{equation}\label{atom:eq:badclear}
 \P\!\left(\mathcal G_M\cap
       \bigcup_{\ell=0}^{J-1}\mathcal G_{\mathrm{clear},\ell}^c\right)
 \le JnB e^{-w/(64n)}
 \le JR^4e^{-14R}.
\end{equation}
Here $\mathcal G_M\subseteq\{M_{s_\ell}\le B\}$ and we bound the
latter measurable event before integrating; we never condition on the
future event $\mathcal G_M$ when applying the killed drift.

For each $\ell$, define the stopping time
\[
 \sigma_\ell\colonequals \inf\{t\ge s_\ell:N_{j_0+\ell}(t)>b_\ell\}
\]
and the stopped entry count
\[
 Z_\ell\colonequals \sum_{t=s_\ell}^{T-1}
       \mathbf1_{\{t<\sigma_\ell\}}C_{j_0+\ell+1}(t).
\]
The indicator is predictable. Equation~\eqref{atom:eq:entries} therefore
implies, without conditioning on any future good event,
\[
 \E Z_\ell\le\frac{m b_\ell^2}{2n^2}.
\]
Markov's inequality at the possibly noninteger threshold $b_{\ell+1}$
gives
\begin{equation}\label{atom:eq:badentry}
 \P(Z_\ell>b_{\ell+1})
 \le\frac{m}{2nR^3}
 \le\frac{513J}{R^2}.
\end{equation}
Consequently the union of all these $J$ failures has probability at most
$513J^2/R^2$. No independence between the $Z_\ell$ is required.

On the complement of all displayed bad events, we prove
\eqref{atom:eq:induction} by induction on $\ell$.
The base is \eqref{atom:eq:baselevel}. If it holds at stage $\ell$, then
$\sigma_\ell>T$, so $Z_\ell$ equals the actual total entries on the
remaining suffix. The clearance event and \eqref{atom:eq:entrytracking} imply
for every $t\in[s_{\ell+1},T]$ that
\[
 N_{j_0+\ell+1}(t)\le Z_\ell\le b_{\ell+1}.
\]
This proves the induction. At its last stage,
$N_{j_0+J}(T)\le b_J<1$, and integrality gives $N_{j_0+J}(T)=0$.

\subsection{Short horizons and the terminal bound}
If $0\le T<m$, use the interval $[0,T]$ for the potential maximum and
put every $s_\ell=0$ in the stopped entry construction. The potential
bound \eqref{atom:eq:badM} still holds, since its length is at most $m$.
The base level and stopped-count expectation bounds are unchanged.
Every coordinate starts below every level in question. Thus
\eqref{atom:eq:entrytracking} holds on the full interval without clearance,
and the same induction gives
\[
 N_{j_0+\ell}(t)\le b_\ell\quad(0\le t\le T)
\]
for each stage. In particular $N_{j_0+J}(T)=0$. This is only an
analysis of short prefixes with a fixed length bound $m(n)$; the rule
does not know or depend on $T$.

For long horizons, the total failure is bounded by
\[
 \frac{16}{R^4}+\frac{257J}{R^3}
       +JR^4e^{-14R}+\frac{513J^2}{R^2}<\frac1{100},
 \qquad R\ge10^6,
\]
using $J\le\sqrt R$. For short horizons the clearance term is absent.
We have therefore proved, for all prescribed $T$,
\begin{equation}\label{atom:eq:explicitheight}
 \P\{\norm{d_T}_\infty\ge j_0+J\}<\frac1{100}.
\end{equation}
Finally
\begin{align*}
 j_0+J&\le2+66\log R,\\
 \log R&\le\log(10^6+\log2)+\log(1+\log n),
\end{align*}
since $\log(n+1)\le\log2+\log n$. These inequalities give
\[
 j_0+J\le1000[1+\log\log(en)].
\]
Multiplication by $\sqrt n$ proves \cref{prop:RareActivationUpper}.

\section*{Acknowledgments}
The author is grateful to Timm Oertel for helpful comments on the first version and to the anonymous reviewers of an earlier conference submission for their careful reading and constructive feedback, which helped improve the presentation.

The author used OpenAI's ChatGPT during the preparation of the first version to assist with proofreading, exposition, and feedback on selected mathematical arguments. For this revision, the author used OpenAI's Codex to assist with literature research, manuscript organization, English exposition, and mathematical exploration, including improvements and extensions of the results and the development, simplification, and checking of proofs. The author retains full responsibility for the content of the manuscript, including all mathematical claims and references.

\bibliographystyle{alpha}
\bibliography{References}
\end{document}